\documentclass[EPiCempty]{easychair}

\usepackage[english]{babel}
\usepackage[T1]{fontenc}
\usepackage[utf8]{inputenc}

\usepackage{stmaryrd}

\usepackage{enumitem}
\setlist[itemize]{leftmargin=2.5ex, topsep=0.5ex}
\setlist[enumerate]{leftmargin=3ex, itemsep=0ex}

\newtheorem{theorem}{Theorem}
\newtheorem{observation}[theorem]{Observation}
\newtheorem{definition}[theorem]{Definition}
\newtheorem{lemma}[theorem]{Lemma}
\newtheorem*{lemma-rep}{Lemma 11}
\newtheorem*{corollary-rep}{Corollary 16}
\newtheorem{proposition}[theorem]{Proposition}

\newtheorem{corollary}[theorem]{Corollary}
\newtheorem{example}[theorem]{Example}
\newtheorem{algor}[theorem]{Algorithm}

\usepackage{xspace}
\usepackage{amsmath}
\usepackage{amssymb}
\usepackage{cleveref}

\usepackage[plain]{algorithm}
\usepackage{algpseudocode}

\algnewcommand\algorithmicinput{\textbf{Input:}}
\algnewcommand\algorithmicoutput{\textbf{Output:}}
\algrenewcommand\algorithmicprocedure{\textbf{Procedure:}}
\algnewcommand\algorithmicbreak{\textbf{break}}
\algrenewcommand\textproc[1]{#1}
\algnewcommand\Input{\item[\algorithmicinput]}
\algnewcommand\Output{\item[\algorithmicoutput]}
\algrenewcommand\Procedure[2]{\item[\algorithmicprocedure]\textproc{#1}\ifthenelse{\equal{#2}{}}{}{(#2)}}

\DeclareMathOperator{\Succ}{succ}

\newcommand{\ALC}[0]{\ensuremath{\mathcal{A}\mathcal{L}\mathcal{C}}}
\newcommand{\ALCQ}[0]{\ensuremath{\ALC\mathcal{Q}}}
\newcommand{\ALCO}[0]{\ensuremath{\ALC\mathcal{O}}}
\newcommand{\ALCHQ}[0]{\ensuremath{\ALC\mathcal{HQ}}}
\newcommand{\ALCSCC}[0]{\ensuremath{\ALC\mathcal{SCC}}}
\newcommand{\ALCISCC}[0]{\ensuremath{\ALC\mathcal{ISCC}}}

\newcommand{\ALCplus}[0]{\ensuremath{\ALCSCC^{++}}}
\newcommand{\ALCIplus}[0]{\ensuremath{\ALCISCC^{++}}}
\newcommand{\suc}[0]{\ensuremath{\mathit{succ}}}
\newcommand{\constr}[0]{\ensuremath{\mathit{sat}}}

\newcommand{\rc}[0]{\ensuremath{\overline{r}}}
\newcommand{\I}[0]{\ensuremath{\mathcal{I}}}
\newcommand{\C}[0]{\ensuremath{\mathcal{C}}}
\newcommand{\E}[0]{\ensuremath{\mathcal{E}}}
\newcommand{\M}[0]{\ensuremath{\mathcal{M}}}
\newcommand{\T}[0]{\ensuremath{\mathcal{T}}}

\newcommand{\A}[0]{\ensuremath{\mathcal{A}}}
\newcommand{\At}[0]{\ensuremath{\mathbb{A}}}

\newcommand{\R}[0]{\ensuremath{\mathcal{R}}}
\newcommand{\universe}[0]{\ensuremath{\mathcal{U}}}

\newcommand{\atleastq}[3]{\mathopen{\geqslant}\, #1\, #2 . #3}
\newcommand{\atmostq}[3]{\mathopen{\leqslant}\, #1\, #2 . #3}

\newcommand{\atmostc}[2]{(\mathopen{\leqslant}\, #1\, #2)}

\newcommand{\Div}[0]{\,\mathsf{dvd}\,}
\DeclareMathOperator{\Var}{Var}
\newcommand{\Ind}[0]{\mathop{\mathsf{Ind}}}
\newcommand{\ars}[0]{\ensuremath{\mathit{ars}^\I}}
\newcommand{\con}[0]{\ensuremath{\mathit{Con}}}
\newcommand{\subs}[2]{\sigma^{#1}_{#2}}
\newcommand{\tubs}[2]{\tau^{#1}_{#2}}

\newcommand{\complclass}[1]{{\sc #1}\xspace} % font for complexity classes

\newcommand{\NP}{\complclass{NP}}

\newcommand{\PSpace}{\complclass{PSpace}}

\newcommand{\ExpTime}{\complclass{ExpTime}}
\newcommand{\NExpTime}{\complclass{NExpTime}}

\newcommand{\TwoExpTime}{\complclass{2ExpTime}}

\newcommand{\state}{\mathfrak{q}}
\newcommand{\query}{{q}}
\newcommand{\Vlang}{{V}}
\newcommand{\kb}{\ensuremath{\mathcal{K}}}
\newcommand{\tuplei}[1]{(#1)}

\def\hmath$#1${\texorpdfstring{{\rmfamily\textit{#1}}}{#1}}

\usepackage{todonotes}

\newcommand{\abox}{\mathcal{A}}
\newcommand{\tbox}{\mathcal{T}}
\newcommand{\ercbox}{\mathcal{R}}

\newcommand{\DI}{\Delta^{\I}}
\newcommand{\IndA}{\Ind_\A}

\newcommand{\DInamed}{\DI_{\mathrm{named}}}

\newcommand{\Iunrav}{\I^{\to}}
\newcommand{\cdotIunrav}{\cdot^{\Iunrav}}
\newcommand{\DIunrav}{\Delta^{\Iunrav}}

\newcommand{\DIunravnamed}{\DIunrav_{\mathrm{named}}}

\newcommand{\J}{\mathcal{J}}
\newcommand{\Ji}[1]{ {\J}_{#1} }
\newcommand{\cdotJi}[1]{ {\cdot}^{\Ji{#1}} }
\newcommand{\DJi}[1]{ {\Delta}^{\Ji{#1}} }
\newcommand{\DJinamed}[1]{ {\Delta_{\mathrm{named}}}^{\Ji{#1}} }

\renewcommand{\DJ}{\Delta^{\J}}
\newcommand{\DJnamed}{\DJ_{\mathrm{named}}}

\newcommand{\indA}{\Ind_\A}
\newcommand{\last}[1]{\mathsf{last}(#1)}
\newcommand{\first}[1]{\mathsf{first}(#1)}
\newcommand{\Neib}[2]{\mathsf{N}_{#1}(#2)}
\newcommand{\izo}{\mathfrak{f}}
\newcommand{\homo}{\mathfrak{h}}

\newcommand{\blocked}[2]{\mathsf{Bl}^{[#1]}_{#2}}

\newcommand{\Iloose}[1]{\I^{[#1]}}
\newcommand{\DIloose}[1]{\Delta^{\Iloose{k}}}
\newcommand{\DIloosenamed}[1]{\DIloose{k}_{\mathrm{named}}}

\newcommand{\Nat}{\mathbb{N}}
\newcommand{\ALCQt}[0]{\ensuremath{\ALC\mathcal{Q}t}}

\newcommand{\eqbisim}{\equiv_{\mathsf{fb}}}

\newcommand{\SuccI}[2]{\mathsf{Succ}_{#1}(#2)}

\newcommand{\ALCH}[0]{\ensuremath{\ALC\mathcal{H}}}
\newcommand{\ALCHcap}{\ALCH^{\cap}}

\newcommand{\suff}[2]{\mathsf{suff}_{#1}(#2)}
\newcommand{\maxfr}[1]{{#1}_{\mathsf{mfr}}} 
\newcommand{\forkrev}[1]{{#1}_{\mathsf{fr}}}

\newcommand\restr[2]{{% we make the whole thing an ordinary symbol
  \left.\kern-\nulldelimiterspace % automatically resize the bar with \right
  #1 % the function
  \vphantom{\big|} % pretend it's a little taller at normal size
  \right|_{#2} % this is the delimiter
 }}

\begin{document}

\title{Satisfiability and Query Answering in Description Logics with Global and Local Cardinality Constraints}

\author{Franz Baader\inst{1}, Bartosz Bednarczyk\inst{1}\inst{2}, and Sebastian Rudolph\inst{1}}

\institute{Faculty of Computer Science, TU Dresden, Germany\\
           \email{firstname.lastname@tu-dresden.de}
\and
Institute of Computer Science, University of Wroc\l aw, Poland\\
          	\email{bartosz.bednarczyk@cs.uni.wroc.pl} }

\authorrunning{Baader, Bednarczyk, and Rudolph}
	
\titlerunning{DLs with Global and Local Cardinality Constraints}

\maketitle

\begin{abstract}
We introduce and investigate the expressive description logic (DL) $\ALCplus$, in which the global and local cardinality constraints introduced in previous papers can
be mixed. On the one hand, we prove that this does not increase the complexity of satisfiability checking and other standard inference problems. 
On the other hand, the satisfiability problem becomes undecidable if inverse roles are added to the languages. In addition, even without
inverse roles, conjunctive query entailment in this DL turns out to be undecidable. We prove that decidability of querying can be regained
if global and local constraints are not mixed and the global constraints are appropriately restricted. The latter result is based on a 
locally-acyclic model construction, and it reduces query entailment to ABox consistency in the restricted setting, i.e., to ABox consistency
w.r.t.\ restricted cardinality constraints in $\ALCSCC$, for which we can show an ExpTime upper bound.
\end{abstract}

\section{Introduction}
\label{sect:introduction}

Description Logics (DLs) \cite{BCNMP03} are a well-investigated family of logic-based knowledge representation languages,
which are frequently used to formalize ontologies for application domains such as biology and medicine \cite{HoSG15}.
To define the important notions of such an application domain as formal concepts, DLs state necessary and
sufficient conditions for an individual to belong to a concept.
These conditions can be Boolean combinations of atomic properties required for the
individual (expressed by concept names) or properties that refer
to relationships with other individuals and their properties
(expressed as role restrictions). Using an example from \cite{BaEc17}, the concept of a motor vehicle can be formalized by the concept description
$$
\textit{Vehicle} \sqcap \exists\textit{part}.\textit{Motor},
$$
which uses the concept names \textit{Vehicle} and \textit{Motor} and the role name \textit{part} as well as
the concept constructors conjunction ($\sqcap$) and existential restriction ($\exists r.C$).
The concept inclusion (CI)
$$
\textit{Motor-vehicle} \sqsubseteq \textit{Vehicle} \sqcap \exists\textit{part}.\textit{Motor}
$$
then states that every motor vehicle needs to belong to this concept description.
Numerical constraints on the number of role successors (so-called number restrictions) have been
used early on in DLs \cite{BBMA89,HoNS90,HoBa91b}. For example, using number restrictions, motorcycles
can be constrained to being motor vehicles with exactly two wheels:
$$
\begin{array}{r@{\ }c@{\ }l}
\textit{Motorcycle} &\sqsubseteq& \textit{Motor-vehicle}\ \sqcap\mbox{} %\\ & & 
(\atmostq{2}{\textit{part}}{\textit{Wheel}}) \sqcap (\atleastq{2}{\textit{part}}{\textit{Wheel}}).
\end{array}
$$
The exact complexity of reasoning in $\ALCQ$, the DL that has all Boolean operations and number
restrictions of the form $(\atmostq{n}{r}{C})$ and $(\atleastq{n}{r}{C})$ as concept constructors,
was determined by Stephan Tobies \cite{Tobi99,Tobi01a}: it is \PSpace-complete without CIs and \ExpTime-complete w.r.t.\ CIs,
independently of whether the numbers occurring in the number restrictions are encoded in unary or binary.
Note that, using unary coding of numbers, the number $n$ is assumed to contribute $n$ %rather than $\log n$
to the size of the input, whereas with binary coding the size of the number $n$ is $\log n$. Thus, for large numbers,
using binary coding is more realistic.

Whereas number restrictions are local in the sense that they consider role successors of an individual
under consideration (e.g.\ the wheels that are part of a particular motor vehicle), cardinality restrictions on concepts (CRs) \cite{BaBH96,Tobi00}
are global, i.e., they consider all individuals in an interpretation. For example, the cardinality restriction
$$
\atmostc{45000000}{(\textit{Car}\sqcap\exists\textit{registered-in}.\textit{German-district})}
$$
states that at most 45 million cars are registered all over Germany. Such cardinality restrictions can be seen as quantitative extensions
of CIs since a CI of the form $C\sqsubseteq D$ can be expressed by the CR $\atmostc{0}{(C\sqcap\neg D)}$.
The availability of CRs
increases the complexity of reasoning: as mentioned above, consistency in $\ALCQ$ w.r.t.\ CIs is \ExpTime-complete,
but consistency w.r.t.\ CRs is \NExpTime-complete 
if the numbers occurring in the CRs are assumed to be encoded in binary \cite{Tobi00}.
With unary coding of numbers, consistency stays \ExpTime-complete even w.r.t.\ CRs \cite{Tobi00}. However, as
the above example considering 45 million cars indicates, unary coding does not yield a realistic measure for the
input size  if numbers with large values are employed.

In two previous publications we have, on the one hand, extended the DL
$\ALCQ$ by more expressive number restrictions using cardinality and set constraints expressed
in the quantifier-free fragment of Boolean Algebra with Presburger Arithmetic (QFBAPA) \cite{KuRi07}.
In the resulting DL $\ALCSCC$, which was introduced and investigated in \cite{Baad17}, 
cardinality and set constraints are applied locally, i.e., they refer to the role successors of an individual under consideration.
For example, we can state that the number of cylinders of a motor must coincide with the number of spark plugs in this motor, without fixing
what this number actually is, using the following $\ALCSCC$ CI:
$$
\mathit{Motor}\sqsubseteq \suc(|\mathit{part}\cap\mathit{Cylinder}| = |\mathit{part}\cap\mathit{SparkPlug}|).
$$
It was shown in \cite{Baad17} that pure concept satisfiability in $\ALCSCC$ is a \PSpace-complete problem, and concept satisfiability
w.r.t.\ a general TBox is \ExpTime-complete. This shows that the more expressive number restrictions do not increase
the complexity of reasoning since reasoning in $\ALCQ$ has the same complexity, as mentioned above.

On the other hand, we have extended the terminological formalism of the well-known description logic $\ALC$\footnote{%
The DL $\ALC$ is the fragment of $\ALCQ$ in which only number restrictions of the form
$(\atmostq{0}{r}{\neg C})$ (written $\forall r.C$) and $(\atleastq{1}{r}{C})$ (written $\exists r.C$) are available.
}
from CIs not only to CRs, but to more general cardinality constraints expressed in QFBAPA \cite{BaEc17}, which we called
extended cardinality constraints (ECBoxes). These constraints are global since they refer to all individuals in the interpretation domain.
An example of a constraint expressible this way, but not expressible using CRs is
$$
\begin{array}{l}
2\cdot |\textit{Car}\sqcap\exists\textit{registered-in}.\textit{German-district}\sqcap\exists\textit{fuel}.\textit{Diesel}| \\
\leq
|\textit{Car}\sqcap\exists\textit{registered-in}.\textit{German-district}\sqcap\exists\textit{fuel}.\textit{Petrol}|,
\end{array}
$$
which states that,
in Germany, cars running on petrol outnumber cars running on diesel by a factor of at least two.
It was shown in \cite{BaEc17} that reasoning w.r.t.\ ECBoxes is still in \NExpTime even if the numbers occurring in the constraints are encoded
in binary. The \NExpTime lower bound follows from the result of Tobies~\cite{Tobi00} CRs mentioned above.
This complexity can be lowered to \ExpTime if a restricted form of
cardinality constraints (RCBoxes) is used. Such RCBoxes are still powerful enough to express statistical knowledge bases \cite{PePo17}.

An obvious way to generalize these two approaches is to combine the two extensions, i.e., to consider extended cardinality constraints,
but now on $\ALCSCC$ concepts rather than just $\ALC$ concepts. This combination was investigated in~\cite{Baad19,Baad19b}, where a \NExpTime upper 
bound was established for reasoning in $\ALCSCC$ w.r.t.\ ECBoxes. It is also shown in \cite{Baad19,Baad19b} that reasoning w.r.t.\ RCBoxes
stays in \ExpTime also for $\ALCSCC$.

Here we go one step further by allowing for a tighter integration of global and local constraints. 
The resulting logic, which we call $\ALCplus$, allows, for example, to relate the number of role successors of a given individual with the overall
number of elements of a certain concept. 
For example, the $\ALCplus$ concept description\footnote{%
To distinguish between constraint expressions in $\ALCSCC$ and in $\ALCplus$, which have a different semantics, we use different keywords for them.
}
$$
\constr(|\mathit{likes}\cap \mathit{Car}| = |\mathit{Car}|)
$$
describes car lovers, i.e., individuals that like \emph{all} cars, independently of whether these cars are related to them by some role or not.
More generally, DLs that can express both local cardinality constraints (i.e., constraints concerning the role successors of  specific individuals)
and global cardinality constraints (i.e., constraints on the overall cardinality of concepts) can,
for instance, be used to check the correctness of statistical statements. For example,
if a German car company claims that they have produced more than $N$ cars in a certain year, and $P$\%\ of the tires used
for their cars were produced by Betteryear, this may be contradictory to a statement of Betteryear that they have
sold less than $M$ tires in Germany. Such statistical information may, of course, also influence the answers to queries.
If we know that the car company VMW uses only tires from Betteryear or Badmonth, but the statistical information
allows us to conclude that another car company has actually bought all the tires sold by Betteryear, then we know that the cars
sold by VMW all have tires produced by Badmonth. This motivates investigating DLs with expressive cardinality constraints,
and to consider not just standard inferences such as satisfiability checking for these DLs, but also query answering.

In the present paper, we show that, from a worst-case complexity point of view, the extended expressivity of $\ALCplus$ 
comes for free if we consider classical reasoning problems. Concept satisfiability in $\ALCplus$ has the same complexity as in 
$\ALC$ and $\ALCSCC$ with global cardinality constraints: it is \NExpTime-complete. However, if we add inverse roles, then
concept satisfiability becomes undecidable. In addition, for effective conjunctive query answering this
logic turns out to be too expressive. We show that conjunctive query entailment w.r.t.\ $\ALCplus$ knowledge bases is, in fact, undecidable. 
In contrast, we can show that conjunctive query entailment w.r.t.\ (an extension of) $\ALCSCC$ RCBoxes is decidable and, in fact, only \ExpTime-complete. 
To proof this result, we first show that standard ABox reasoning in this setting is \ExpTime-complete. Then, we
reduce query entailment over arbitrary structures to query entailment over locally acyclic graphs,
based on an appropriate model construction, which proceeds in three steps. 
Once this is achieved, the \ExpTime upper bound for conjunctive query entailment is shown by a reduction to ABox reasoning,
adapting the approach used by Lutz in~\cite{Lutz08} for $\ALCHQ$.

We assume the reader to be sufficiently familiar with all the standard notions of description logics \cite{BCNMP03,baader_horrocks_lutz_sattler_2017,rudolph2011fodl}.

\section{The logic~\hmath$\ALCplus$}
\label{prelim:sect}

As in \cite{Baad17,BaEc17,Baad19,Baad19b}, we use the quantifier-free fragment of Boolean Algebra with Presburger Arithmetic (QFBAPA) \cite{KuRi07}
to express our constraints.  We start with a brief introduction of QFBAPA (see \cite{KuRi07} and \cite{Baad17} for more details).

In the logic QFBAPA, one can build \emph{set terms} by applying Boolean operations
(intersection $\cap$, union $\cup$, and complement $\cdot^c$)
to set variables as well as the constants $\emptyset$ and $\universe$. Set terms $s,t$ can then be used to state
\emph{set constraints}, which are equality and inclusion constraints of the form $s = t, s\subseteq t$,
where $s,t$ are set terms. \emph{Presburger Arithmetic (PA) expressions}
are built from integer constants and set cardinalities $|s|$ using addition as well as multiplication with
an integer constant.\footnote{%
The definition of QFBAPA in \cite{KuRi07} also allows for integer variables, which we do not use when integrating QFBAPA into our DL.
}
They can be used to form \emph{cardinality constraints}  of the form $k = \ell, k < \ell, N \Div \ell$, where
$k, \ell$ are PA expressions, $N$ is an integer constant, and $\Div$ stands for divisibility. A \emph{QFBAPA formula}
is a Boolean combination of set and cardinality constraints using the Boolean operations $\wedge, \vee, \neg$.
  
A \emph{substitution} $\sigma$ assigns a finite set $\sigma(\universe)$ to $\universe$, the empty set to $\emptyset$, and subsets of $\sigma(\universe)$ to set variables.
It is extended to set terms by interpreting the Boolean operations $\cap$, $\cup$, and $\cdot^c$ as set intersection, set union,
and set complement w.r.t.\ $\sigma(\universe)$, respectively. The substitution $\sigma$ satisfies the set constraint $s = t$ ($s\subseteq t$) if
$\sigma(s) = \sigma(t)$ ($\sigma(s) \subseteq \sigma(t)$). It is further extended to a mapping from PA expressions to integers by interpreting $|s|$ as the cardinality
of the finite set $\sigma(s)$, and addition and multiplication with an integer constant in the usual way. The substitution $\sigma$ satisfies the
cardinality constraint $k = \ell$ if $\sigma(k) = \sigma(\ell)$, $k < \ell$ if $\sigma(k) < \sigma(\ell)$, and $N \Div \ell$ if the integer constant $N$ is a divisor of $\sigma(\ell)$.
The notion of satisfaction of a Boolean combination of set and cardinality constraints is now defined in the obvious way by interpreting $\wedge, \vee, \neg$ as in propositional logic.
The substitution $\sigma$ is a \emph{solution} of the QFBAPA formula $\phi$ if it satisfies $\phi$ in this sense.
A QFBAPA formula $\phi$ is \emph{satisfiable} if it has a solution. In \cite{KuRi07} it is shown that the satisfiability problem 
for QFBAPA formulae is \NP-complete.

We are now ready to define our new logic, which we call $\ALCplus$ to indicate that it is an extension of the logic
$\ALCSCC$ introduced in \cite{Baad17}. When defining the semantics of \ALCplus,
we restrict the attention to \emph{finite} interpretations to ensure that cardinalities of concept descriptions
are always well-defined non-negative integers.

\begin{definition}[$\ALCplus$]
Given disjoint \emph{finite} sets $N_C$ and $N_R$ of \emph{concept names} and \emph{role names}, respectively, 
$\ALCplus$ \emph{concept descriptions} (short: \emph{concepts}) are inductively defined as follows:
\begin{itemize}
\item
   Every concept name $A\in N_C$ is an $\ALCplus$ concept.
\item
   If $C, D$ are $\ALCplus$ concepts, then so are $C\sqcap D$ (\emph{conjunction}), $C\sqcup D$ (\emph{disjunction}), and 
   $\neg C$ (\emph{negation}).
\item
   If $\con$ is a set constraint or a cardinality constraint that uses role names and 
   already defined $\ALCplus$ concepts in place of set variables, then
   $\constr(\con)$ is an $\ALCplus$ concept. We call $\constr(\con)$ a \emph{constraint expression}.
\end{itemize}
As usual, we use $\top$ (\emph{top}) and $\bot$ (\emph{bottom}) as abbreviations for $A\sqcup \neg A$ and $A\sqcap \neg A$,
respectively, where $A$ is an arbitrary concept name.

A \emph{finite interpretation} of $N_C$ and $N_R$ consists of a finite, non-empty set $\Delta^\I$ and a mapping $\cdot^\I$ that
maps every concept name $A\in N_C$ to a subset $A^\I$ of $\Delta^\I$ and every role name $r\in N_R$ to a binary relation $r^\I$ over $\Delta^\I$.
For a given element $d\in \Delta^\I$ we define
   $$
     r^\I(d) := \{e\in \Delta^\I \mid (d,e) \in r^\I\}.
   $$
The substitution $\subs{\I}{d}$ assigns 
the finite set $\Delta^\I$ to $\universe$, the empty set to $\emptyset$, and
the sets $r^\I(d)$ to $r$ and $A^\I$ to $A$, where $r\in N_R$ and $A\in N_C$ are viewed as set variables.
 
The interpretation function $\cdot^\I$ and the substitutions $\subs{\I}{d}$ for $d\in \Delta^\I$ are
inductively extended to $\ALCplus$ concepts by interpreting
the Boolean operators as usual: 
\begin{itemize}
\item
  $\subs{\I}{d}(C\sqcap D) = (C\sqcap D)^\I = C^\I\cap D^\I$,
\item
  $\subs{\I}{d}(C\sqcup D) = (C\sqcup D)^\I = C^\I\cup D^\I$,
\item
  $\subs{\I}{d}(\neg C) = (\neg C)^\I = \Delta^\I \setminus C^\I$,
\end{itemize}
and the constraint expressions as follows:
\begin{itemize}
\item
  $\subs{\I}{d}(\constr(\con)) = \constr(\con)^\I = \{d \in \Delta^\I \mid \mbox{the substitution $\subs{\I}{d}$ satisfies $\con$}\}$.\footnote{%
Note that, by induction, we can assume that $\subs{\I}{d}$ is defined on the set variables (i.e., role names and concepts)
occurring in $\con$.
}
\end{itemize}

The $\ALCplus$ concept description $C$ is \emph{satisfiable} if there is a finite interpretation
$\I$ such that $C^\I \neq \emptyset$.
\end{definition} 
  
Note that the interpretation of concepts as set variables in $\ALCplus$ is global in the sense that it does not depend on $d$,
i.e., $\subs{\I}{d}(C) = C^\I = \subs{\I}{e}(C)$ for all $d, e\in \Delta^\I$. In contrast, the interpretation of role names $r$ as set variables
is local since only the $r$-successors of $d$ are considered by $\subs{\I}{d}(r)$.
In $\ALCSCC$, also the interpretation of concepts as set variables is local since in the semantics of
$\ALCSCC$ the substitution $\subs{\I}{d}$ considers only the elements of $C^\I$ that are role successors of $d$ for
some role name in $N_R$ (see \cite{Baad17}). To reflect this difference in the semantics also on the syntactic level,
we use the keyword $\suc$ (for \emph{succ}sessor) in place of $\constr$ for constraint expressions in $\ALCSCC$,
and call these expressions \emph{successor expressions}. 
For the sake of completeness, we now give a detailed definition of the DL $\ALCSCC$ as well as of $\ALCSCC$ TBoxes, ABoxes, and ECBoxes
(see also \cite{Baad17,BaEc17} and \cite{Baad19,Baad19b}).

\begin{definition}[$\ALCSCC$]\label{ALCSCC:def}
Given disjoint \emph{finite} sets $N_C$ and $N_R$ of \emph{concept names} and \emph{role names}, respectively, 
$\ALCSCC$ \emph{concept descriptions} (short: \emph{concepts}) are inductively defined as follows:
\begin{itemize}
\item
   Every concept name $A\in N_C$ is an $\ALCSCC$ concept.
\item
   If $C, D$ are $\ALCSCC$ concepts, then so are $C\sqcap D$ (\emph{conjunction}), $C\sqcup D$ (\emph{disjunction}), and 
   $\neg C$ (\emph{negation}).
\item
   If $\con$ is a set constraint or a cardinality constraint that uses role names and 
   already defined $\ALCSCC$ concepts in place of set variables, then
   $\suc(\con)$ is an $\ALCSCC$ concept. We call $\suc(\con)$ a \emph{successor expression}.
\end{itemize}

An $\ALCSCC$ \emph{concept inclusion (CI)} is of the form $C\sqsubseteq D$ where $C, D$ are $\ALCSCC$ concepts,
and an $\ALCSCC$ \emph{TBox} is a finite set of $\ALCSCC$ CIs. 
An $\ALCSCC$ \emph{ABox} is a finite set of \emph{concept assertions} $C(a)$ and \emph{role assertions} $r(a,b)$ where
$C$ is an $\ALCSCC$ concept, $r$ is a role name, and $a,b$ are individual names from a set $N_I$ of such names, which is disjoint with
$N_C$ and $N_R$.
We define \emph{extended cardinality constraints on $\ALCSCC$ concepts} as follows:
\begin{itemize}
\item
  \emph{$\ALCSCC$ cardinality terms} are built from integer constants and concept cardinalities $|C|$ for $\ALCSCC$ concepts $C$
  using addition and multiplication with integer constants;
\item
  \emph{extended $\ALCSCC$ cardinality constraints} are of the form $k = \ell, k < \ell, N \Div \ell$, where $k, \ell$ are $\ALCSCC$ cardinality terms and
  $N$ is an integer constant;
\item
  an \emph{extended $\ALCSCC$ cardinality box (ECBox)} is a Boolean combination of extended $\ALCSCC$ cardinality constraints.
\end{itemize}

A \emph{finite interpretation} of $N_C$ and $N_R$ consists of a finite, non-empty set $\Delta^\I$ and a mapping $\cdot^\I$ that
maps every concept name $A\in N_C$ to a subset $A^\I$ of $\Delta^\I$, every role name $r\in N_R$ to a binary relation $r^\I$ over $\Delta^\I$,
and ever individual name $a\in N_I$ to an element $a^\I$ of $\Delta^\I$.
For a given element $d\in \Delta^\I$ we define
   $$
     r^\I(d) := \{e\in \Delta^\I \mid (d,e) \in r^\I\}\ \ \mbox{and}\ \ 
     \ars(d) := \bigcup_{r\in N_R} r^\I(d).\footnote{%
$\mathit{ars}$ stands for ``all role successors.''}
   $$
The substitution $\tubs{\I}{d}$ assigns 
the finite set $\ars(d)$ to $\universe$, the empty set to $\emptyset$, and
the sets $r^\I(d)$ to $r$ and $A^\I\cap\ars(d)$ to $A$, where $r\in N_R$ and $A\in N_C$ are viewed as set variables.
 
The interpretation function $\cdot^\I$ and the substitutions $\tubs{\I}{d}$ for $d\in \Delta^\I$ are
inductively extended to $\ALCSCC$ concepts by interpreting
the Boolean operators as usual: 
\begin{itemize}
\item
  $(C\sqcap D)^\I = C^\I\cap D^\I$ and $\tubs{\I}{d}(C\sqcap D) = C^\I\cap D^\I\cap\ars(d)$.
\item
  $(C\sqcup D)^\I = C^\I\cup D^\I$ and $\tubs{\I}{d}(C\sqcup D) = (C^\I\cup D^\I)\cap\ars(d)$.
\item
   $(\neg C)^\I = \Delta^\I \setminus C^\I$ and $\tubs{\I}{d}(\neg C) = (\Delta^\I \setminus C^\I)\cap\ars(d)$.
\end{itemize}
and the successor expressions as follows:
\begin{itemize}
\item
   $\suc(\con)^\I = \{d \in \Delta^\I \mid \mbox{the substitution $\tubs{\I}{d}$ satisfies $\con$}\}$ and
\item
   $\tubs{\I}{d}(\suc(\con)) = \suc(\con)^\I \cap\ars(d)$. 
\end{itemize}

The finite interpretation $\I$ is a \emph{model} of the $\ALCSCC$ TBox $\T$ if it satisfies all the CIs
$C\sqsubseteq D$ in \T, which is the case if $C^\I\subseteq D^\I$ holds. It is a \emph{model} of the
$\ALCSCC$ ABox $\A$ if it satisfies all the assertions in $\A$, where $\I$  satisfies the concept assertion
$C(a)$ if $a^\I\in C^\I$ holds, and the role assertion $r(a,b)$ if $(a^\I,b^\I)\in r^\I$ holds.
Concept cardinalities within an ECBox $\E$ are interpreted  in the obvious
way, i.e., $|C|^\I := |C^\I|$. Cardinality terms and cardinality constraints as well as their Boolean combination are then
interpreted as in QFBAPA. 
The finite interpretation $\I$ is a \emph{model} of an ECBox $\E$ if it satisfies the Boolean formula $\E$ according to this semantics.

The $\ALCSCC$ concept description $C$ is \emph{satisfiable} w.r.t.\ the ECBox $\E$ if there is a model 
$\I$ of $\E$ such that $C^\I \neq \emptyset$. The ABox $\A$ is \emph{consistent} w.r.t.\ $\E$ 
if there is a model $\I$ of $\E$ that is also a model of $\A$.
\end{definition} 

The following examples illustrates the difference between the semantics of constraint expressions in $\ALCplus$ and
successor expressions in $\ALCSCC$.

\begin{example}\label{alcplus:ex}
\rm
If $A$ is a concept name and $r$ is a role name, then the following is an $\ALCplus$ concept description:
$$
E := \constr(|A| \geq 4) \sqcap \constr(A\subseteq r) \sqcap \constr(|r| \leq 3).
$$
The first constraint expression requires that the overall size of the concept $A$ is at least four. Thus, if $\I$ is an interpretation
with $|A^\I| \leq 3$, then no element of $\Delta^\I$ can belong to $\constr(|A| \geq 4)^\I$. Otherwise, every element of $\Delta^\I$
belongs to $\constr(|A| \geq 4)^\I$. The second constraint says that every element of $A$ must be an $r$ successor of the given
individual. Thus, $\constr(A\subseteq r)^\I$ consists of those elements of $\Delta^\I$ that are connected, via the role $r$, with every
element of $A^\I$. The third constraint is satisfied by those element of $\Delta^\I$ that have at most three $r$ successors. Thus,
the third and the second constraint put together require that $A^\I$ has at most three elements, which contradicts the first constraint.
Thus, we have seen that the concept $E$ is actually unsatisfiable.

Using the syntax for $\ALCSCC$ introduced in \cite{Baad17}, we can write the following $\ALCSCC$ concept description
$$
E' := 
      \suc(A\subseteq r) \sqcap \suc(|r| \leq 3),
$$
and state the global constraint $|A| \geq 4$ in an ECBox. But now we have that $E'$ is satisfiable w.r.t.\ this ECBox since
the constraints in $E'$ are local. In fact, the first constraint in $E'$ is satisfied by individuals for which every role successor that belongs to
$A$ is also an $r$ successors of this individual. Together with the second constraint, this only implies that an individual that belongs
to $E'$ has at most three role successors belonging to $A$, but this does not constrain the overall number of elements of $A$, and thus
does not contradict the statement in the ECBox, which is global. For example, an interpretation $\I$ consisting of four individuals belonging to $A$,
none of which has any role successors, is a model of the global constraint $|A| \geq 4$, and every of its elements belongs to $E'$. In contrast,
none of the individuals in $\I$ belongs to the $\ALCplus$ concept $E$ since the second constraint of $E$ is clearly violated.
\end{example}

The local successor constraints of $\ALCSCC$ can clearly be simulated in $\ALCplus$ by using $C\cap (\bigcup_{r\in N_R} r)$ instead
of $C$ when formulating the constraints. Thus, $\ALCSCC$ concepts can be expressed by 
$\ALCplus$ concepts. In addition, extended cardinality constraints (ECBoxes), as introduced above,
are expressible within $\ALCplus$ concept descriptions, as are nominals, the universal role, and role negation.
Recall that a \emph{nominal} is of the form $\{a\}$ where $a\in N_I$, and is interpreted as the singleton set $\{a^\I\}$
by any finite interpretation $\I$.  The \emph{universal role} $u$ is interpreted as $u^\I = \Delta^\I\times\Delta^\I$,
\emph{role conjunction} as $(r\sqcap s)^\I = r^\I\cap s^\I$, and
\emph{role negation} as $(\neg r)^\I = (\Delta^\I\times\Delta^\I)\setminus r^\I$.

\begin{proposition}\label{express:prop}
$\ALCplus$ concepts can polynomially express nominals, role conjunctions, and $\ALCSCC$ ECBoxes, and thus also ABoxes,
$\ALC$ ECBoxes and $\ALCSCC$ TBoxes. In addition, they have the same expressivity as concepts of $\ALCSCC$ extended
with the universal role or with role negation, whereas both of these features are not expressible in plain $\ALCSCC$.
\end{proposition}

\begin{proof}
ECBoxes correspond to Boolean combinations of concepts of the form
$\constr(\con)$ where $\con$ contains only concept descriptions as set variables. Since the concepts occurring in $\con$ are
interpreted globally when viewed as set variables, such a constraint expression $\constr(\con)$ is satisfied either by no element of 
$\Delta^\I$ or by all of them. 
Consequently, their effect is to enforce the constraint on the whole interpretation domain if they are conjoined to a concept description.

Nominals are concepts that must be interpreted as singleton sets. Given a concept name $A$, we can enforce that it is interpreted as a
singleton set using the constraint expression $\constr(|A| = 1)$. Regarding role conjunction, the constraint
$\constr(\top \subseteq \constr(t = r\cap s))$ ensures that, for every individual $d$, its $t$ successors are exactly the individuals
that are both its $r$ and $s$ successors. 

The constraint $\constr(\top \subseteq \constr(u = \universe))$ ensures that $u$ is the universal role since it says that the $u$-successors
of every individual are all the elements of the interpretation domain. Conversely, if the universal role is available, then every
individual has all individuals as a role successors, and thus the difference between the semantics of $\ALCSCC$ and $\ALCplus$
goes away.

Regarding role negation, for given role names $r,\rc$, the constraint $\constr(\top \subseteq \constr(r\cap\rc\subseteq\emptyset))$
enforces that, for every individual, the sets of its $r$ and $\rc$ successors are disjoint. In addition, the constraint
$\constr(\top \subseteq \constr(|r|+|\rc| = |\universe|))$ says that elements of the domain that are not $r$ successors
of a given individual must be $\rc$ successors. Thus, we can express in $\ALCplus$ that the role $\rc$ is interpreted as the complement
of $r$, i.e. $\rc^\I = \Delta^\I\times\Delta^\I\setminus r^\I$ for every finite interpretation $\I$. Conversely, role negation allows
us to express the universal role in $\ALCSCC$: the $\ALCSCC$ constraint $\constr(r\cup \neg r = u)$ is satisfied by an individual
$d$ if the set of its $u$ successors consists of it $r$ and its $\neg r$ successors, and thus all elements of the interpretation domain.
Thus, conjoining such constraint at every place where $u$ is used ensures that $u$ really acts as the universal role.

Inexpressibility of role negation and of the universal role in $\ALCSCC$ can easily be shown using the fact that
models of $\ALCSCC$ TBoxes are closed under disjoint union of finite interpretations, whereas this is not the case
in the presence of role negation or the universal role.
\end{proof}

\section{Satisfiability of~\hmath$\ALCplus$ concept descriptions}
\label{NExpTime:complex:sect}

In the following we consider an $\ALCplus$ concept description $E$ and show how to test $E$ for satisfiability by reducing this problem
to the problem of testing satisfiability of QFBAPA formulae. Since the reduction is exponential and satisfiability
in QFBAPA is in \NP, this yields a \NExpTime upper bound for satisfiability of $\ALCplus$ concept descriptions.
This bound is optimal since consistency of extended cardinality constraints in \ALC, as introduced in \cite{BaEc17},
is already \NExpTime hard, and can be expressed as an $\ALCplus$ satisfiability problem by Proposition~\ref{express:prop}. 

Our \NExpTime algorithm combines ideas from the satisfiability algorithm for $\ALCSCC$ concept descriptions \cite{Baad17}
and the consistency procedure for $\ALC$ ECBoxes \cite{BaEc17}. In particular, we use the notion of a type, as introduced
in \cite{BaEc17}. This notion is also similar to the Venn regions employed in \cite{Baad17}.
Given a set of concept descriptions $\M$,
the \emph{type} of an individual in an interpretation consists of the elements of $\M$ to which the
individual belongs. Such a type~$t$ can also be seen as a concept description $C_t$, which is the
conjunction  of all the elements of~$t$.
We assume in the following that $E$ is an arbitrary, but fixed $\ALCplus$ concept and $\M_E$ consists of \emph{all subdescriptions} of the concept description $E$ 
as well as the \emph{negations of these subdescriptions}.
In Example~\ref{alcplus:ex}, the set $\M_E$ consists of 
$$
E, \neg E, \constr(|A| \geq 4), \neg\constr(|A| \geq 4), \constr(A\subseteq r), \neg \constr(A\subseteq r), \constr(|r| \leq 3), \neg\constr(|r| \leq 3), A,\neg A.
$$

\begin{definition}\label{type:def}
A subset $t$ of $\M_E$ is a \emph{type} for $E$ if it satisfies the following properties:
\begin{enumerate}
\item\label{type:item:one}
  for every concept description $\neg C\in \M_E$, either $C$ or $\neg C$ belongs to $t$; 
\item\label{type:item:two}
  for every concept description $C\sqcap D\in \M_E$, we have that $C\sqcap D\in t$ iff $C\in t$ and $D\in t$;
\item\label{type:item:three}
  for every concept description $C\sqcup D\in \M_E$, we have that $C\sqcup D\in t$ iff $C\in t$ or $D\in t$.
\end{enumerate}
We denote the set of all types for $E$ with $\text{types}(E)$.
Given an interpretation $\I$ and a domain element $d\in \Delta^\I$, the \emph{type of $d$ w.r.t.\ $E$} is the set
$
t^E_\I(d) := \{ C\in \M_E \mid d\in C^\I\}.
$
\end{definition}
 It is easy to show that the type of an individual really satisfies the conditions stated in the
definition of a type. In our example, the following are the only types containing $E$:
\begin{eqnarray}
t_1 &:=& \{E, \constr(|A| \geq 4), \constr(A\subseteq r), \constr(|r| \leq 3), A\}, \label{type:one}\\
t_2 &:=& \{E, \constr(|A| \geq 4), \constr(A\subseteq r), \constr(|r| \leq 3), \neg A \label{type:two} \}.
\end{eqnarray}
 
Due to Condition~(1) in the definition of types, concept descriptions $C_t, C_{t'}$ induced by different types $t\neq t'$ are disjoint,
and all concept descriptions in $\M_E$ can be obtained as the union of the concept descriptions induced by
the types containing them, i.e., we have 
$$
C^\I = \bigcup_{t\,\textit{type\,with}\,C\in t} C_t^\I
$$ 
for all $C\in \M_E$ and finite interpretations $\I$. 
Since the concepts induced by types are disjoint, the following holds for all finite interpretations $\I$:
$$
|C^\I| = \sum_{t\,\textit{type\,with}\,C\in t} |C_t^\I|\ \ \ \
\mbox{and}\ \ \ \
|C_t^\I| = |\bigcap_{C\in t} C^\I|,
$$
where the latter identity is an immediate consequence of the definition of $C_t$ as the
conjunction  of all the elements of $t$. In our example, we have $|E^\I| = |C_{t_1}^\I| + |C_{t_2}^\I|$.

Given a type $t$, the constraints occurring in the top-level Boolean structure of $t$
induce a QFBAPA formula $\psi_t$, in which the concepts $C$ and roles $r$ occurring in 
these constraints are replaced by set variables $X_C$ and $X^t_r$, respectively.
In our example, $t_1$ and $t_2$ contain the same constraints, and the associated QFBAPA formulae are clearly unsatisfiable:
$$
\psi_{t_i} = |X_A|\geq 4 \wedge X_A \subseteq X_r^{t_i} \wedge |X_r^{t_i}|\leq 3\ \ \mbox{for}\ i=1,2.
$$
Note that set variables corresponding to concepts are independent of the type $t$, i.e., they are shared by all types, whereas the
set variables corresponding to roles are different for different types. This corresponds to the fact that roles are evaluated locally, 
but concepts are evaluated globally in the semantics of $\ALCplus$.
In order to ensure that the Boolean structure of concepts is respected by the set variables, we introduce the formula 
\begin{eqnarray*}\label{beta:t}
\beta = 
\bigwedge_{{C\sqcap D}\in \M_E} X_{C\sqcap D}=X_C\cap X_D\wedge 
\bigwedge_{{C\sqcup D}\in \M_E} X_{C\sqcup D}=X_C\cup X_D\wedge
\bigwedge_{{\neg C}\in \M_E} X_{\neg C}={(X_C)}^c.
\end{eqnarray*}
Overall, we translate the $\ALCplus$ concept $E$ into the QFBAPA formula
$$
\delta_E := (|X_E| \geq 1) \wedge \beta \wedge \bigwedge_{t\in\text{types}(E)} (|\bigcap_{C\in t} X_C| = 0) \vee \psi_t.
$$
Intuitively, to satisfy $E$, we need to have at least one element in it, which explains the first conjunct. 
The third conjunct together with $\beta$ ensures that, for any type that is realized (i.e., has elements), 
the constraints of this type are satisfied.
 
In our example, $\beta$ ensures that $X_E = \bigcap_{C\in t_1} X_C \cup \bigcap_{C\in t_2} X_C$ is satisfied.
Together with $|X_E| \geq 1$ this implies that there is an $i\in\{1,2\}$ such that $|\bigcap_{C\in t_i} X_C| > 0$
must hold. But then we need to satisfy $\psi_{t_i}$, which is impossible since this QFBAPA formula is unsatisfiable.
Thus, we have seen that $\delta_E$ is not solvable, which corresponds to the fact $E$ that is unsatisfiable.

The following two lemmas state that  solvability of $\delta_E$ and satisfiability of $E$ are indeed equivalent.

\begin{lemma}\label{completeness:alcplus}
If the $\ALCplus$ concept description $E$ is satisfiable, then the QFBAPA formula $\delta_E$ is also satisfiable.
\end{lemma}

\begin{proof}
Assume that the finite interpretation $\I$ satisfies $E$, i.e., there is a $d_0\in \Delta^\I$ such that $d_0\in E^\I$.
We define $\sigma(X_C) := C^\I$ for all concepts $C\in \M_E$. Then
we have $d_0\in \sigma(X_E)$, and thus $\sigma$ satisfies the cardinality constraint $|X_E|\geq 1$.
In addition, $\sigma$ clearly satisfies $\beta$. For example,
$\sigma(X_{C\sqcap D}) = (C\sqcap D)^\I = C^\I\cap D^\I = \sigma(X_C)\cap \sigma(X_D) = \sigma(X_C\cap X_D)$.
For every type $t$ we have $C_t^\I = \bigcap_{C\in t} C^\I = \bigcap_{C\in t} \sigma(X_C) = \sigma(\bigcap_{C\in t} X_C)$,
and thus $\sigma(|\bigcap_{C\in t} X_C|) = 0$ iff $C_t^\I = \emptyset$.

Let $t$ by a type such that $\sigma(|\bigcap_{C\in t} X_C|) \neq 0$. Then there is an individual $d\in \Delta^\I$
such that $d\in C_t^\I$. The semantics of $\ALCplus$ then implies that we can extend $\sigma$ to a solution of $\psi_t$ by interpreting the
set variables with superscript $t$ using the role successors of $d$:
$$
\sigma(X_r^t) := \{e\mid (d,e)\in r^\I\}.
$$
If $t$ is a type such that $\sigma(|\bigcap_{C\in t} X_C|) = 0$, then it is not necessary for $\sigma$ to satisfy $\psi_t$.
We can thus extend $\sigma$ to the set variables with superscript $t$ in an arbitrary way, e.g.\ by interpreting all of
them as the empty set. Overall, this show that we can use an interpretation satisfying $E$ to define a solution $\sigma$ of $\delta_E$.
\end{proof}

Next, we show that the converse of Lemma~\ref{completeness:alcplus} holds as well.

\begin{lemma}\label{soundness:alcplus}
If the QFBAPA formula $\delta_E$ is  satisfiable, then
the $\ALCplus$ concept description $E$ is also satisfiable.
\end{lemma}

\begin{proof}
Assume that there is a solution $\sigma$ of $\delta_E$.
We claim that, for every element $e\in \sigma(\universe)$,
there is a unique type $t_e$ such that $e\in \bigcap_{C\in t_e} \sigma(X_C)$.
In fact, we can define $t_e$ as
$$
t_e := \{C\in \M_E \mid e\in \sigma(X_C)\}.
$$
Since $\sigma$ satisfies $\beta$, the set $t_e$ is indeed a type. For example, assume that $C\sqcup D\in t_e$.
Then $e\in \sigma(X_{C\sqcup D}) = \sigma(X_C)\cup\sigma(X_D)$ iff $e\in \sigma(X_C)$ or $e\in \sigma(X_D)$ iff
$C\in t_e$ or $D\in t_e$. Satisfaction of the other conditions in the definition of a type
can be shown similarly. Regarding uniqueness, assume that $t$ is a type different from $t_e$. Then there
is an element $C\in \M_E$ such that (modulo removal of double negation) $C\in t_e$ and $\neg C\in t$.
But then $e\in \sigma(X_C)$ implies $e\not\in \sigma((X_C)^c) = \sigma(X_{\neg C})$, and thus
$e\not\in \bigcap_{D\in t} \sigma(X_D)$.

Let
$$
T_\sigma := \{t \mid t\ \mbox{type with}\ \sigma(|\bigcap_{C\in t} X_C|) \neq 0\}
$$
be the set of all types that are realized by $\sigma$. Note that, by what we have shown above, we have
$T_\sigma = \{t_e \mid e\in \sigma(\universe)\}$.

We now define a finite interpretation $\I$ and show that
it satisfies $E$. The interpretation domain consists of copies of the realized types, where the number of copies is
determined by $\sigma$:
$$
\Delta^\I := \{(t,j) \mid t\in T_\sigma\ \mbox{and}\ 1\leq j\leq \sigma(|\bigcap_{C\in t} X_C|)\}.
$$
Since for every element $e\in \sigma(\universe)$ there is a unique type $t_e$ such that $e\in \bigcap_{C\in t_e} \sigma(X_C)$,
there is a bijection $\pi$ from $\sigma(\universe)$ to $\Delta^\I$ such that $\pi(e) = (t,j)$ implies that $t = t_e$.

For concept names $A$ we define
$$
A^\I := \{(t,j)\in \Delta^\I \mid A\in t\}
$$
and for role names $r$
$$
r^\I := \{((t,j),\pi(e)) \mid (t,j)\in \Delta^\I \wedge e \in \sigma(X_r^t)\}.
$$
Since $\sigma$ solves the constraint $X_E\geq 1$, there is a $d_0\in \sigma(X_E)$. Let $t_0$ be the unique type
such that $d_0 \in \bigcap_{C\in t_0} \sigma(X_C)$. Then we have $\sigma(|\bigcap_{C\in t_0} X_C)|) \neq 0$, and thus
$(t_0,1)\in \Delta^\I$. To show that $\I$ satisfies $E$, it is sufficient to show that $(t_0,1)\in E^\I$.

For this, we show the following more general claim:
for all concept descriptions $C\in \M_E$ and all $(t,j)\in \Delta^\I$ we have
\begin{equation}\label{completeness:eq:two}
(t,j) \in C^\I\ \mbox{iff}\ \ C\in t.
\end{equation}
We show $(\ref{completeness:eq:two})$ by \emph{induction on the structure of $C$}:
\begin{itemize}
\item
  Let $C = A$ for $A\in N_C$. Then $(\ref{completeness:eq:two})$ is an immediate consequence of the definition of $A^\I$ for concept names $A$.
\item
  Let $C = \neg D$. Then induction yields $(t,j) \in D^\I$ iff $D\in t$.
  By contraposition, this is the same as $(t,j) \not\in D^\I$ iff $D\not\in t$.
  By Condition~\ref{type:item:one} in the definition of types and the semantics of negation, this
  is in turn equivalent to $(t,j) \in (\neg D)^\I$ iff $\neg D\in t$.
\item
  Let $C = D_1 \sqcap D_2$.
  Then induction yields $(t,j) \in D_1^\I$ iff $D_1\in t$ and $(t,j) \in D_2^\I$ iff $D_2\in t$.
  From this, we obtain $(t,j) \in (D_1\sqcap D_2)^\I$ iff $D_1\sqcap D_2\in t$ using
  Condition~\ref{type:item:two} in the definition of types and the semantics of conjunction.
\item
  The case where $C = D_1 \sqcup D_2$ can be handled similarly, using Condition~\ref{type:item:three} in the definition of types and the semantics of disjunction.
\item
$C = \constr(\con)$ be a constraint expression. First, assume that $C \in t$. Then the translation $\con'$ of $\con$ using set variables $X_D$ and $X_r^t$
is a conjunct in $\psi_{t}$. In addition, since $(t,j)\in \Delta^\I$, we have $\sigma(|\bigcap_{D\in t} X_D|) \neq 0$.
Consequently, $\sigma$ satisfies this translation $\con'$.
Thus, to show that $(t,j)\in C^\I$, it is sufficient to show that the following holds: 
\begin{enumerate}
\item
  $\pi(\sigma(X_r^t)) = r^\I(t,j)$ and
\item
  $\pi(\sigma(X_D)) = D^\I$ for all concepts $D$ occurring in the constraint $c$.
\end{enumerate}
The first statement is an immediate consequence of the definition of the interpretation of the roles in $\I$.
 
To show the second statement, first assume that $e\in \sigma(X_D)$. Then $\pi(e) = (t_e,j')$ where $t_e$ is the unique type such that
$e\in \bigcap_{F\in t_e} \sigma(X_F)$. Thus, $e\in \sigma(X_D)$ implies that $D\in t_e$. By induction, we obtain $\pi(e) = (t_e,j')\in D^\I$.
Second, assume that $\pi(e) = (t_e,j')\in D^\I$. Then induction yields $D\in t_e$, and thus $e\in \sigma(X_D)$.

Conversely, assume that $C \not\in t$. Then $\neg\Succ(\con)\in t$, and thus the translation $\neg \con'$ of $\neg \con$ using set variables $X_D$ and $X_r^t$
is a conjunct in $\psi_{t}$. We can now proceed as in the first case, but with $\neg \con$ and $\neg \con'$ in place of $\con$ and $\con'$.
\end{itemize}
This completes the proof of $(\ref{completeness:eq:two})$ and thus the proof of the lemma.
\end{proof}

We have shown that the question of whether an $\ALCplus$ concept description $E$ is satisfiable can be reduced to checking whether
the corresponding QFBAPA formula $\delta_E$ is satisfiable. Since the size of $\delta_E$ is exponential in  the size of $E$, this yields
the following complexity result.

\begin{theorem}\label{NExp:upper:thm}
Satisfiability of $\ALCplus$ concept descriptions 
is \NExpTime-complete independently of whether the numbers occurring in these descriptions are encoded in unary or binary.
\end{theorem}

\begin{proof}
Since satisfiability of QFBAPA formulae can be decided within \NP even for binary coding of numbers \cite{KuRi07},
it is sufficient to show that the size of the QFBAPA formula $\delta_E$ is
at most exponential in the size of $E$. This is an easy consequence of the fact that there are at most exponentially many types $t$
since the cardinality of $\M_E$ is linear in the size of $E$.
This implies that the conjunction over all types in $\delta_E$ has only exponentially many conjuncts. The conjunct for a type $t$ is
of the form $(|\bigcap_{C\in t} X_C| = 0) \vee \psi_t$. Since every type contains only linearly many concepts, and these concepts have linear size,
both $(|\bigcap_{C\in t} X_C| = 0)$ and $\psi_t$ is of polynomial size.
Obviously, $(|X_E| \geq 1)$ has linear size,
and the formula $\beta$ has polynomial size since $\M_E$ contains linearly many elements of linear size.

The \NExpTime lower bound is inherited from consistency of $\ALC$ ECBoxes \cite{BaEc17} due to Proposition~\ref{express:prop}.
As argued in \cite{BaEc17}, this lower  bound already holds if numbers are encoded in  unary since one can use small
ECBoxes to generate large numbers from small ones.
\end{proof}

Thanks to Proposition~\ref{express:prop}, the \NExpTime upper bound carries over to satisfiability of $\ALCplus$ knowledge bases, which may feature an ABox, a TBox and an ECBox.

\section{Restricted Cardinality Constraints and ABoxes\\ in~\hmath$\ALCSCC$}
\label{restr:card:sect}

In Definition~\ref{ALCSCC:def}, we have introduced the DL $\ALCSCC$ and ECBoxes. As mentioned above, \NExpTime hardness already
holds for consistency of $\ALCSCC$ ECBoxes, and Theorem~\ref{NExp:upper:thm} yields the matching upper bound since ECBoxes
can be expressed by $\ALCplus$ concepts by Proposition~\ref{express:prop}. The same proposition also states that ABoxes can be
expressed by $\ALCplus$ concepts, which yields a \NExpTime upper bound also for consistency of $\ALCSCC$ ABoxes w.r.t.\ $\ALCSCC$ ECBoxes.

For the sub-logic $\ALC$ of $\ALCSCC$, a restricted notion of cardinality boxes, called RCBoxes, was introduced in \cite{BaEc17},
and it was shown that this restriction lowers the complexity of the consistency problem from \NExpTime to \ExpTime.
In \cite{Baad19,Baad19b} it was shown that the same is true for $\ALCSCC$. 
Here we prove that this result can be extended to consistency of $\ALCSCC$ ABoxes w.r.t.\ $\ALCSCC$ RCBoxes.
In the presence of ECBoxes, this extension is irrelevant since ECBoxes can express nominals, and thus also ABoxes. However, this is
not the case for RCBoxes. Below, we actually consider an extension of RCBoxes, which were called ERCBoxes in \cite{Rudo19}.

\begin{definition}[RCBoxes]\label{RCBox:def}
  \emph{Semi-restricted $\ALCSCC$ cardinality constraints} are of the form
  \begin{align}\label{rc:constr}
    N_1 |C_1|+ \dots+N_k |C_k| + M \le N_{k+1}|C_{k+1}|+\dots+N_{k+\ell}|C_{k+\ell}|,
  \end{align}
  where $C_i$ are $\ALCSCC$ concept descriptions, $N_i$ are integer constants for $1\le i\le k+\ell$, and 
  $M$ is a non-negative integer constant.
  An \emph{extended restricted $\ALCSCC$ cardinality box (ERCBox)} is a positive Boolean combination of semi-restricted $\ALCSCC$ cardinality constraints.

An interpretation $\I$ is a model of the semi-restricted $\ALCSCC$ cardinality constraint \eqref{rc:constr} if
  \begin{align*}
    N_1 |C_1^\I|+ \dots+N_k |C_k^\I| + M \le N_{k+1}|C_{k+1}^\I|+\dots+N_{k+\ell}|C_{k+\ell}^\I|.
  \end{align*}
The notion of a model is extended to ERCBoxes using the usual interpretation of conjunction and disjunction 
in propositional logic.
\end{definition}
Note that $\ALCSCC$ ECBoxes can express both ERCBoxes and ABoxes. The restricted cardinality boxes (RCBoxes)
introduced in \cite{BaEc17,Baad19,Baad19b}
differ from ERCBoxes in that the number $M$ in constraints of the form \eqref{rc:constr} must be zero, and that only
conjunction of such constraints is allowed.
Since \ExpTime-hardness already holds for consistency of RCBoxes in $\ALCSCC$ without an ABox \cite{BaEc17,Baad19,Baad19b}, 
we obtain the following complexity lower bound.
Actually, the hardness proof does not require large number, and thus \ExpTime-hardness even holds
for unary coding of numbers.

\begin{proposition}\label{prop:exptime}
  The consistency of $\ALCSCC$ ERCBoxes w.r.t.\ $\ALCSCC$ ABoxes 
  is \ExpTime-hard, independently of whether numbers are encoded in unary or binary.
\end{proposition}

Following the approach in \cite{Baad19,Baad19b} for consistency of $\ALCSCC$ RCBoxes, 
we show the \ExpTime upper bound for numbers encoded in binary using type elimination, 
where the notion of augmented type from \cite{Baad17} is used, and 
a second step for removing types is added to take care of the ERCBox,
similarly to what is done in \cite{BaEc17}. 
In addition, the ABox individuals are taken into account by making them elements of exactly one augmented type.

The \ExpTime upper bound for our procedure on the one hand depends on the following lemma, which applies in our setting due to the 
special form of semi-restricted cardinality constraints. It is an extension of Lemma~10 in \cite{BaEc17}.

\begin{lemma}\label{lem:ineq-system-properties}
  Let $\phi$ be a system of linear inequalities consisting of $A\cdot \boldsymbol{v}\ge \boldsymbol{b}$ and $\boldsymbol{v}\ge \boldsymbol{0}$,
  where $A, B$ are matrices of integer coefficients, $\boldsymbol{b}$ is a vector of non-negative integer parameters, and $\boldsymbol{v}$ is the variable vector. 
  \begin{enumerate}
    \item The solutions of $\phi$ are closed under addition.  
    \item If $\{v_1,\ldots,v_k\}$ is a set of variables such that, for all $v_i\ (1\leq i\leq k)$, 
          $\phi$ has a solution $\boldsymbol{c}_i$ in which the $i$th component $c_i^{(i)}$ is not $0$, 
          then there is a non-negative integer solution $\boldsymbol{c}$ of $\phi$ such that, 
          for all $i, 1\leq i\leq k$, the $i$th component $c^{(i)}$ of $\boldsymbol{c}$ satisfies $c^{(i)}\geq 1$. 
    \item Deciding whether $\phi$ has a non-negative integer solution can be done in polynomial time.
  \end{enumerate}
\end{lemma}

\begin{proof}
(1) Let $\boldsymbol{c}, \boldsymbol{d}$ be two solutions of $\phi$. Since all components of these vectors are non-negative,
this is clearly also the case for their sums. In addition, we have
$$
A\cdot (\boldsymbol{c} + \boldsymbol{d}) =
A\cdot \boldsymbol{c} + A\cdot \boldsymbol{d} \geq \boldsymbol{b} + \boldsymbol{b} \geq \boldsymbol{b},
$$
where the first inequality holds since $\boldsymbol{c}, \boldsymbol{d}$ are solutions of $\phi$,
and the last inequality holds since the components of $\boldsymbol{b}$ are non-negative. 

(2) Given solutions $\boldsymbol{c}_i$ as described in the second part of the lemma, the solution
$\boldsymbol{c}$ satisfying the stated properties can be obtained as their sum. 

(3) It is well-known that solvability in the rational numbers of a system of inequalities of the form stated in the lemma 
can be decided in polynomial time \cite{GrLS88}. In addition, if $\phi$ has a rational solution $\boldsymbol{c}$, then
it also has an integer solution. In fact, let $D$ be the least common multiple (lcm) of the denominators of the components of $\boldsymbol{c}$.
Then $D\cdot\boldsymbol{c}$ is an integer vector that is a solution of $\phi$ due to closure under addition of solutions,
as stated in the first part of the lemma.
\end{proof}

Another important ingredient of our \ExpTime procedure are augmented types, which have been introduced in \cite{Baad17} to show that
satisfiability in $\ALCSCC$ w.r.t.\ concept inclusions is in \ExpTime.
We use the notion
of a type as introduced in Definition~\ref{type:def}
(see also Definition~3 of \cite{Baad19}), but extended such that it takes the ABox $\A$ and the ERCBox $\R$ into account,
i.e., the set $\M(\R,\A)$ of all relevant concept descriptions contains all subdescriptions of the concept
descriptions occurring in $\R$ or $\A$ as well as their negations. In addition, for every individual name $b\in\Ind_\A$
(where $\Ind_\A$ denotes the set of individual name occurring on $\A$), the set
$\M(\R,\A)$  contains this name and its negation.

\begin{definition} 
Let $\A$ be an $\ALCSCC$ ABox and $\R$ be an $\ALCSCC$ ERCBox.
A subset $t$ of $\M(\R,\A)$ is a \emph{type for $\R$ and $\A$} if it satisfies the following properties:
\begin{enumerate}
\item 
  for every concept description $\neg C\in \M(\R,\A)$, either $C$ or $\neg C$ belongs to $t$; 
\item
  for every individual name $b\in \M(\R,\A)$, either $b$ or $\neg b$ belongs to $t$;
\item 
  for every concept description $C\sqcap D\in \M(\R,\A)$, we have that $C\sqcap D\in t$ iff $C\in t$ and $D\in t$;
\item
  for every concept description $C\sqcup D\in \M(\R,\A)$, we have that $C\sqcup D\in t$ iff $C\in t$ or $D\in t$.
\end{enumerate}
\end{definition}

Intuitively, a type containing $b\in\Ind_\A$ is supposed to represent the individual $b$. 
Our type elimination procedure will ensure that, for every individual $b$ exactly one type is available.
However, ERCBoxes do not allow us to express that this type should be realized by only one element of
the model. In our model construction, we will actually have several individuals that realize such a type,
and choose one of them to actually interpret the individual $b$. With respect to membership in concepts,
this ``chosen'' individual and its copies behave the same. However, to satisfy role assertions we must ensure
that role successors are always the chosen individuals. This can be achieved by adding an appropriate
cardinality constraint when defining augmented types (see below).

Augmented types consider not just the concepts to which a single individual belongs, but also the Venn regions to which its role successors
belong. Basically, we define the notion of a Venn region as in \cite{Baad17,Baad19,Baad19b}, but extend it by (i)~always considering the set
of all set variables $X_D$ for subdescriptions $D$ occurring in $\R$ or $\A$ and $X_r$ for $r\in N_R$ rather than just the ones occurring 
in the given QFBAPA formula; and (ii)~additionally considering set variables $X_b$ for all individuals $b\in\Ind_\A$. 

\begin{definition}[Venn region]
Let $\A$ be an $\ALCSCC$ ABox and $\R$ be an $\ALCSCC$ ERCBox, and let $X_1,\ldots,X_k$ be an enumeration of
all set variables $X_C$ for subdescriptions $C$ occurring in $\R$ or $\A$, $X_r$ for $r\in N_R$, and $X_a$ for individual names $a\in \Ind_\A$. 
A \emph{Venn region for $\R$ and $\A$} is of the form
$$
X_1^{c_1}\cap\ldots \cap X_k^{c_k},
$$
where $c_i$ is either empty or $c$ for $i = 1,\ldots,k$. 
\end{definition}

Again, a Venn region containing $b$ says
that this element corresponds to the individual $b\in\Ind_\A$. But now QFBAPA allows us to formulate constraints on
the cardinality of the sets $X_b$. In particular, by adding $|X_b| \leq 1$ we can ensure that there is only one
role successor that belongs to a type containing $b$.

Given a type $t$ for $\R$ and $\A$, we consider the corresponding QFBAPA formula $\phi_t$, which is induced by the (possibly negated) successor
constraints occurring in $t$. We conjoin to this formula the set constraint
    $$
    X_{r_1}\cup\ldots\cup X_{r_n} = \universe,
    $$
where $N_R = \{r_1,\ldots,r_n\}$,\footnote{%
Without loss of generality we assume that $N_R$ contains only the role names occurring in $\R$ and $\A$.
}
as well as the cardinality constraints 
$$
|X_b|\leq 1
$$ 
for $b\in\Ind_\A$. In case $a\in\Ind_\A$ belongs to $t$,
we consider
all role assertions $r_1(a,b_1), \ldots, r_k(a,b_\ell)$ with $a$ in the first component in $\A$,
and add the conjuncts 
$$
|X_{b_1}\cap X_{r_1}|\geq 1, \ldots, |X_{b_\ell}\cap X_{r_\ell}|\geq 1.
$$
For the resulting formula $\phi_t'$, we compute the number $N_t$ that bounds the number of Venn regions that
need to be non-empty in a solution of $\phi_t'$ (see Lemma~1 in \cite{Baad19}). 

\begin{definition}\label{aug:typ:def}
Let $\R$ be an $\ALCSCC$ ERCBox and $\A$ be an $\ALCSCC$ ABox.
An \emph{augmented type} $(t,V)$ for $\R$ and $\A$ consists of a type $t$ for $\R$ and $\A$ together with a set of Venn region $V$
such that $|V|\leq N_t$ and the formula $\phi_t'$ has a solution in which exactly the Venn regions in $V$ are
non-empty.
\end{definition}
The existence of a solution of $\phi_t'$ in which exactly the Venn regions in $V$ are
non-empty can obviously be checked (within NP) by adding to $\phi_t'$ conjuncts that
state non-emptiness of the Venn regions in $V$ and the fact that the union of these
Venn regions is the universal set (see the description of the \PSpace algorithm in the proof of
Theorem~1 in \cite{Baad17}). Another easy to show observation is that there are only exponentially many
augmented types (see the accompanying technical report of \cite{Baad17} for a proof of the following lemma).

\begin{lemma}
\label{aug:type:comp}
Let $\R$ be an $\ALCSCC$ ERCBox and $\A$ be an $\ALCSCC$ ABox.
The set of augmented types for $\R$ and $\A$ contains at most exponentially many elements in the
size of $\R$ and $\A$, and it can be computed in exponential time.
\end{lemma}

The type elimination procedure checking the consistency of $\ALCSCC$ RCBoxes introduced in \cite{Baad19}
starts with the set of all augmented types, and then successively eliminates augmented
types 
\begin{enumerate}
\item[(i)]whose Venn regions are not realized by the currently available augmented types, or 
\item[(ii)] whose first component is forced to be empty by the constraints in $\R$. 
\end{enumerate}
To make the first reason for elimination more precise, assume that $\At$ is a set of augmented types and that $v$ is a Venn region. 
In the following, let $D$ denote an $\ALCSCC$ concept and $b$ an individual name.
The Venn region $v$ yields a set of concept descriptions $S_v$ that contains, for every set variable $X_D$ ($X_b$)
occurring in $v$, the element $D$ ($b$) in case $v$ contains $X_D$ ($X_b$) 
and the element $\neg D$ ($\neg b$) in case $v$ contains $X_D^c$ ($X_b^c$).
It is easy to see that $S_v$ is actually a subset of $\M(\R,\A)$ (modulo removal of double negation).

\begin{definition}
Let $\At$ be a set of augmented types and $v$ a Venn region,
We say that $v$ is \emph{realized by $\At$} if there is an augmented type $(t,V)\in \At$ such that
$S_v\subseteq t$.
\end{definition}

The fact that both Venn regions and types contain every concept or individual (set variable) 
either positively or
negatively implies that, modulo elimination of double negation, we actually have $S_v = t$ whenever $S_v\subseteq t$.
Note that, for some Venn regions $v$, there may not be a type $t$ such that $S_v\subseteq t$ since in the
definition of Venn regions we do not consider the Boolean structure of concepts (e.g., a Venn region may contain
$X_{C\sqcap D}$ positively, but $X_D$ negatively). However this will not be a problem since in our proofs we will 
always work with Venn regions that are contained in types. 

Also note that the condition that Venn regions must  be realized also takes
care of role assertions. In fact, consider an augmented type $(t,V)$ and assume that the type $t$ contains $a$ and $r(a,b)\in \A$. 
Then $\phi_t'$ contains the conjuncts $|X_{b}\cap X_{r}|\geq 1$ and $|X_b|\leq 1$. Consequently, $V$ contains a
Venn region $v$ in which $X_{r}$ and $X_{b}$ occur positively, and thus $b\in S_v$. If this Venn region is realized by the augmented type $(s,W)$,
then $s$ must contain $b$ (i.e., represent the individual $b$). Intuitively, this ensures that $a$ has $r$-successor $b$.
In order to show this formally, however, some more work is needed since we must ensure that $a$ is actually linked
to the copy chosen to represent $b$ rather than just to a type containing $b$ (see the proof of Lemma~\ref{soundness:RCBoxes} below).

We are now ready to formulate our algorithm.
We assume without loss of generality that $\A$ is non-empty, and thus contains at least one individual.
In addition, we assume that $\R$ is a conjunction of semi-restricted constraints, which we call a \emph{conjunctive ERCBox}. 
We will argue later why is is sufficient to restrict the attention to conjunctive ERCBoxes.

\begin{algor}\label{alg:restr}
Let $\R$ be a conjunctive  $\ALCSCC$ ERCBox and $\A\neq\emptyset$ be an $\ALCSCC$ ABox. 
First, we compute the set $\M(\R,\A)$ consisting of all subdescriptions of $\R$ and $\A$ as well as the
negations of these subdescriptions, together with the set of all individual names occurring in $\A$ and their negations.
Based on this set $\M(\R,\A)$, we compute the set $\widehat{\At}$ of all augmented types for $\R$ and $\A$. 
We now decide consistency of $\A$ w.r.t.\ $\R$ by 
performing the following three steps: 
\begin{enumerate}
\item\label{step:one}
  Compute all maximal subsets $\At$ of $\widehat{\At}$ such that 
  \begin{enumerate}
  \item\label{cond:one}
    for every individual $b\in \Ind_\A$, there is exactly one augmented type $(t,V)\in \At$ with $b\in t$, 
  \item\label{cond:two}
    if $(t,V)\in \At$ and $b \in t$ for an individual $b\in \Ind_\A$, then $C\in t$ for all concept assertion $C(b)\in \A$,
  \end{enumerate}
  To achieve this,
  in a first step, we can remove all augmented types that do not satisfy condition (\ref{cond:two}). In case
  there is an individual $b\in \Ind_\A$ such that all types containing $b$ have been removed, then the algorithm \emph{fails}.
  Otherwise, choose for every $b\in \Ind_\A$ exactly one of the remaining augmented types whose first component contains
  $b$ and remove all the other augmented types containing $b$.

  Check whether the following two steps succeed for one of the sets $\At$ computed this way.
\item\label{step:three}
   If there is an individual $b\in\Ind_\A$ such that $\At$ does not contain an augmented type $(t,V)$ such that $b\in t$, 
   then the algorithm \emph{fails} for the current set of augmented types.
   Otherwise, it checks whether $\At$ contains an element $(t,V)$ such that not all the Venn regions in $V$ are realized by $\At$.
   If there is no such element $(t,V)$ in $\At$, then continue with the next step.
   Otherwise, let $(t,V)$ be such an element, and set $\At := \At \setminus \{(t,V)\}$.
   Continue with this step, but now using the new current set of augmented types.
\item\label{step:four}
  Let $T_\At : = \{t \mid \mbox{there is $V$ such that}\ (t,V)\in \At\}$, and let
  $\phi_{T_\At}$ be obtained from $\R$ by replacing each $|C|$ in $\mathcal{R}$ with $\sum_{t\in T_\At\text{ s.t. }C\in t}v_t$ and adding $v_t\ge 0$ for each $t\in T_\At$.
  Check whether $T_\At$ contains an element $t$ such that $\phi_{T_\At}\wedge v_t\ge 1$ has no solution.
  If this is the case for $t$, then remove all augmented types of the form $(t,\cdot)$ from $\At$, and continue with the previous step. 
  If no type $t$ is removed in this step, then the algorithm \emph{succeeds}. 
\end{enumerate}
\end{algor}

Before proving that this algorithm runs in exponential time, we show that it is sound and complete.

\begin{lemma}[Soundness]\label{soundness:RCBoxes}
Let $\R$ be a conjunctive $\ALCSCC$ ERCBox and $\A\neq \emptyset$ an $\ALCSCC$ ABox.
If Algorithm~\ref{alg:restr} succeeds on input $\R$ and $\A$, then $\A$ is consistent w.r.t.\ $\R$.
\end{lemma}

\begin{proof}
Assume that the algorithm succeeds on input $\R$ and $\A$, 
and let $\At$ be the final set of augmented types when the algorithm stops successfully. 
Note that $\At\neq\emptyset$ since there is at least one individual $b$ in $\A$, and thus
the algorithm would have failed for an empty set of augmented types.
We show how $\At$ can be used to construct a model $\I$ of $\R$ and $\A$. 

For this construction, we first consider the formula $\phi_{T_\At}$,
which is obtained from $\R$ by replacing each $|C|$ in $\mathcal{R}$ with $\sum_{t\in T_\At\text{ s.t. }C\in t}v_t$ and adding 
$v_t\ge 0$ for each $t\in T_\At$. 
Note that, due to the special form of conjunctive ERCBoxes, we know that this yields a
system of linear inequalities of the form $A\cdot \boldsymbol{v}\ge \boldsymbol{b}$, $\boldsymbol{v}\ge \boldsymbol{0}$.
Since the algorithm has terminated successfully, we know for all $t\in T_\At$ that the formula
$\phi_{T_\At}\wedge v_t\ge 1$ has a solution. By 
Lemma~\ref{lem:ineq-system-properties} this implies that $\phi_{T_\At}$ has a solution in which
all variables $v_t$ for $t\in T_\At$ have a value $\geq 1$ and all variables $v_t$ with $t\not\in T_\At$ have value $0$.
In addition, given an arbitrary number $N\geq 1$, we know that there is a solution $\sigma_N$ of $\phi_{T_\At}$ such that
$\sigma_N(v_t) \geq 1$ and $N | \sigma_N(v_t)$ holds for all $t\in T_\At$. To see this, note that we can just multiply with $N$
a given solution satisfying the properties mentioned in the previous sentence.

We use the  augmented types in $\At$ to determine the right $N$: 
\begin{itemize}
\item
For each augmented type $(t,V)$, we know that the formula
$\phi_t'$ has a solution where exactly the Venn regions in $V$ are non-empty (see Definition~\ref{aug:typ:def}).  
Assume that this
solution assigns a set of cardinality $k_{(t,V)}$ to the universal set.
\item
For each $t\in T_\At$, let $n_t$ be the cardinality of the set
$\{V \mid (t,V)\in \At\}$, i.e., the number of augmented types in $\At$ that have $t$ as their first component. 
\end{itemize}
We now define $N$ as 
$$
N := (\max\{k_{(t,V)} \mid (t,V)\in \At\}) \cdot \prod_{t\in T_\At}n_t,
$$ 
and use the solution
$\sigma_N$ of $\phi_{T_\At}$ to construct a finite interpretation $\I$ as follows. The domain of $\I$ is defined as
$$
\Delta^\I := \{ (t,V)^i \mid (t,V)\in \At\ \mbox{and}\ 1\leq i \leq \sigma_N(v_t)/n_t\}.
$$
Note that $\sigma_N(v_t)/n_t$ is a natural number since $N|\sigma_N(v_t)$ implies $n_t|\sigma_N(v_t)$.
In addition, $\Delta^\I\neq\emptyset$ because $\At\neq\emptyset$ and $\sigma_N(v_t)/n_t \geq 1$ since $\sigma_N(v_t) \geq 1$.
Moreover, for each type $t\in T_\At$, the set
$\{ (t,V)^i \mid (t,V)^i \in \Delta^\I\}$ has cardinality $\sigma_N(v_t)$.
 
The interpretation of the concept names $A$ is based on the occurrence of these names in the first component
of an augmented type, i.e.,
$$
A^\I := \{(t,V)^{i} \in \Delta^\I \mid A\in t\}.
$$
Individual names are treated similarly, however we need to ensure that an individual name is interpreted by a single
element of $\Delta^\I$, and not by a set of cardinality $> 1$. First, note that, due to step~(1) and the failure condition in
step~(2), for each individual name $a\in \Ind_\A$, $\At$ contains exactly one augmented type $(t,V)$ such that $a\in t$.
Let us denote this augmented type with $(t_a,V_a)$. The interpretations domain may contain several copies of
$(t_a,V_a)$, but we interpret $a$ using the first one, i.e., we define
$$
a^\I := (t_a,V_a)^1.
$$

Defining the interpretation of the role names is a bit more tricky. Obviously, it is sufficient to define, for
each role name $r\in N_R$ and each $d\in \Delta^\I$, the set $r^\I(d)$. Thus, consider an element
$(t,V)^{i}\in \Delta^\I$. Since $(t,V)$ is an augmented type in $\At$, the formula $\phi_t'$ has a solution $\sigma$ in which exactly
the Venn regions in $V$ are non-empty, and which assigns a set of cardinality $m:=k_{(t,V)}$ to the universal set. 
In addition, 
  each Venn region $w\in V$ is realized by an augmented type $(t^w,V^w)\in \At$.
Assume that the solution $\sigma$ assigns the finite set $\{d_1,\ldots,d_m\}$ to the
set term $\universe$. We consider an injective mapping $\pi$ of $\{d_1,\ldots,d_m\}$ into
$\Delta^\I$ such that the following holds for each element $d_j$ of $\{d_1,\ldots,d_m\}$:
if $d_j$ belongs to the Venn region $w\in V$, then 
\begin{itemize}
\item
$\pi(d_j) = (t^w,V^w)^{\ell}$ for some $1\leq\ell\leq \sigma_N(v_{t^w})/n_{t^w}$; 
\item
if $w$ contains $X_b$ for an individual name $b\in\Ind_\A$ positively, then $\ell = 1$.
\end{itemize}
Such a bijection exists since,
\begin{itemize}
\item 
  $\sigma_N(v_{t^w})/n_{t^w} \geq \max\{k_{(t',V')} \mid (t',V')\in \At\}\geq k_{(t,V)}=m$;
\item
  due to the presence of the cardinality constraints $|X_b|\leq 1$ in the QFBAPA formula
  $\phi_t'$, there is at most one individual $d_j$ that belongs to a Venn region $w$ containing
  $X_b$ positively. Any other individual $d_k$ belongs to a different Venn region $w'$ not containing $X_b$ positively, and 
  thus $S_w\subseteq t^w$ and $S_{w'}\subseteq t^{w'}$ implies $t^w \neq t^{w'}$ since $b\in t^w$ but $b\not\in t^{w'}$.
  This shows that choosing the index $\ell=1$ when defining $\pi(d_j)$ is possible without getting into conflict
  with the required choice of the index $1$ for a different individual.
\end{itemize}
We now define
$$
r^\I((t,V)^{i}) := \{\pi(d_j) \mid d_j\in \sigma(X_r)\}.
$$

First, note that this definition of the interpretation of roles in $\I$ satisfies the role assertions
in $\A$. To see this, assume that $r(a,b)\in \A$, and let $a^\I = (t_a,V_a)^1$. 
Then $a\in t_a$, which implies that $\phi_{t_a}'$ contains the cardinality constraint $|X_b\cap X_{r}|\geq 1$
as well as the constraint $|X_b|\leq 1$. 
Consequently, $V_a$ contains exactly one Venn region $w$ that contains $X_b$ and $X_r$ positively.
Consider the solution of $\phi_{t_a}'$ used above to define $r^\I((t_a,V_a)^{1})$, and let
$d_j$ be the unique individual belonging to $X_b$ under this solution. 
Then this individual also belongs to $X_r$ under this solution, and we have
$\pi(d_j) = (t^w,V^w)^1 \in r^\I((t_a,V_a)^{1})$.
In addition, $t_w$ contains $b$ since $S_w$ contains $b$ and $S_w\subseteq t_w$. This shows
that $b^\I = (t^w,V^w)^1$, and thus that the role assertion $r(a,b)$ is satisfied by $\I$.

To prove that $\I$ also satisfies the concept assertions in $\A$ and the ERCBox $\R$, we first show
the following claim:\\[.7em]
\textbf{Claim:} \emph{For all concept descriptions
$C\in \M(\R,\A)$, all augmented types $(t,V)\in \At$, and all $i, 1\leq i \leq  \sigma_N(v_t)/n_t$, we have $C\in t$ iff $(t,V)^i\in C^\I$.}\\[.7em]
We prove the claim by induction on the size of $C$:
\begin{itemize}
\item
The cases $C = A$, $C = \neg D$, $C = D_1\sqcap D_2$, and $C = D_1\sqcup D_2$ can be handled as in
the proof of
$(\ref{completeness:eq:two})$ in the proof of Lemma~\ref{soundness:alcplus}.
\item
  Now assume that $C = \suc(\con)$ for a set or cardinality constraint $\con$.
\begin{itemize}
\item
  If $C\in t$, then this constraint
  is part of the QFBAPA formula $\phi'_t$ obtained from $t$, and thus satisfied by the solution $\sigma$ of
  $\phi'_t$ used to define the role successors of $(t,V)^i$. According to this definition, there is a
  $1$--$1$ correspondence between the elements of $\sigma(\universe)$ and the role successors of $(t,V)^i$.
  This bijection $\pi$ also respects the assignment of subsets of $\sigma(\universe)$ to set variables of
  the form $X_r$ (for $r\in N_R$) and $X_D$ (for concept descriptions $D$) occurring in $\phi'_t$, i.e.,
   $$
   \begin{array}{l@{\ \ \ }l}
   (*) & d_j\in \sigma(X_r)\ \mbox{iff}\  \pi(d_j) \in r^\I((t,V)^{i}),\\
       & d_j\in \sigma(X_D)\ \mbox{iff}\  \pi(d_j) \in D^\I.
   \end{array}
   $$
  Once $(*)$ is shown it is easy to see that $(t,V)^i\in \suc(\con)^\I = C^\I$. 
  In fact, the translation $\phi_\con$ of $\con$,
  where $r$ is replaced by $X_r$ and $D$ by $X_D$, is a conjunct in $\phi'_t$ and
  thus $\sigma$ satisfies $\phi_c$. Now $(*)$ shows that (modulo the application of the bijection $\pi$),
  when checking whether $(t,V)^i\in \suc(\con)^\I$,
  roles $r$ and concepts $D$ in $\phi_\con$ are interpreted in the same way as the set variables $X_r$ and $X_D$ in the solution $\sigma$ of $\phi'_t$.
  Thus the fact that $\sigma$ satisfies the conjunct $\phi_\con$ of $\phi'_t$ implies that the role successors of $(t,V)^i$ satisfy $\con$, i.e., $(t,V)^i\in \suc(\con)^\I$ holds.
  Note that, though $\phi'_t$ also contains set variables of the form $X_b$ for individual names $b$, this is not the case for $\phi_\con$
  since individuals occur only in the ABox and not in concepts.

  For role names $r$, property $(*)$ is immediate by the definition of $r^\I((t,V)^{i})$.
  Now consider a concept description $D$ such that $X_D$ occurs in $\phi'_t$. Then $D$ occurs in $\con$, and is thus
  smaller than $C$, which means that we can apply induction to it. If $d_j\in \sigma(X_D)$, then the Venn region
  $w$ to which $d_j$ belongs contains $X_D$ positively. Consequently, $S_w$ contains $D$, and
  the augmented type $(t^w,V^w)$ realizing $w$ satisfies $D\in t_w$. By induction, we obtain
  $\pi(d_j) = (t^w,V^w)^\ell \in D^\I$. Conversely, assume that $\pi(d_j) = (t^w,V^w)^\ell \in D^\I$, where
  $w$ is the Venn region to which $d_j$ belongs w.r.t.\ $\sigma$. By induction, we obtain $D\in t^w$, and thus
  the Venn region $w$ contains $X_D$ positively. Since $d_j$ belongs to this Venn region, we obtain
  $d_j\in \sigma(X_D)$.
\item
  The case where $C\not\in t$ can be treated similarly. In fact, in this case the constraint $\neg \con$
  is part of the QFBAPA formula $\phi'_t$ obtained from $t$, and we can employ the same argument as above, just
  using $\neg \con$ instead of $\con$.
\end{itemize}
\end{itemize}
This finishes the proof of the claim. As an easy consequence of this claim we have for all $C$ occurring in $\R$
that 
$$
\mbox{$C^\I = \{ (t,V)^i \mid C\in t, (t,V)\in\At,\ \mbox{and}\ 1\leq i \leq  \sigma_N(v_t)/n_t\}$.}
$$
Consequently, $|C^\I| = \sum_{t\in T_\At\text{ s.t. }C\in t}\sigma_N(v_t)$, which shows that
$\I$ satisfies $\R$ since $\sigma_N$ solves $\phi_{T_\At}$.

Finally, assume that $C(a)\in \A$. Then $a^\I = (t_a,V_a)^1$ and $C\in t_a$. The claim thus yields
$(t_a,V_a)^1\in C^\I$, which shows that $\I$ also satisfies the concept assertions in $\A$. 
\end{proof}

Next we show that the algorithm is also complete, i.e., whenever $\A$ is consistent w.r.t.\ $\R$, then
it succeeds on this input.

\begin{lemma}[Completeness]
Let $\R$ be a conjunctive $\ALCSCC$ ERCBox and $\A\neq \emptyset$ an $\ALCSCC$ ABox. 
If $\A$ is consistent w.r.t.\ $\R$, then Algorithm~\ref{alg:restr} succeeds on input $\R$ and $\A$.
\end{lemma}

\begin{proof}
Assume that $\I$ is a model of $\R$ and $\A$. Consider the set of all types of elements of $\I$, i.e.,
$
T_{\I} := \{t_\I(d) \mid d \in \Delta^\I\}, 
$
where 
$$
\begin{array}{r@{\ }c@{\ }l}
t_\I(d) &:=& \{D\in \M(\R,\A) \mid D\ \mbox{concept description and}\ d\in D^\I\}\cup\mbox{}\\
        &  & \{a\in \M(\R,\A)\cap N_I \mid a^\I = d\}\cup\{\neg a\in \M(\R,\A) \mid a\in N_I, a^\I \neq d\}.
\end{array}
$$
It is easy to see that the elements $t_\I(d)$ of $T_{\I}$ are indeed types.  In addition, for every $a\in \Ind_\A$,
there is exactly one type $t$ in $T_{\I}$ that contains $a$, which is $t_\I(a^\I)$. Also note that
$C(a)\in \A$ implies $a^\I\in C^\I$, and thus $C\in t_\I(a^\I)$. This shows that the types in $T_{\I}$
satisfy the conditions on the sets of augmented types computed in step (1) of the algorithm.
However, we still need to equip our types with Venn regions.

Consider $t := t_\I(d)$ for an element $d \in \Delta^\I$. We claim that the QFBAPA formula $\phi'_t$ corresponding to $t$ has
as solution the substitution $\sigma$ in which the universal set $\universe$ consists of all the role successors of $d$, and the other set
variables are assigned sets according to the interpretations of individuals, roles, and concept descriptions in the model $\I$.
The fact that $d\in C^\I$ for all concept descriptions $C\in t$ implies that $\sigma$ satisfies $\phi_t$, and
the fact that $\sigma(\universe)$   consists of all the role successors of $d$ implies that
$X_{r_1}\cup\ldots\cup X_{r_n} = \universe$ is also satisfied by $\sigma$.
The constraints $|X_b|\leq 1$ for $b\in\Ind_\A$ are satisfied since at most one role successors of $d$ can be equal
to $b^\I$.
If $a\in\Ind_\A$ belongs to $t$, then $t = t_\I(a^\I)$ and thus $d = a^\I$. If $r(a,b)\in \A$, then 
$b^\I$ is an $r$-successor of $d$ in $\I$, and thus $b^\I\in \sigma(X_b)\cap\sigma(X_r)$. This shows
that $\sigma$ also satisfies the cardinality constraints of the form $|X_{b}\cap X_{r}|\geq 1$ in $\phi_t'$.

Now, let $\{d_1,\ldots,d_m\} = \sigma(\universe)$ be the set of all role successors of $d$ in $\I$, and $w_i$ the Venn region to
which $d_i$ belongs w.r.t.\ $\sigma$.
By Lemma~1 in \cite{Baad19},
there is a solution $\sigma'$ of $\phi'_t$ such that the set $V$ of non-empty Venn regions
w.r.t.\ $\sigma'$ has cardinality $\leq N_t$ and each of these non-empty Venn regions in $V$ is one of the Venn regions
$w_i$, i.e., $V\subseteq \{w_1\ldots,w_m\}$. If $a\in\Ind_\A$ belongs to $t$ and $r(a,b)\in \A$, then
there is an $i$ such that $d_i = b^\I$. Note that the Venn region $w_i$ then belongs to $V$ since otherwise
$\sigma'$ could not be a solution of $|X_{b}\cap X_{r}|\geq 1$.

By construction, $(t,V)$ is an augmented type. Let $\At_\I$ denote
the set of augmented types obtained by extending the types in $T_\I$ in this way for every $d \in \Delta^\I$. By construction,
for every $t\in T_\I$ there is a set of Venn regions $V$ such that $(t,V)\in \At_\I$. 
It is easy to see that $\At_\I$ satisfies the conditions (\ref{cond:one}) and (\ref{cond:two}) considered
in the first step of the algorithm, and thus there is a set of augmented types $\At$ computed in this first step such that
$\At_\I\subseteq \At$. We now perform the other steps using $\At$ as a starting point.

First, note that no element of $\At_{\I}$ can be removed in Step~\ref{step:four} of our algorithm.
This is an easy consequence of the following observation. Let $T$ be a set of types such that
$T_{\I}\subseteq T$, and let $\phi_{T}$ be obtained from $\R$ by replacing each $|C|$ in $\mathcal{R}$ with 
$\sum_{t\in T\text{ s.t. }C\in t}v_t$ and adding $v_t\ge 0$ for each $t\in T$. Since $\I$ is a model of
$\R$, it is easy to see that $\phi_{T}$ has a solution that also satisfies $v_t\ge 1$ for all $t\in T_{\I}$.

Next, we show that the Venn regions occurring in some augmented type in $\At_\I$ are realized by $\At_\I$.
Thus, let $(t,V)$ be an augmented type constructed from a type $t = t_\I(d)$ as described above, and
let $w\in V$ be a Venn region occurring in this augmented type. Then there is a role successor
$d_i$ of $d$ such that $d_i$ belongs to the Venn region $w = w_i$ w.r.t.\ the solution $\sigma$ of $\phi'_t$
induced by $\I$. We know that $d_i\in D^\I$ for all $D\in S_w$, and thus
$S_w\subseteq t_\I(d_i)$. Since $\At_\I$ contains an augmented type with first component $t_\I(d_i)$, this
shows that $w$ is realized by $\At_\I$.

We claim that, in a the run of Algorithm~\ref{alg:restr}, 
we always have $\At_\I\subseteq \At$ and $T_\I\subseteq T_\At$.
Obviously, this is true 
when we enter the second step for the first time with the set $\At$ satisfying $\At_\I\subseteq \At$.
In addition, in Step~\ref{step:three} of our algorithm, no element of $\At_\I$ can be removed since
we have seen that the Venn regions occurring in some augmented type in $\At_\I$ are realized by $\At_\I$.
Finally, we have also seen above that, in Step~\ref{step:four} of our algorithm, no element
of $T_\I = T_{\At_\I}$ can be removed.

Since $\At_\I$ contains, for every $a\in\Ind_\A$, an augmented type $(t,V)$ such that $a\in t$,
the algorithm cannot fail.
This completes the proof of completeness.
\end{proof}

We have now proved that both the positive and the negative answers given by the algorithm are correct.
This allows us to show our \ExpTime complexity upper bound.

\begin{theorem}\label{exptime:thm}
Consistency of $\ALCSCC$ ABoxes w.r.t.\ $\ALCSCC$ ERCBoxes is an \ExpTime-complete problem.
\end{theorem}

\begin{proof}
Given an arbitrary, not necessarily conjunctive ERCBox $\R$, we consider all Boolean valuations of the
semi-restricted cardinality constraints occurring in $\R$, and collect those that evaluate the positive Boolean
structure of $\R$ to true. For each of these valuations $\rho$, we consider the conjunctive ERCBox $\R_\rho$
that is the conjunction of all the semi-restricted cardinality constraints evaluated to true by $\rho$.
There are exponentially many such conjunctive ERCBoxes $\R_\rho$, but each of them has a size that is linearly bounded by
the size of $\R$. In addition, $\R$ is satisfiable iff one of the conjunctive ERCBox $\R_\rho$ obtained this way
is satisfiable.

Thus, it remains to prove that
Algorithm~\ref{alg:restr} indeed runs in exponential time on conjunctive ERCBoxes. To see this, first note that, 
according to Lemma~\ref{aug:type:comp}, 
there are only exponentially many augmented types, and they can be computed in exponential time. 
In the first step, we first need to check whether condition (\ref{cond:two}) is satisfied for 
exponentially many augmented types. This can clearly be done in exponential time.
Then, we consider all possible ways of choosing, for every individual $a$, 
an appropriate augmented type. Since the number of individuals is polynomial
and for each one there are at most exponentially many augmented types containing this individual
in the first component, there are only exponentially many sets that can be generated by a combination
of these choices. 

For each of the sets generated in the first step, 
the iteration between the other two steps can happen only exponentially often since in
each iteration at least one augmented type is removed. A single Step~\ref{step:three} takes only
exponential time since for each of the exponentially many augmented types $(t,V)$, only exponentially
many other augmented types need to be considered. Finally, a single Step~\ref{step:four} takes only
exponential time. In fact, we need to consider exponentially many systems of linear inequalities
$\phi_{T_\At}\wedge v_t\ge 1$. Each of these systems may be of exponential size, but its solvability can
be tested in time that is polynomial in this size, 
and thus exponential in the size of the input.
Lemma~\ref{lem:ineq-system-properties} is applicable since adding $v_t\ge 1$ does not destroy the specific
form of the system required by the lemma.
\end{proof}

One might ask whether the approach used here to deal with individuals in ABoxes could also be used to treat
nominals in concept descriptions, where a nominal is a concept that must be interpreted as a singleton set.
As usual in Description Logic, we write such a nominal as $\{o\}$ where $o$ is an individual name.
The answer to the above question is, unfortunately, negative. From a technical point of view, the
claim in the proof of Lemma~\ref{soundness:RCBoxes} is no longer correct since for a nominal it only
holds for $i =1$, but not for $i>1$. However, in the induction assumption we would need this for arbitrary $i$ and
not just for $i = 1$. Using a reduction from \cite{Tobi01a}, it is actually easy to see that adding nominals
increases the complexity of ERCBox consistency from \ExpTime to \NExpTime even for $\ALC$. As usual, we use $\mathcal{O}$ in the name
of the DL to indicate the presence of nominals.

\begin{proposition}
Consistency of conjunctive $\ALCO$ ERCBoxes is \NExpTime-complete.
\end{proposition}

\begin{proof}
Membership in \NExpTime follows from the fact that $\ALCO$ ERCBox can be expressed
using $\ALC$ ECBoxes, whose consistency problem was shown to be in \NExpTime in \cite{BaEc17}.

In \cite{Tobi01a}, Tobies has shown that consistency of $\ALCQ$ CBoxes is \NExpTime-hard,
using a reduction from a bounded tiling problem. Looking closer at this reduction, one sees that
actually only $\ALC$ concept descriptions, $\ALC$ CIs, and cardinality restrictions of the forms
$({\geq}\,1\,C)$, $({\leq}\,1\,C)$, and $({\leq}\,2^n{\cdot}2^n\,C)$ for $\ALC$ concepts $C$ are needed.
CIs and cardinality restrictions $({\geq}\,1\,C)$ can easily be expressed using semi-restricted cardinality constraints,
as introduced in Definition~\ref{RCBox:def}. Using a new nominal $\{o\}$, we can express $({\leq}\,1\,C)$
as $|C|\leq|\{o\}|$. To express $({\leq}\,2^n{\cdot}2^n\,C)$, we need a new nominal and additional
auxiliary new concept names: the constraints
$$
|A_0|\leq|\{o\}|\wedge |A_1|\leq 2|A_0|\wedge \ldots\wedge |A_{2n}|\leq 2|A_{2n-1}|
$$
ensures that the cardinality of $A_{2n}$ is bounded by $2^{2n} = 2^n{\cdot}2^n$.
\end{proof}

\section{Undecidability of~\hmath$\ALCIplus$}

We next observe that a seemingly harmless extension of $\ALCplus$ turns the satisfiability problem undecidable.
We obtain $\ALCIplus$ by adding \emph{role inverses} to $\ALCplus$ by additionally allowing expressions of the form $r^-$ for any $r \in N_R$ in all places where role names are allowed to occur. The semantics of the expression $r^-$ is defined by $(r^-)^\I = \{(e,d) \mid (d,e)\in r^\I\}$. The key insight for showing our result is that adding this feature enables us to encode multiplication of concept extensions, allowing for a reduction from Hilbert's tenth problem. We first provide an example illustrating how ``class extension multiplication'' can be expressed.

\begin{example}
In order to express that the cardinality of a concept $C$ coincides with the product of the cardinalities of concepts $A$ and $B$, we employ two auxiliary roles $r$ and $s$. We first enforce that role $r$ connects precisely each member of $A$ with every member of $B$:
$$
A \equiv \exists r.\top
\quad\quad
B \equiv \exists r^-.\top
\quad\quad
A \sqsubseteq sat(B=r)
\quad\quad
B \sqsubseteq sat(A=r^-)
$$
Next, we make sure that every domain element has precisely as many outgoing $r$ roles as outgoing $s$ roles:
$$\top \sqsubseteq sat(|r|=|s|)$$
Moreover, the elements with incoming $s$ roles are precisely the instances of concept $C$:
$$C \equiv \exists s^-.\top$$
Finally, no element can have more than one incoming $s$ role (in other words, $s$ is inverse functional):
$$\top \sqsubseteq sat(|s^-| \leq 1)$$ 
\end{example}

A construction very much along the lines of the given example allows us to express Hilbert's tenth problem as an $\ALCIplus$ concept satisfiability problem and hence establish undecidability of the latter. 

\begin{theorem}
Satisfiability of $\ALCIplus$ concept descriptions is undecidable.
\end{theorem}
\begin{proof}
We show the claim via a reduction from Hilbert's tenth problem, i.e., the solvability of Diophantine equations.
Note that any Diophantine equation $D$ can be transformed (possibly introducing fresh auxiliary variables) into a system $\E$ of equations, where each equation has one of the following three forms: (i)~$x=y \cdot z$, (ii)~$x= y + z$, or (iii)~$x = n$, for a natural number $n$, such that $D$ has an integer solution if and only if $\E$ has a solution in the natural numbers.

Given such a system $\E$ of equations over a set $Var$ of variables, we now construct an $\ALCIplus$ concept expression $C_\E$ containing concept names $A_x$ for all variables $x$ occurring in $\E$, such that satisfiability of $C_\E$ coincides with the existence of a natural solution for $\E$. 
We let $C_\E = \bigsqcap_{eq\in \E} C_{eq}$, where $C_{eq}$ stands for
\begin{itemize}
\item the concept expression $C^\mathrm{1}_{eq} \sqcap C^\mathrm{2}_{eq} \sqcap C^\mathrm{3}_{eq} \sqcap C^\mathrm{4}_{eq}$ if $eq$ is of the form $x=y \cdot z$, where
\begin{itemize}
\item 
$C^\mathrm{1}_{eq} = sat\big(\neg A_y  \subseteq  sat(|s_{eq}| {=} 0)\big)$,
\item 
$C^\mathrm{2}_{eq} = sat\big(A_y  \subseteq  sat(|s_{eq}| {=} |A_z|)\big)$,
\item 
$C^\mathrm{3}_{eq} = sat\big(\top \subseteq  sat(|s_{eq}^-| {\leq} 1)\big)$, and
\item 
$C^\mathrm{4}_{eq} = sat\big(A_x = sat(|s_{eq}^-|{\geq} 1)\big)$.
\end{itemize}
\item $sat(|A_x| = |A_y|+|A_z|)$ if $eq$ is of the form $x= y + z$, 
\item $sat(|A_x| = n)$ if $eq$ is of the form $x = n$.
\end{itemize}
We now show that $\E$ has a solution in the natural numbers if and only if $C_\E$ is satisfiable.

For the ``if'' direction, assume there is some finite interpretation $\I$ and domain element $d \in \Delta^\I$ such that $d \in C_\E^\I$. Let $\sigma: Var \to \mathbb{N}$ be the variable assignment mapping every variable $x$ in $\E$ to $|A_x^\I|$. Then, clearly, $\sigma$ maps every equation $eq$ of the form (ii) or (iii) to a true statement due to $\delta \in C_{eq}^\I$. Now consider some equation $x=y+z$ of the form (i).
For this, we obtain  
$$
\begin{array}{ll}
|A_x^\I| & = |(sat(|s^-|{\geq} 1))^\I| \mbox{ due to $C^4_{eq}$}\\
& = |\{e \mid (e',e) \in s_{eq}^\I \}| \\
& = |\{(e',e) \mid (e',e) \in s_{eq}^\I \}| \mbox{ due to $C^3_{eq}$}\\ 
& = |s_{eq}^\I| \\ 
& = \sum_{e' \in \Delta^\I} |\{e \mid (e',e) \in s_{eq}^\I \}| \\
& = \sum_{e' \in A_y^\I} |\{e \mid (e',e) \in s_{eq}^\I \}| \mbox{ due to $C^1_{eq}$}\\
& = \sum_{e' \in A_y^\I} |A_z^\I| \mbox{ due to $C^2_{eq}$}\\
& = |A_y^\I| \cdot |A_z^\I|,\\
\end{array}
$$
which finishes the proof of the ``if'' direction.

For the ``only if'' direction, let $\sigma: Var \to \mathbb{N}$ be a variable mapping satisfying all equations in $\E$. We now construct a model $\I$ of $C_{eq}$ as follows:
\begin{itemize}
\item $\Delta^\I = \{n \in \mathbb{N} \mid 1 \leq n\leq \max_{v \in Var} \sigma(v)\}$ 
\item $A_v^\I = \{n \in \mathbb{N} \mid 1 \leq n \leq \sigma(v)\}$
\item $s_{x=y\cdot z}^\I = \{(i,k\cdot \sigma(z)+i) \mid 0 \leq k \leq \sigma(y)-1,\ 1 \leq i \leq \sigma(z)\}$
\end{itemize}
It is straightforward to check modelhood of $\I$. This concludes the ``only if'' direction and hence the proof.
\end{proof}

\section{Query entailment in~\hmath$\ALCplus$}

The final result of this section is the undecidability of conjunctive query entailment for $\ALCplus$. To this end, we first briefly recap the notion of (Boolean) conjunctive queries and define query entailment. 

\medskip

In queries, we use \emph{variables} from a countably infinite set $\Vlang$.
A Boolean \emph{conjunctive query} (CQ) $q$ is a finite set of atoms of the form $r(x,y)$ or $C(z)$, where $r$ is a role, $C$ is concept, and $x,y,z \in \Vlang$.
A CQ $q$
is \emph{satisfied} by $\I$ (written: $\I \models q$) if there is a \emph{variable assignment} $\pi:\Vlang\to\Delta^\I$ (called \emph{match}) such that $\tuplei{\pi(x),\pi(y)}\in r^\I$ for every $r(x,y) \in q$ and $\pi(z)\in C^\I$ for every $C(z) \in q$.
A CQ $q$ is \emph{(finitely) entailed} from a knowledge base $\mathcal{K}$ (written: $\mathcal{K} \models q$) if every (finite) model $\I$ of $\mathcal{K}$ satisfies $q$.

\medskip

We actually show undecidability of CQ entailment for a much weaker logic, thereby providing a very restricted fragment of constant-free and equality-free two-variable first-order logic for which finite CQ entailment is already undecidable, significantly strengthening and solidifying earlier results along those lines~\cite{PrHa09}. Our proof makes use of deterministic Turing machines (DTMs). For our purposes, it is sufficient to consider only computations starting with an empty tape. For space reasons, we assume the reader to be familiar with standard notions and constructions concerning DTMs. We call a DTM \emph{looping} if its run starting contains repeating configurations,i.e., there are two different (and hence -- due to determinism -- infinitely many) points in time, where the machine's tape content, head position, and state are the same. 
It is easy to see that the problem of determining if a given TM is looping is undecidable. 

\medskip

We show our undecidability result for the DL $\ALC{^\mathrm{cov}}$, a slight extension of $\ALC$ by \emph{role cover axioms} of the form $\mathrm{cov}(r,s)$ for role names $r$ and $s$. An interpretation $\mathcal{I}$ satisfies $\mathrm{cov}(r,s)$ if $r^\mathcal{I} \cup s^\mathcal{I} = \Delta^\mathcal{I} \times \Delta^\mathcal{I}$.
Role cover axioms can be expressed in $\ALCplus$ via $sat\big(\top \subseteq sat (|r \cup s|=|\mathcal{U}|)\big)$, hence $\ALC{^\mathrm{cov}}$ is subsumed by $\ALCplus$.

\medskip

In what follows, assume that a DTM $\mathcal{M}$ is given.
We now describe an $\ALC{^\mathrm{cov}}$ TBox $\mathcal{T}$ and conjunctive query $ \query$ such that $\mathcal{T} \models  \query$ exactly if $\mathcal{M}$ is not looping.
We provide $ \query$ and $\mathcal{T}$ together with the underlying intuitions of our construction.
The goal of our construction is that a countermodel (i.e., an interpretation satisfying $\mathcal{T}$ but not $ \query$) corresponds to a looping configuration sequence of $\mathcal{M}$. Thereby, the domain elements represent tape cells at certain computation steps of $\mathcal{M}$. The role $h$ connects consecutive tape cells of the same configuration, whereas the role $v$ connects a configuration's tape cell with the same tape cell of the successor configuration.       

We start by providing the query. Intuitively, the query is meant to catch the unwanted situation that two corresponding tape cells of consecutive configurations are $v$-connected, but the cells to their right aren't.   
\begin{equation}
\query = \exists x,y,x',y'. v(x,y) \wedge h(x,x') \wedge h(y,y') \wedge \overline{v}(x',y')
\label{eqn:query}
\end{equation}
We proceed by giving the axioms of $\mathcal{T}$.
The following covering axiom ensures that, whenever two elements are not $v$-connected, they must be $\overline{v}$-connected. This is needed to enable the above query to catch the described problem.  
\begin{equation}
\mathrm{cov}(v,\overline{v})
\end{equation}
The remaining TBox axioms can be found in Table~\ref{TMaxioms}.
\begin{table}[t]
\caption{TBox axioms for DTM implementation \label{TMaxioms}}
\vspace{-3ex}
\begin{equation}
\top \sqsubseteq \exists aux.(\mathit{TapeStart} \sqcap \mathit{InitConf} \sqcap \mathit{State}_{ \state_\mathrm{ini}})
\label{eqn:origin}
\end{equation}
\begin{equation}
\top \sqsubseteq \exists h.\top \sqcap \exists v.\top
\label{eqn:succs}
\end{equation}
\begin{equation}
\mathit{TapeStart} \sqsubseteq \forall v.\mathit{TapeStart}
\label{eqn:tapestart}
\end{equation}
\begin{equation}
\mathit{InitConf} \sqsubseteq \forall h.\mathit{InitConf} \quad\quad
\mathit{InitConf} \sqsubseteq\mathit{Symbol}_\Box
\label{eqn:initconf}
\end{equation}
\begin{eqnarray}
& & \hspace{-5ex}\mathit{State}_ \state \sqsubseteq\forall h.\mathit{NoHeadR} \quad\ \   
\mathit{NoHeadR} \sqsubseteq \forall h.\mathit{NoHeadR} \quad  
\mathit{NoHeadR} \sqsubseteq \mathit{NoHead} 
\label{eqn:uniquehead1}\\
& & \hspace{-8.3ex} \exists h.State_ \state \sqsubseteq \mathit{NoHeadL} \quad\quad\!  
\exists h.\mathit{NoHeadL} \sqsubseteq \mathit{NoHeadL} \quad\quad  \ \; 
\mathit{NoHeadL} \sqsubseteq \mathit{NoHead} 
\label{eqn:uniquehead2}\\
& & \hspace{28.3ex}\mathit{State}_ \state \sqcap \mathit{NoHead} \sqsubseteq \bot 
\label{eqn:uniquehead3}
\end{eqnarray}
\begin{equation}
Symbol_\sigma \sqcap\mathit{Symbol}_{\sigma'} \sqsubseteq \bot \quad\quad\quad
State_ \state \sqcap\mathit{State}_{ \state'} \sqsubseteq \bot 
\label{eqn:exclusive}
\end{equation}
\begin{equation}
NoHead \sqcap\mathit{Symbol}_\sigma \sqsubseteq \forall v.Symbol_\sigma
\label{eqn:inertia}
\end{equation}
\begin{eqnarray}
State_ \state \sqcap\mathit{Symbol}_\sigma  &  \sqsubseteq  & \forall v.(Symbol_{\sigma'} \sqcap \forall h.State_{ \state'})
\label{eqn:action1}
\\ 
\exists h.(State_ \state \sqcap\mathit{Symbol}_\sigma) &   \sqsubseteq  & \forall v.(State_{ \state'} \sqcap \forall h.Symbol_{\sigma'})
\label{eqn:action2}
\\ 
\mathit{TapeStart} \sqcap\mathit{State}_ \state \sqcap \mathit{Symbol}_\sigma &   \sqsubseteq  & \forall v.(State_{ \state'} \sqcap\mathit{Symbol}_{\sigma'})
\label{eqn:action3} 
\end{eqnarray}
\end{table}
Axiom~\ref{eqn:origin} ensures (by means of an auxiliary role $aux$ which serves no further purpose) that there is a first tape cell of the first (initial) configuration where the head of the TM is positioned in the initial state.
Axiom~\ref{eqn:succs} enforces that for every cell of every configuration there is both a tape cell to its right and a corresponding tape cell in the successor configuration.
Axiom~\ref{eqn:tapestart} makes sure that, for every cell that is the first on its tape, the corresponding successor configuration's tape cell is also the first.
Axioms~\ref{eqn:initconf} propagates the information that a cell belongs to the initial configuration  along the tape, and fills the tape with blanks.
Axioms~\ref{eqn:uniquehead1}--\ref{eqn:uniquehead3} (instantiated for every state~$\state$) make sure that in every configuration there can only be one cell where the head is positioned. 
Every cell can only carry one symbol and the head can be in only one state, as ensured by Axioms~\ref{eqn:exclusive} (for distinct symbols $\sigma, \sigma'$ and distinct states $ \state,  \state'$). 
Thanks to Axiom~\ref{eqn:inertia}, symbols on head-free cells carry over to the next configuration.
As specified by the DTM's transition function, the head reads a symbol $\sigma$, writes a symbol $\sigma'$, changes its state from $ \state$ to $ \state'$ and moves
right (Axiom~\ref{eqn:action1}) or left (Axiom~\ref{eqn:action2}) or stays in its place whenever it is supposed to move left but is already at the leftmost tape cell (Axiom~\ref{eqn:action3}).
This finishes the description of the TBox $\mathcal{T}$, allowing us to establish the claimed property 
and consequenty the undecidability result.

\begin{proposition}
$\mathcal{M}$ is looping iff 
there is a finite model $\mathcal{I}$ of $\mathcal{T}$ with $\mathcal{I} \not\models Q$.\label{prop:finitemodel}
\end{proposition}

\begin{theorem}
Finite CQ entailment over $\ALC{^\mathrm{cov}}$ TBoxes is undecidable.
\end{theorem}

\begin{proof}
According to Proposition~\ref{prop:finitemodel}, the TM looping problem can be reduced to the problem if for a given $\ALC{^\mathrm{cov}}$ TBox $\mathcal{T}$ and conjunctive query $ \query$, there is a finite interpretation $\mathcal{I}$ with $\mathcal{I} \models \mathcal{T}$ with $\mathcal{I} \not\models  \query$. Note that the latter is the case exactly if $\mathcal{T}$ does not finitely entail $ \query$.
\end{proof}

Finally, taking into account that \ALCplus subsumes $\ALC{^\mathrm{cov}}$ and only allows for finite models, we obtain the wanted result.

\begin{corollary}
Conjunctive query entailment for \ALCplus is undecidable. 	
\end{corollary}

\section{Decidable querying for~\hmath$\ALCSCC$}

In stark contrast to the undecidability result just presented, we prove that conjunctive query entailment by~~$\ALCSCC$ ABoxes w.r.t.~$\ALCSCC$ ERCBoxes
is only~\ExpTime-complete, thus not harder than deciding knowledge base consistency for plain~$\ALC$.

Our result employs a construction by Lutz~\cite{Lutz08}, but careful and non-trivial argumentation is needed to show that the idea, conceived for arbitrary models, carries over to our finite-model case.
The approach reduces entailment of some CQ $q$ to an exponential number of \ExpTime inconsistency checks in the spirit of Theorem~\ref{exptime:thm}, resulting in an overall \ExpTime procedure.  
In their entirety, these mentioned checks verify if some model exists that does not admit any matches of $q$ having a specific, forest-like shape. 

It remains to argue that these specific, forest-shaped query matches of~$q$ are the only ones that matter for checking entailment. To this end, we show that all other matches can be ``removed'' by a model transformation consisting of the following three consecutive steps: (i) forward-unraveling, resulting in possibly-infinite structures (in Section~\ref{sec:unrav}) then (ii) cautious collapsing to regain 
finiteness while keeping the model ``forest-like enough'' for small conjunctive queries to match only in a tree-shaped way (in Section~\ref{sec:loose})
and finally (iii) enriching the model by copies of domain elements to again satisfy the global counting constraints which had possibly become violated in the course of the previous steps  (in Section~\ref{sec:dupl}).

To the end of this Section let~$\kb_0 = (\abox_0, \tbox_0, \ercbox_0)$ be an~$\ALCSCC$ 
knowledge base composed of an ABox~$\abox_0$, a Tbox~$\tbox_0$ and an ERCBox~$\ercbox_0$.
Without loss of generality we will assume that~$\kb_0$ is \emph{normalized}, i.e. all concepts appearing in~$\tbox_0$
are of depth at most one and all concepts occurring in~$\abox_0$ and~$\ercbox_0$ are atomic. This can be done via a routine transformations.

\subsection{The construction of sufficiently tree-like models} \label{subsec:nacyclic}

We start with some preliminary definitions on morphism, neighbourhoods and bisimulations.

\paragraph*{Morphisms.}
A \emph{homomorphism} from an interpretation~$\I$ to an interpretation~$\J$ is a function~$\homo : \I \rightarrow \J$ 
satisfying for all concept names~$A$ and all role names~$r$ the following properties: 
if~$d \in A^{\I}$ then~$\homo(d) \in A^{\J}$ and if~$(d,d') \in r^{\I}$ then~$\big( \homo(d), \homo(d') \big) \in r^{\J}$.
An \emph{isomorphism} is a bijection~$\izo$ such that both~$\izo$ and~$\izo^{-1}$ are homomorphisms.

\paragraph*{Neighbourhoods.} For a given interpretation~$\I$ and an element~$d \in \DI$ we denote with~$\SuccI{\I}{d}$ 
the set of role successors of~$d$, i.e. the set~$\bigcup_{r \in N_R} \{ d' : (d,d') \in r^{\I}\}$. 
Note that it is possible that~$d \in \SuccI{\I}{d}$. The \emph{forward neighbourhood} (or simply 
\emph{neighbourhood})~$\Neib{\I}{d}$ of~$d$ is the interpretation~$\Neib{\I}{d} = (\Delta^{\Neib{\I}{d}}, \cdot^{\Neib{\I}{d}})$ 
such that~$\Delta^{\Neib{\I}{d}} = \SuccI{\I}{d} \cup \{ d \}$, $A^{\Neib{\I}{d}} = A^{\I} \cap \Delta^{\Neib{\I}{d}}$ for any 
concept name~$A \in N_C$ and~$r^{\Neib{\I}{d}} = r^{\I} \cap (\{ d \} \times \Delta^{\I})$ for any role name~$r \in N_R$.

The next definition introduces a notion of bisimulation tailored to normalized~$\ALCSCC$ kbs.
\begin{definition} \label{def:bisim}
Let~$\I, \J$ be interpretations with~$d \in \I, d' \in \J$. We say that~$d$ and~$d'$ are \emph{forward-neighbourhood bisimilar} 
(or simply \emph{bisimilar}), denoted with~$d \eqbisim d'$, if there exist a function~$\izo : \Neib{\I}{d} \rightarrow \Neib{\J}{d'}$ 
(called \emph{bisimulation}) satisfying the following conditions:
\begin{itemize}
	\item $\izo : \restr{\Neib{\I}{d}}{\SuccI{\I}{d}} \rightarrow \restr{\Neib{\J}{d'}}{\SuccI{\J}{d'}}$ is a bijection, and
	\item For all~$d' \in \Neib{\I}{d}$, for all concept names~$A \in N_C$ and all role names~$r \in N_R$ equivalences
		$d' \in A^{\I} \Leftrightarrow \izo(d') \in A^{\J}$ and~$(d,d') \in r^{\I} \Leftrightarrow (\izo(d), \izo(d')) \in r^{\J}$ hold.
\end{itemize}
\end{definition}

The following observation simplifies most of the forthcoming proofs. It can be either shown by a straightforward 
structural induction over the shape of~$\ALCSCC$ concepts or deduced from Proposition 2 from~\cite{BaaderB19}, 
where the notion of~$\ALCQt$--bisimulation was developed.
\begin{observation}\label{obs:neighbourhood}
Let~$\I \models \kb$ be a model of a normalized~$\ALCSCC$ knowledge base~$\kb$. 
For any two domain elements~$d,d' \in \DI$, if~$d$ and~$d'$ are bisimilar then they satisfy the same~$\ALCSCC$ concepts of depth at most one.
\end{observation}
\subsubsection{Forward-unravelings of finite models} \label{sec:unrav}

For a finite interpretation~$\I$ with~$\DInamed$ we denote those elements~$d \in \DI$ for which~$a^\I = d$ holds for some individual name~$a \in \indA$.

\begin{definition} \label{def:unrav}
Let~$\I$ be a \emph{finite} interpretation. We define a forward-unraveling~$\Iunrav = (\DIunrav, \cdotIunrav)$ of~$\I$ 
as a (potentially infinite) interpretation satisfying the following conditions:
\begin{itemize}
\item $\DIunrav = (\DI)^+ \; \setminus \; \big( \DInamed \cdot \DInamed \cdot (\DI)^* \big)$\\ 
In words,~$\DIunrav$ consists of all nonempty sequences of elements from~$\DI$ except those, 
where the first two elements are named in~$\I$.
\item For any~$a \in \indA$, let~$a^{\Iunrav} = a^{\I}$, i.e.~$a$ is interpreted 
by the one-element sequence consisting of the named element~$a^{\I}$ 
from~$\I$.\footnote{For convenience, we will not syntactically distinguish elements 
from~$\Delta^\I$ and one-element sequences from~$\DIunrav$; in particular this means~$\DI \subseteq \DIunrav$.}
\item
For concept names~$A$, we let~$A^{\Iunrav} = \{w \mid \last{w} \in A^{\I}\}$, where 
for a given element~$w \in \DIunrav$ we use~$\last{w}$ to denote the last\footnote{We define~$\first{w}$ analogously.}~$d \in \DI$ in the sequence~$w$.
\item For role names~$r$, we let~$r^{\Iunrav} =  r^\I \cap (\DInamed \times \DInamed ) \cup \{ (w,wd) \mid (\last{w},d) \in r^\I \}$. 
\end{itemize}
\end{definition}

The notion of forward-unravelings differs only slightly from the classical notion of unraveling. 
The only difference is that the sequences starting from two named individuals are excluded 
from the domain and that roles linking named individuals are assigned manually by the last item 
from Definition~\ref{def:unrav}. It is not surprising that forward-unravellings preserve satisfaction 
of~$\ALCSCC$ Aboxes and Tboxes as well as conjunctive query non-entailment. The proof is standard and 
hinges on the fact that~$w \in \DIunrav$ and~$\last{w} \in \DI$ satisfy the same~$\ALCSCC$ concepts. 
For CQ non-entailment it is enough to see that~$\last{\cdot}$ is a homomorphism from~$\Iunrav$ to~$\I$.

\begin{lemma} \label{lem:unravabox}
For any normalized ABox~$\abox$ and any finite interpretation~$\I$, if~$\I \models \abox$ holds, then also~$\Iunrav \models \abox$ holds.
\end{lemma}
\begin{proof}
Take an arbitrary normalized ABox~$\abox$ as well as arbitrary finite interpretation~$\I$. Assume that~$\I \models \abox$ holds.
Note that~$\DInamed = \DIunravnamed$ holds since we agreed that we will not 
syntactically distinguish elements from $\DI$ and one-element sequences. 
First, see that satisfaction of assertions of the form~$A(a) \in \abox$ is 
guaranteed due to the third point of Definition~\ref{def:unrav} and the fact that the 
property~$\last{w} = w$ holds for any~$w \in \DIunravnamed$. Second, we can conclude 
that any assertion of the form~$r(a,b) \in \abox$ is also satisfied in~$\Iunrav$, 
due to the last item of Definition~\ref{def:unrav}, more precisely the 
fact that~$r^\I \cap (\DInamed \times \DInamed ) \subseteq r^{\Iunrav}$ holds. 
Hence~$\Iunrav \models \abox$.
\end{proof}

An important step towards proving that forward unravelings preserve normalized~$\ALCSCC$ TBoxes is to show that 
any sequence~$w \in \DIunrav$ is forward-bisimilar to~$\last{w} \in \DI$, i.e, the element from which~$w$ originated.
\begin{lemma} \label{lem:bisimlast}
Let~$\kb = (\abox, \tbox, \ercbox)$ be a normalized~$\ALCSCC$ knowledge base and let~$\I$ be its arbitrary finite model.
Then for all domain elements~$d \in \DI$ and all sequences~$w \in \DIunrav$ the implication~$d = \last{w} \Rightarrow d \eqbisim w$ holds.
\end{lemma}
\begin{proof}
We define a function~$\izo : \Neib{\Iunrav}{w} \rightarrow \Neib{\I}{d}$, which maps the 
neighbourhood of~$w$ in~$\Iunrav$ to the neighbourhood of~$d$ in~$\I$, as~$\izo(x) = \last{x}$.
The definition of~$\izo$ is sound, since~$\last{d}$ is defined uniquely for each sequence from~$(\DI)^+$.
Moreover, see that~$\izo^{-1}:\Neib{\I}{d} \rightarrow \Neib{\Iunrav}{w}$
is defined as~$\izo^{-1}(x) = x $ for named individuals and~$\izo^{-1}(x) = wx$ otherwise,
which is also sound due to the second and the last item of Definition~\ref{def:unrav}.

We will first show that~$\izo : \restr{\Neib{\Iunrav}{w}}{\SuccI{\Iunrav}{w}} \rightarrow \restr{\Neib{\I}{d}}{\SuccI{\I}{d}}$ is a bijection. 
One can show it by proving that equations~$\izo \circ \izo^{-1} = \textit{id} = \izo^{-1} \circ \izo$ hold, 
where~$\textit{id}$ is the identity function and~$\circ$ is a function-composition operator. 
Take an arbitrary element~$w$ from~$\Neib{\Iunrav}{w}$ and assume that both~$w, w'$ are named.
Then~$w = \last{w}$ and~$w' = \last{w'}$ (since we identify named individuals with one-element sequences) 
and the following equations hold:
\[
\izo(\izo^{-1}(w')) = \izo(w') = \last{w'} = w' = \izo^{-1}(w') = \izo^{-1}(\last{w'}) =  \izo^{-1}(\izo(w')).
\]
Now assume that one of~$w,w'$ is not named. Then~$w'$ is in the form~$w' = we$ 
and the presented equations~$\izo \circ \izo^{-1} = \textit{id} = \izo^{-1} \circ \izo$ hold again, as it is written below:
\[
\izo(\izo^{-1}(e)) = \izo(we) = \last{we} = e \; \text{and} \; we = \izo^{-1}(e) = \izo^{-1}(\last{we}) = \izo^{-1}(\izo(we)).
\]
Hence~$\izo$ restricted to role successors of~$w$ is a bijection. 
Note that for any atomic concept~$A$ we know that~$w \in A^{\Iunrav}$ holds iff~$d \in A^{\I}$ holds, 
due to the third item of Definition~\ref{def:unrav} (and since~$d = \last{w} = \izo(w))$. Thus, the only thing which remains to be
done is to show that for all~$w' \in \Neib{\Iunrav}{w}$ the equivalence~$(w,w') \in r^{\Iunrav} \Leftrightarrow (\izo(w),\izo(w')) \in r^{\Iunrav}$ holds.

Let us fix an arbitrary neighbour~$w'$ of~$w$, i.e., a domain element~$w \in \DIunrav$
s.t.~$(w,w') \in r^{\Iunrav}$ holds for some role name~$r$. Let~$d' = \last{w'} = \izo(w')$ be the corresponding element in~$\DI$.

We distinguish two cases.
\begin{itemize}
\item~$w,w'$ are not named.\\
Since we agreed that~$\DIunravnamed = \DInamed$ holds,
we infer that~$d = \izo(w) = w$ and~$d' = \izo(w') = w'$. Thus we can use the last item of Definition~\ref{def:unrav}, 
namely the part stating that~$r^\I \cap (\DInamed \times \DInamed ) = r^{\Iunrav} \cap (\DInamed \times \DInamed )$ 
and conclude the mentioned property. 

\item At least one of~$w,w'$ is not named.\\ 
In this case,
from the second part of the third item of Definition~\ref{def:unrav} we know that~$w'$ is actually a sequence in the form~$w \cdot e$.
But from the same definition as above,~$(w,w') = (w, we) \in r^{\Iunrav}$ holds if and only if~$(\last{w}, e) = (d, e) \in r^{\I}$
holds, which is exactly what we wanted to prove. 
\end{itemize}

Since we have shown preservation (and non-preservation) of atomic concepts 
and roles by~$\izo$ and since~$\izo$ is a bijection, we infer that~$\izo$ is a bisimulation. Hence~$w \eqbisim d$ holds.
\end{proof}

As an immediate consequence of Lemma~\ref{lem:bisimlast} we obtain that any two sequences~$w,w' \in \DIunrav$ 
having the same last element are forward-bisimilar, as stated below.
\begin{lemma}
For any finite interpretation~$\I$ being a model of a normalized~$\ALCSCC$ knowledge base~$\kb$
and any sequences~$w, w' \in \DIunrav$ with~$\last{w} = \last{w'}$, the property~$w \eqbisim w'$ holds.
\end{lemma}
\begin{proof}
By applying Lemma~\ref{lem:bisimlast} to~$w$ and~$w'$, we infer that~$w \eqbisim \last{w}$ and~$w' \eqbisim \last{w'}$ holds.
Since the elements~$\last{w}$ and~$\last{w'}$ are equal, we conclude that~$w$ is bisimilar to~$w'$.
\end{proof}

Once we have shown that~$w \eqbisim \last{w}$ for any~$w \in \DIunrav$, we can employ this fact to show that forward-unraveling 
preserve satisfaction of normalized TBoxes.
\begin{lemma} \label{lem:unravtbox}
For any normalized~$\ALCSCC$ TBox~$\tbox$ and any finite interpretation~$\I$, the implication~$\I \models \tbox \Rightarrow \Iunrav \models \tbox$ holds.
\end{lemma}
\begin{proof}
Let~$w \in \DIunrav$ be an arbitrary domain element from~$\Iunrav$ and let~$d = \last{w}$ be the corresponding element from~$\DI$.
Let~$\varepsilon = C_0 \sqsubseteq C_1$ be an arbitrary GCI from the TBox~$\tbox$. Note that~$C_0, C_1$ are not 
necessary atomic, but since we restricted our attention to normalized knowledge bases only, 
we can assume that~$C_0$ and~$C_1$ are~$\ALCSCC$ concepts of depth at most one. Assume that~$w \in C_0^{\Iunrav}$ holds.
Then, to prove that~$\Iunrav \models \varepsilon$ holds, we need to show that~$w \in C_1^{\Iunrav}$ holds. 
Since~$w \eqbisim d$ holds (by Lemma~\ref{lem:bisimlast}), from Observation~\ref{obs:neighbourhood} we know that~$d$ 
and~$w$ satisfy the same~$\ALCSCC$ concepts of depth at most one. Hence~$d \in C_0^{\I}$. 
From the fact that~$\I$ satisfies~$\varepsilon$ we infer that~$d \in \C_1^{\I}$ holds.
Again, since~$d$ and~$w$ are bisimilar, they satisfy the same~$\ALCSCC$ concepts of depth~$\leq 1$ and thus~$w \in \C_1^{\Iunrav}$ holds too.
Due to the fact that~$w$ and~$\varepsilon$ were arbitrarily chosen, we conclude that~$\Iunrav \models \tbox$ holds.
\end{proof}

From the construction of forward unravelings one can immediately 
see that it also preserves non-entailment of conjunctive queries.
Without loss of generality we can always assume that CQs contains only atomic concepts 
(e.g. by introducing a fresh name~$A_C$ for each concept~$C$ and putting the GCI~$C \equiv A_C$ inside the TBox).

\begin{lemma} \label{lem:unravnonentail}
For any finite interpretation~$\I$ and any conjunctive query~$q$, if~$\I \not\models q$ holds then~$\Iunrav \not\models q$ holds too.
\end{lemma}
\begin{proof}
Assume that~$\I \not\models q$ holds but~$\Iunrav$ entails~$q$. Then there exists 
a match~$\pi$ of~$q$ on~$\Iunrav$. Note that~$\homo(x) = \last{x}$ is a homomorphism~$\Iunrav$ to~$\I$.
Indeed, the preservation of atomic concepts by~$\homo$ can be deduced from the third item of Definition~$\ref{def:unrav}$, and
the fact that if~$(d,d') \in r^{\Iunrav}$ holds then~$(\homo(d), \homo(d')) \in r^{\I}$ holds can be inferred from 
the last item of Definition~$\ref{def:unrav}$. 
However, in that case~$\pi'$ with~$\pi'(x) = \homo(\pi(x))$ 
would be a match of~$q$ on~$\I$, which contradicts the initial assumption~$\I \not\models q$. Thus~$\Iunrav \not\models q$ holds.
\end{proof}

\subsubsection{Loosening of finite unravelings} \label{sec:loose}

Unraveling removes non-forest-shaped query matches, however, $\Iunrav$ does not need to be finite even if $\I$ is. To regain finiteness without re-introducing query matches, we are going to introduce the notion of~\emph{$k$-loosening}.

For a given finite interpretation~$\I$, we say that an element~$u \in \DIunrav$ is~\emph{$k$--blocked} 
by its prefix~$w$, if~$u = ww'$ for some~$w'$ of length longer than~$k$, and $w$'s and~$u$'s 
suffixes of length~$k$ coincide. The definition is depicted below.
The definition is depicted below.
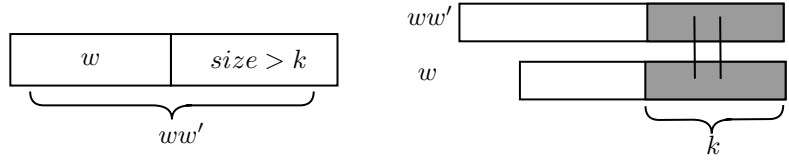
\begin{figure}[!h]
\centering
\tikzset{every picture/.style={line width=0.75pt}} %set default line width to 0.75pt        

\vspace{0.09cm}

\begin{tikzpicture}[x=0.75pt,y=0.75pt,yscale=-1,xscale=1]

%Shape: Rectangle [id:dp9707438841484286] 
\draw   (387.5,48) -- (521.5,48) -- (521.5,67) -- (387.5,67) -- cycle ;
%Shape: Rectangle [id:dp7766644372874836] 
\draw   (211.5,34) -- (295.5,34) -- (295.5,60) -- (211.5,60) -- cycle ;
%Shape: Rectangle [id:dp26097577219048507] 
\draw  [fill={rgb, 255:red, 155; green, 155; blue, 155 }  ,fill opacity=1 ] (450.75,48) -- (521.5,48) -- (521.5,67) -- (450.75,67) -- cycle ;
%Shape: Rectangle [id:dp6229850375630743] 
\draw   (130.5,34) -- (211.5,34) -- (211.5,60) -- (130.5,60) -- cycle ;
%Shape: Brace [id:dp6496009251099921] 
\draw   (140.5,63) .. controls (140.5,67.67) and (142.83,70) .. (147.5,70) -- (202,70) .. controls (208.67,70) and (212,72.33) .. (212,77) .. controls (212,72.33) and (215.33,70) .. (222,70)(219,70) -- (276.5,70) .. controls (281.17,70) and (283.5,67.67) .. (283.5,63) ;
%Shape: Rectangle [id:dp32433612834138636] 
\draw   (357,18.67) -- (520.5,18.67) -- (520.5,37.67) -- (357,37.67) -- cycle ;
%Shape: Rectangle [id:dp3031811651168892] 
\draw  [fill={rgb, 255:red, 155; green, 155; blue, 155 }  ,fill opacity=1 ] (451.75,18.67) -- (520.5,18.67) -- (520.5,37.67) -- (451.75,37.67) -- cycle ;
%Straight Lines [id:da9175457955480988] 
\draw    (475.5,25) -- (475.75,56.5) ;

%Straight Lines [id:da326058021594102] 
\draw    (488.5,25.17) -- (488.75,56.67) ;

%Shape: Brace [id:dp04152409242291988] 
\draw   (450.86,70.43) .. controls (450.89,75.1) and (453.24,77.41) .. (457.91,77.37) -- (474.21,77.24) .. controls (480.88,77.19) and (484.23,79.49) .. (484.26,84.15) .. controls (484.23,79.49) and (487.54,77.13) .. (494.21,77.07)(491.21,77.1) -- (513.34,76.92) .. controls (518.01,76.88) and (520.32,74.53) .. (520.29,69.86) ;

% Text Node
\draw (342.33,24.67) node   {$ww'$};
% Text Node
\draw (341,54.33) node   {$w$};
% Text Node
\draw (171,47) node   {$w$};
% Text Node
\draw (217,84) node   {$ww'$};
% Text Node
\draw (253.5,47) node   {$\ size > k$};
% Text Node
\draw (482.91,90) node   {$\ k$};

\end{tikzpicture}
\end{figure}\\
We also say that~$w$ is \emph{minimally~$k$--blocked} if it is~$k$--blocked (by some prefix), but none of its prefixes is~$k$--blocked. 
With~$\blocked{k}{\Iunrav}$ we denote the set of minimally~$k$--blocked elements in~$\Iunrav$.  

\begin{definition} \label{def:loose}
For a given finite interpretation~$\I$ we define its~$k$--loosening~$\Iloose{k} = (\DIloose{k}, \cdot^{\DIloose{k}})$ 
as an interpretation obtained from~$\Iunrav$ by exhaustively selecting minimally~$k$--blocked elements~$v$ from~$\blocked{k}{\Iunrav}$ 
($k$--blocked by some~$w$), removing all of descendants of~$v$ and identifying~$v$ and~$w$.\\
More formally, we enumerate the set of minimally~$k$--blocked elements~$\blocked{k}{\Iunrav} = \{ v_1, v_2, \ldots, v_n \}$ and 
define a sequence of auxiliary interpretations~$\Ji{0} = \I, \ldots, \Ji{n} = \Iloose{k}$,
where the~$i$--th interpretation~$\J^{i} = (\DJi{i}, \cdotJi{i})$ for any~$i > 0$ is defined as:
\begin{itemize}
	\item $\DJi{i} = \DJi{i-1} \setminus  \big(v_i \cdot ({\DI})^{*} \big)$
	\item $\DJinamed{i} = \DJinamed{i-1}$ and for any~$a \in \indA$ the condition~$a^{\Ji{i}} = a^{\Ji{i-1}}$ is satisfied,
	\item $A^{\Ji{i}} = A^{\Ji{i-1}} \cap \DJi{i}$ for any concept name~$A \in N_C$
	\item $r^{\Ji{i}} = r^{\Ji{i-1}} \cap \left( \DJi{i} \times \DJi{i} \right) \cup \{ (w, v_i') \mid (w,v_i) \in r^{\Ji{i-1}} \}$,
	for any role name~$r \in N_R$, where~$v_i'$ is the element~$k$--blocking~$v_i$ in~$\Iunrav$. 
\end{itemize}
\end{definition}
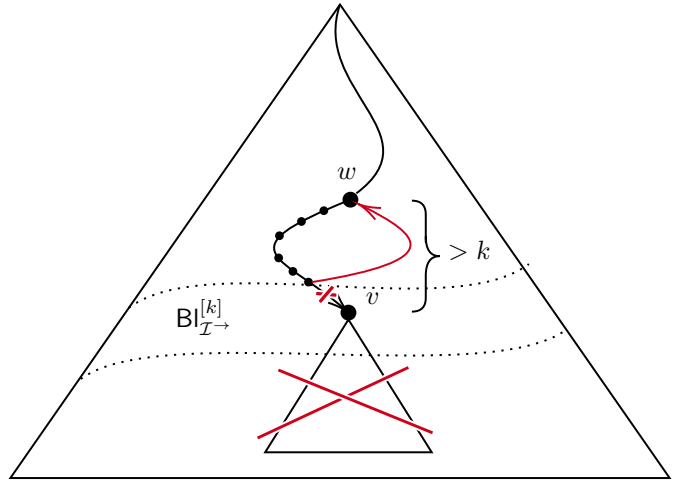
\begin{figure}[!h]
\centering

\tikzset{every picture/.style={line width=0.75pt}} %set default line width to 0.75pt        

\begin{tikzpicture}[x=0.75pt,y=0.75pt,yscale=-1,xscale=1]

%Shape: Triangle [id:dp7160564857064435] 
\draw   (333.25,20) -- (499.5,259) -- (167,259) -- cycle ;
%Curve Lines [id:da3460951896326099] 
\draw    (341.5,115.3) .. controls (381.5,85.3) and (318.5,63) .. (333.25,20) ;

%Shape: Circle [id:dp9172742745194269] 
\draw  [fill={rgb, 255:red, 0; green, 0; blue, 0 }  ,fill opacity=1 ] (335,118.5) .. controls (335,116.57) and (336.57,115) .. (338.5,115) .. controls (340.43,115) and (342,116.57) .. (342,118.5) .. controls (342,120.43) and (340.43,122) .. (338.5,122) .. controls (336.57,122) and (335,120.43) .. (335,118.5) -- cycle ;
%Shape: Triangle [id:dp5976146781559313] 
\draw   (337.5,179) -- (379.5,246) -- (295.5,246) -- cycle ;
%Shape: Circle [id:dp580831283436009] 
\draw  [fill={rgb, 255:red, 0; green, 0; blue, 0 }  ,fill opacity=1 ] (334,175.5) .. controls (334,173.57) and (335.57,172) .. (337.5,172) .. controls (339.43,172) and (341,173.57) .. (341,175.5) .. controls (341,177.43) and (339.43,179) .. (337.5,179) .. controls (335.57,179) and (334,177.43) .. (334,175.5) -- cycle ;
%Curve Lines [id:da2955835515963996] 
\draw    (338.5,118.5) .. controls (271.82,145) and (305.74,146.95) .. (336.12,174.24) ;
\draw [shift={(337.5,175.5)}, rotate = 222.87] [color={rgb, 255:red, 0; green, 0; blue, 0 }  ][line width=0.75]    (10.93,-3.29) .. controls (6.95,-1.4) and (3.31,-0.3) .. (0,0) .. controls (3.31,0.3) and (6.95,1.4) .. (10.93,3.29)   ;

%Shape: Circle [id:dp25254619827977787] 
\draw  [fill={rgb, 255:red, 0; green, 0; blue, 0 }  ,fill opacity=1 ] (323.4,124.05) .. controls (323.4,123.11) and (324.16,122.35) .. (325.1,122.35) .. controls (326.04,122.35) and (326.8,123.11) .. (326.8,124.05) .. controls (326.8,124.99) and (326.04,125.75) .. (325.1,125.75) .. controls (324.16,125.75) and (323.4,124.99) .. (323.4,124.05) -- cycle ;
%Shape: Circle [id:dp692159931413119] 
\draw  [fill={rgb, 255:red, 0; green, 0; blue, 0 }  ,fill opacity=1 ] (312.1,129.5) .. controls (312.1,128.56) and (312.86,127.8) .. (313.8,127.8) .. controls (314.74,127.8) and (315.5,128.56) .. (315.5,129.5) .. controls (315.5,130.44) and (314.74,131.2) .. (313.8,131.2) .. controls (312.86,131.2) and (312.1,130.44) .. (312.1,129.5) -- cycle ;
%Shape: Circle [id:dp9206186353826544] 
\draw  [fill={rgb, 255:red, 0; green, 0; blue, 0 }  ,fill opacity=1 ] (301,136.7) .. controls (301,135.76) and (301.76,135) .. (302.7,135) .. controls (303.64,135) and (304.4,135.76) .. (304.4,136.7) .. controls (304.4,137.64) and (303.64,138.4) .. (302.7,138.4) .. controls (301.76,138.4) and (301,137.64) .. (301,136.7) -- cycle ;
%Shape: Circle [id:dp03700320338097063] 
\draw  [fill={rgb, 255:red, 0; green, 0; blue, 0 }  ,fill opacity=1 ] (300.6,147.9) .. controls (300.6,146.96) and (301.36,146.2) .. (302.3,146.2) .. controls (303.24,146.2) and (304,146.96) .. (304,147.9) .. controls (304,148.84) and (303.24,149.6) .. (302.3,149.6) .. controls (301.36,149.6) and (300.6,148.84) .. (300.6,147.9) -- cycle ;
%Shape: Circle [id:dp3679600915509671] 
\draw  [fill={rgb, 255:red, 0; green, 0; blue, 0 }  ,fill opacity=1 ] (307.8,154.7) .. controls (307.8,153.76) and (308.56,153) .. (309.5,153) .. controls (310.44,153) and (311.2,153.76) .. (311.2,154.7) .. controls (311.2,155.64) and (310.44,156.4) .. (309.5,156.4) .. controls (308.56,156.4) and (307.8,155.64) .. (307.8,154.7) -- cycle ;
%Shape: Circle [id:dp4937032319165784] 
\draw  [fill={rgb, 255:red, 0; green, 0; blue, 0 }  ,fill opacity=1 ] (315.65,160.1) .. controls (315.65,159.16) and (316.41,158.4) .. (317.35,158.4) .. controls (318.29,158.4) and (319.05,159.16) .. (319.05,160.1) .. controls (319.05,161.04) and (318.29,161.8) .. (317.35,161.8) .. controls (316.41,161.8) and (315.65,161.04) .. (315.65,160.1) -- cycle ;
%Shape: Brace [id:dp8495970871637657] 
\draw   (369.5,175) .. controls (374.17,175) and (376.5,172.67) .. (376.5,168) -- (376.5,157) .. controls (376.5,150.33) and (378.83,147) .. (383.5,147) .. controls (378.83,147) and (376.5,143.67) .. (376.5,137)(376.5,140) -- (376.5,126) .. controls (376.5,121.33) and (374.17,119) .. (369.5,119) ;
%Shape: Rectangle [id:dp6944846807368177] 
\draw  [color={rgb, 255:red, 255; green, 255; blue, 255 }  ,draw opacity=1 ][fill={rgb, 255:red, 208; green, 2; blue, 27 }  ,fill opacity=1 ][line width=0.75]  (290.7,237.89) -- (367.67,201.86) -- (368.71,204.08) -- (291.75,240.12) -- cycle ;
%Curve Lines [id:da01301949917857459] 
\draw  [dash pattern={on 0.84pt off 2.51pt}]  (202.5,209) .. controls (242.5,179) and (407.5,214) .. (447.5,184) ;

%Curve Lines [id:da13102679089579805] 
\draw  [dash pattern={on 0.84pt off 2.51pt}]  (226.71,173.08) .. controls (266.71,143.08) and (388.5,181) .. (428.5,151) ;

%Shape: Rectangle [id:dp9000955837819218] 
\draw  [color={rgb, 255:red, 255; green, 255; blue, 255 }  ,draw opacity=1 ][fill={rgb, 255:red, 208; green, 2; blue, 27 }  ,fill opacity=1 ][line width=0.75]  (331.82,168.07) -- (320.65,167.12) -- (320.87,164.47) -- (332.04,165.42) -- cycle ;
%Shape: Rectangle [id:dp4667004363119056] 
\draw  [color={rgb, 255:red, 255; green, 255; blue, 255 }  ,draw opacity=1 ][fill={rgb, 255:red, 208; green, 2; blue, 27 }  ,fill opacity=1 ][line width=0.75]  (321.69,169.66) -- (328.98,161.14) -- (331,162.88) -- (323.71,171.39) -- cycle ;
%Shape: Rectangle [id:dp18128951436815877] 
\draw  [color={rgb, 255:red, 255; green, 255; blue, 255 }  ,draw opacity=1 ][fill={rgb, 255:red, 208; green, 2; blue, 27 }  ,fill opacity=1 ][line width=0.75]  (379.92,237.5) -- (301.23,205.42) -- (302.16,203.14) -- (380.85,235.23) -- cycle ;
%Curve Lines [id:da8770882120758392] 
\draw [color={rgb, 255:red, 208; green, 2; blue, 27 }  ,draw opacity=1 ]   (319.05,160.1) .. controls (387.46,149.17) and (374.61,136.68) .. (344.11,121.5) ;
\draw [shift={(342.7,120.8)}, rotate = 386.2] [color={rgb, 255:red, 208; green, 2; blue, 27 }  ,draw opacity=1 ][line width=0.75]    (10.93,-3.29) .. controls (6.95,-1.4) and (3.31,-0.3) .. (0,0) .. controls (3.31,0.3) and (6.95,1.4) .. (10.93,3.29)   ;

% Text Node
\draw (398.5,144) node [xslant=-0.05]  {$ >k$};
% Text Node
\draw (336.83,105.23) node   {$w$};
% Text Node
\draw (349.73,168.5) node   {$v$};
% Text Node
\draw (264.83,177.83) node   {$\blocked{k}{\Iunrav}$};

\end{tikzpicture}
\caption{A single step of the construction of~$\Iloose{k}$.}
\end{figure}
We first argue that~$k$--loosening of a finite interpretation is also finite.
\begin{lemma} \label{lem:loosefin}
For any finite interpretation~$\I$, its~$k$--loosening~$\Iloose{k}$ for any natural~$k > 0$ is finite.
\end{lemma}
\begin{proof}
Take an arbitrary finite~$\I$ and observe that the branching of~$k$--loosening is 
finite due to finiteness of~$\I$ and each element of~$\Iloose{k}$ has only finite 
number of successors (by pigeon-hole principle the blocking eventually occurs on every branch of~$\Iunrav$).
Hence by employing (the contraposition) of the K\"{o}nig's Lemma, we conclude that~$\Iloose{k}$ is finite.
\end{proof}
Like unravelings,~$k$-loosenings preserve satisfaction of normalized Aboxes and Tboxes, as well as CQ non-entailment.
However, \mbox{ERCBoxes} might become violated in the construction.
We startfrom the ABox preservation.

\begin{lemma} \label{lem:unravaboxpreservation}
For any finite~$\I$ and any normalized ABox~$\abox$ and any natural~$k > 0$, 
the implication if~$\I \models \abox$ then~$\Iloose{k} \models \abox$ holds.
\end{lemma}
\begin{proof}
Assume that~$\I \models \abox$ holds. Then, due to Lemma~\ref{lem:unravabox} we know that~$\Iunrav \models \abox$ holds. 
Observe that~$\DIloose{k}$ is a subset of~$\DIunrav$, due to the first item of Definition~\ref{def:loose}. 
Moreover the sets~$\DIunravnamed$ and~$\DIloosenamed{k}$ are equal, due to the second item of Definition~\ref{def:loose}. 
Since the~$k$--loosening construction does not affect the ABox part of~$\Iunrav$ (e.g. those 
elements are not~$k$--blocked for any~$k$, see also the second item of 
Definition~\ref{def:loose}) we conclude that~$\Iloose{k}$ is a model of~$\abox$.
\end{proof}

Towards proving the TBox preservation of~$k$--loosening, we prepare a bisimulation argument.
\begin{lemma} \label{lem:loosebisim}
Let~$\kb = (\abox, \tbox, \ercbox)$ be a normalized~$\ALCSCC$ knowledge base and let~$\I$ be its arbitrary finite model.
Then any~$w \in \DIloose{k}$ is bisimilar to~$\last{w} \in \DI$.
\end{lemma}
\begin{proof}
Take an arbitrary domain element~$w = w_{\Iloose{k}} \in \DIloose{k}$ and, since~$\DIloose{k} \subseteq \DIunrav$ holds 
(see: Definition~\ref{def:loose}), let~$w_{\Iunrav} = w$ be the corresponding element from~$\DIunrav$. To show that~$w$ and~$\last{w}$ are bisimilar,
is sufficient prove that~$w_{\Iloose{k}} \eqbisim w_{\Iunrav}$ and use Lemma~\ref{lem:bisimlast}.

We proceed as follows. We define a function~$\izo : \Neib{\Iloose{k}}{w} \rightarrow \Neib{\Iunrav}{w}$ 
as~$\izo(w') = w'$ for all~$w' \in \Neib{\Iloose{k}}{w} \cap \Neib{\Iunrav}{w}$
and~$f(w') = w \cdot \last{w'}$ otherwise (note that in this case~$w'$ is some of minimally~$k$--blocked elements). 

We first argue that~$\izo$ is a function. Since~$\DIloose{k} \subseteq \DIunrav$ holds, we infer that~$\izo$ is an identity function on the set~$\Neib{\Iloose{k}}{w} \cap \Neib{\Iunrav}{w}$, thus well-defined. The problematic case is when~$w'$ is not included in~$\Neib{\Iloose{k}}{w} \cap \Neib{\Iunrav}{w}$. 
Observe that in this case~$w'$ was identified, during the construction of~$\Iloose{k}$, 
with some~$k$--blocked element~$v \in \blocked{k}{\Iunrav}$,  which originally was a successor 
of~$w$. It means that~$v$ was~$k$--blocked by~$w'$ and from the definition of~$k$--blocked elements 
we infer that~$w'$ and~$v$ share the same suffix of length~$k$. Thus~$w'$ and~$v$ share the same last element.  
Since~$v$ is a successor of~$w$, then~$v = w \cdot \last{v} = w \cdot \last{w'}$. Hence the definition of~$\izo$ is sound.

To see that~$\izo : \restr{\Neib{\Iloose{k}}{w}}{\SuccI{\Iloose{k}}{w}} \rightarrow \restr{\Neib{\Iunrav}{w}}{\SuccI{\Iunrav}{w}}$ is a bijection, we 
can restrict our attention only to the elements not included in the set~$\Neib{\Iloose{k}}{w} \cap \Neib{\Iunrav}{w}$, since, as we already mentioned, on such set~$\izo$ is the identity function and thus, also a bijection. 
Observe that~$\izo$ is injection for any~$w' \in \Neib{\Iloose{k}}{w} \setminus \Neib{\Iunrav}{w}$.
Indeed, if there would be~$w', w''$ satisfying~$\izo(w') = \izo(w'')$, then it would
imply that they originated from the same successor of~$w$ in~$\Iunrav$ (since they share the same suffix),
which is clearly not possible.
To see that~$\izo$ is a surjection it is enough to see that for any 
successor~$w'= we$ of~$w$ in~$\Iunrav$ the function~$\izo$ is either identity (thus~$\izo(w')= w'$)
or~$w'$ was minimally~$k$--blocked and hance was identified with an element sharing the same last element. 
Hence,~$\izo$ (restricted to appropriate sets) is a bijection.

We will prove that~$\izo$ is a bisimulation. In the first part we will prove the following statement:
\[
\forall{A \in N_C} \; \forall{w' \in \Neib{\Iloose{k}}{w}} \;
\text{the equivalence} \; w' \in A^{\Iloose{k}} \Leftrightarrow  \izo(w') \in A^{\Iunrav} \; \text{holds.}
\]
Take an arbitrary concept name~$A$ and arbitrary domain element~$w' \in \Neib{\Iloose{k}}{w}$. 
If~$\izo(w') = w'$ then the above condition trivially holds. Assume that~$\izo(w') \neq w'$.
Then~$\izo(w') = w \last{w'}$ and the preservation of concepts follows from Definition~\ref{def:unrav}.

In the second part we will prove:
\[
\forall{r \in N_R} \; \forall{w' \in \Neib{\Iloose{k}}{w}} \;
\text{the equivalence} \; (w, w') \in r^{\Iloose{k}} \Leftrightarrow  (\izo(w), \izo(w')) \in r^{\Iunrav} \; \text{holds.}
\]
Take an arbitrary role name~$r$ and arbitrary domain element~$w' \in \Neib{\Iloose{k}}{w}$.  
Once more, if~$\izo(w') = w'$ then the above condition trivially holds. Assume that~$\izo(w') \neq w'$. 
Then again~$\izo(w') = w \last{w'} = v$ and~$v$ is minimally~$k$--blocked by~$w'$.
From Definition~\ref{def:loose} we know that~$(w,v) \in r^{\Iunrav}$ iff~$(w,w') \in r^{\Iloose{k}}$,
which proves the statement about (non)preservation of roles during the construction of~$\Iloose{k}$.

We conclude that~$\izo$ is a bisimulation and hence~$w_{\Iloose{k}} \eqbisim w_{\Iunrav}$ holds.
\end{proof}

The TBox preservation follows immediately from the previous lemma.

\begin{lemma} \label{lem:loosemodelhood}
For any finite~$\I$ and any normalized TBox~$\tbox$ and any natural~$k > 0$, 
the implication if~$\I \models \tbox$ then~$\Iloose{k} \models \tbox$ holds.
\end{lemma}
\begin{proof}
Take an arbitrary finite interpretation~$\I$, a normalized TBox~$\tbox$ and a positive integer~$k$.
Assume that~$\I \models \tbox$ holds. To prove that each GCI~$\varepsilon$ from~$\tbox$ is also satisfied in~$\Iloose{k}$,
we apply the same reasoning as we already done for~Lemma~\ref{lem:unravtbox}. Namely, it is sufficient to prove that the~$k$--loosening
construction is concept preserving but it can be concluded from Definition~\ref{def:bisim} (of bisimulation) and from Lemma~\ref{lem:loosebisim}.
\end{proof}

\begin{lemma} \label{lem:loosenonentail}
For any~$k \in \Nat$, if~$\I \not\models q$ then~$\Iloose{k} \not\models q$.
\end{lemma}
\begin{proof}
Assume that~$\I \not\models q$, but~$\Iloose{k} \models q$. In this case there exists a match~$\pi$ of~$q$ on~$\Iloose{k}$.
By using the same ideas as for Lemma~\ref{lem:unravnonentail} we argue that in this case~$\pi'$ with~$\pi'(x) = \last{\pi(x)}$ 
would be a match of~$q$ on~$\I$, which contradicts with~$\I \not\models q$. Thus~$\Iunrav \not\models q$ holds.
\end{proof}

For a given interpretation~$\J$, an \emph{anonymous cycle} is simply a 
word~$w \in (\DJ)^{+} \cdot (\DJ \setminus \DJnamed) \cdot (\DJ)^{+}$, 
where first and the last element are the same, and for any two consecutive 
elements~$d_{i},d_{i+1}$ of~$w$ there exists a role~$r$ 
witnessing~$(d_{i}, d_{i+1}) \in r^{\J}$. The \emph{girth} of~$\J$ is the 
length of the smallest anonymous cycle in~$\J$ if such a cycle exists or~$\infty$ otherwise.
The main feature of the~$k$--loosening~$\Iloose{k}$ is that the 
girth of~$\Iloose{k}$ is at least~$k$, as proven below. 

\begin{lemma} \label{lem:loosegirthk}
For any~$k \in \Nat$ and any finite interpretation~$\I$, the girth of~$\Iloose{k}$ is at least~$k$. 
\end{lemma}
\begin{proof}
We will prove inductively over immediate structures~$\Ji{0} = \Iunrav, \Ji{1}, \ldots, \Ji{n} = \Iloose{k}$ 
produced in Definition~\ref{def:loose} that each of them have girth greater than~$k$.
For~$i = 0$ it is clear that~$\Ji{0}$ has girth at least~$k$ (actually its girth is~$\infty$).
Assume that for all~$i < m$ the girth of each~$\Ji{i}$ for~$i < m$ is at least~$k$. We will show that the girth of~$\I_m$ is at least~$k$.

For contradiction assume that the girth of~$\Ji{m}$ is smaller than~$k$. We recall that~$v_m$ is the~$m$--th 
minimally~$k$--blocked elements from~$\blocked{k}{\Iunrav}$ and~$v_m'$ is the element~$k$--blocking~$v_m$.
Since~$\Ji{m}$ was obtained from~$\Ji{m-1}$ and the girth of~$\Ji{m-1}$ is at least~$k$ then the only
possibility of a anonymous cycle of length at least~$k$ to be present in~$\Ji{m}$ is to contain 
a freshly added edge between predecessors~$w$ of~$v_m$ and~$v_m'$, namely $(w, v_m')$ for 
some~$r \in N_R$ as a replacement for an original edge~$(w, v_m)$.

Let~$\rho$ be an arbitrary shortest anonymous cycle in~$\Ji{m-1}$. As we already discussed it contains an edge~$(w, v_m)$
between some domain element~$w$. Hence~$\rho$ is in the form~$(w, v_m') \rho'$ where~$\rho'$ is some path from~$v_m'$ to~$w$.
But note that due the definition of~$k$--blocked element the distance between~$v_m$ and~$v_m'$ is at least~$k$. Hence~$\rho'$ is 
of length at least~$k$. Thus~$\rho$ is not shorter than~$k$, which contradict our initial assumption.
Hence the girth of~$\Ji{m}$ is at least~$k$, which allows us to conclude that the girth of~$\Ji{n} = \Iloose{k}$ is also at least~$k$. 
\end{proof}

Once~$k$ is greater than the number of atoms in~$q$ (denoted with~$|q|$), 
the~$k$--loosening of a model is still ``locally acyclic enough'' so the query 
matches only in a ``forest-shaped'' manner. We will exploit this
property when designing an algorithm for deciding conjunctive query entailment
in Section~\ref{sec:queryentailmentalgo}.

\begin{lemma} \label{lem:eqqueriesunravloose}
For every conjunctive query~$q$, a positive integer~$k > |q|$ and a finite interpretation~$\I$, 
the following equivalence~$\Iunrav \models q \Leftrightarrow \Iloose{k} \models q$ holds.
\end{lemma}

\begin{proof}
Let~$\suff{s}{w}$ be a function which for an input word~$w \in (\DI)^+$ returns~$w$ if~$|w| \leq s$ or its suffix of length~$s$ otherwise.
Moreover let~$\Iunrav_{k}$ be the substructure of~$\Iunrav$ with domain restricted to sequences of length at most~$k$ only.
Note that~$\homo(w) = \suff{k}{w}$ is a homomorphism from~$\Iunrav$ to~$\Iunrav$ (since~$w \eqbisim \suff{k}{w}$, 
see the proof of Lemma~\ref{lem:unravtbox}). Hence if there is a match~$\pi$ of~$q$ in~$\Iunrav$, there is also 
a match~$\pi'$ of~$q$ in~$\Iunrav_{k}$. Since~$\Iunrav_{k}$ is a substructure of~$\Iloose{k}$ (due to 
the definition of minimally~$k$--blocked elements and Definition~\ref{def:loose}), hence~$\pi'$ is also a match in~$\Iloose{k}$.

For the opposite way, that i.e.,~$\Iloose{k} \models q$ implies~$\Iunrav \models q$, it is sufficient to show (since~$k > |q|$)
that there is a homomorphism from any substructure of the size~$k$ of~$\Iloose{k}$ to~$\Iunrav$.
Take an arbitrary element~$w \in \DIloose{k}$ and take a interpretation~$\Iloose{k}_w$ be an interpretation obtained
by restricting the domain to elements reachable from~$w$ in at most~$k$ steps. More formally we define the 
sets~$R_i(w)$ of those elements reachable from~$w$ in at most~$i$ steps, i.e.~$R_0(w) = \{ w \}$,
and~$R_i(w) = R_{i-1}(w) \cup \{ v \in \DIloose{k} \mid \exists{r \in N_R} \; (u,v) \in r^{\Iloose{k}} \wedge u \in R_{i-1}(w) \}$ 
for all~$i>0$. We set~$\DIloose{k}_w = R_k(w)$. First see that~$\Iloose{k}_w$ is a tree-shaped. 
Indeed if it would contain an anonymous cycle of length at most~$k$ it would contradict the fact that the girth of~$\Iloose{k}$ is at 
least~$k$ (by Lemma~\ref{lem:loosegirthk}). Hence we take a homomorphism~$\homo : \Iloose{k}_w \rightarrow \Iunrav$ defined as~$\homo(x) = \suff{k}{x}$
and see that if there is a match~$\pi$ of~$q$ in~$\Iloose{k}$, then~$\pi' = (\homo \circ \pi)$ would also be a match of~$q$ in~$\Iunrav$. 
\end{proof}

\subsubsection{Making ERCBoxes be satisfied again} \label{sec:dupl}

We next consider how to adjust a $k$-loosening such that it again satisfies the initial ERCBox.  
Since role inverses are not expressible in~$\ALCSCC$, creating multiple copies of a single element and forward-linking them to other elements precisely in the same way as the original element, can be done without any harm to modelhood nor query-non-entailment. We formalize this intuition below.

\begin{definition} \label{def:repair}
For any interpretation~$\I$ and any sets~$S \subseteq (\DI \times \Nat_+)$ 
we define the~\emph{$S$--duplication} of~$\I$ as the interpretation~$\I_{{+}S} = (\DI_{{+}S}, \cdot^{\I_{{+}S}})$ with:
\begin{itemize}
	\item~$\DI_{{+}S} = \DI \cup \; \bigcup_{(v,n) \in S} \{ v_{\mathit{cpy}}^{(i)} \mid 1 \leq i \leq n \}$,
	\item~$a^{\I_{{+}S}} = a^{\I}$ for each individual name~$a \in \indA$,
	\item For concept names~$A \in N_C$ and role names~$r \in N_R$ we set:
	\begin{itemize}
	\item $A^{\I_{{+}S}} = A^{\I} \cup \; \bigcup_{(v,n) \in S} \Big\{ v_{\mathit{cpy}}^{(i)} \mid 1 \leq i \leq n \wedge v \in A^{\I} \Big\}, \; \text{and,} $
	\item $r^{\I_{{+}S}} = r^{\I} \cup \; \bigcup_{(v,n) \in S} \Big\{ (v_{\mathit{cpy}}^{(i)}, w) \mid 1 \leq i \leq n \wedge (v,w) \in r^{\I} \Big\}$. 
	\end{itemize}
\end{itemize}
\end{definition}

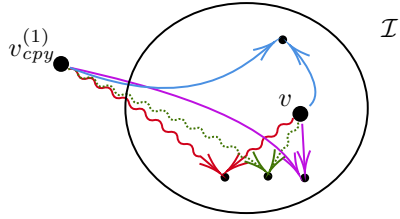
\begin{figure}[!h]
\centering

\tikzset{every picture/.style={line width=0.75pt}} %set default line width to 0.75pt        

\begin{tikzpicture}[x=0.75pt,y=0.75pt,yscale=-1,xscale=1]

%Shape: Circle [id:dp07544601330436196] 
\draw  [fill={rgb, 255:red, 0; green, 0; blue, 0 }  ,fill opacity=1 ] (335,118.5) .. controls (335,116.57) and (336.57,115) .. (338.5,115) .. controls (340.43,115) and (342,116.57) .. (342,118.5) .. controls (342,120.43) and (340.43,122) .. (338.5,122) .. controls (336.57,122) and (335,120.43) .. (335,118.5) -- cycle ;
%Shape: Circle [id:dp6971000529082729] 
\draw  [fill={rgb, 255:red, 0; green, 0; blue, 0 }  ,fill opacity=1 ] (298.77,150.5) .. controls (298.77,149.56) and (299.53,148.8) .. (300.47,148.8) .. controls (301.41,148.8) and (302.17,149.56) .. (302.17,150.5) .. controls (302.17,151.44) and (301.41,152.2) .. (300.47,152.2) .. controls (299.53,152.2) and (298.77,151.44) .. (298.77,150.5) -- cycle ;
%Shape: Circle [id:dp6344288632485375] 
\draw  [fill={rgb, 255:red, 0; green, 0; blue, 0 }  ,fill opacity=1 ] (320.48,149.99) .. controls (320.48,149.05) and (321.24,148.29) .. (322.18,148.29) .. controls (323.12,148.29) and (323.88,149.05) .. (323.88,149.99) .. controls (323.88,150.92) and (323.12,151.69) .. (322.18,151.69) .. controls (321.24,151.69) and (320.48,150.92) .. (320.48,149.99) -- cycle ;
%Shape: Circle [id:dp9418225339971791] 
\draw  [fill={rgb, 255:red, 0; green, 0; blue, 0 }  ,fill opacity=1 ] (339.03,150.83) .. controls (339.03,149.89) and (339.79,149.13) .. (340.73,149.13) .. controls (341.67,149.13) and (342.43,149.89) .. (342.43,150.83) .. controls (342.43,151.77) and (341.67,152.53) .. (340.73,152.53) .. controls (339.79,152.53) and (339.03,151.77) .. (339.03,150.83) -- cycle ;
%Straight Lines [id:da19147010351945704] 
\draw [color={rgb, 255:red, 208; green, 2; blue, 27 }  ,draw opacity=1 ]   (334.75,120.57) .. controls (334.5,122.92) and (333.2,123.96) .. (330.86,123.71) .. controls (328.51,123.46) and (327.22,124.5) .. (326.97,126.85) .. controls (326.72,129.2) and (325.43,130.25) .. (323.08,130) .. controls (320.73,129.75) and (319.44,130.79) .. (319.19,133.14) .. controls (318.94,135.49) and (317.65,136.53) .. (315.3,136.28) .. controls (312.95,136.03) and (311.66,137.07) .. (311.41,139.42) -- (309.16,141.23) -- (302.94,146.26) ;
\draw [shift={(301.38,147.52)}, rotate = 321.08000000000004] [color={rgb, 255:red, 208; green, 2; blue, 27 }  ,draw opacity=1 ][line width=0.75]    (10.93,-3.29) .. controls (6.95,-1.4) and (3.31,-0.3) .. (0,0) .. controls (3.31,0.3) and (6.95,1.4) .. (10.93,3.29)   ;
%Straight Lines [id:da9180203156561269] 
\draw [color={rgb, 255:red, 65; green, 117; blue, 5 }  ,draw opacity=1 ] [dash pattern={on 0.75pt off 0.75pt}]  (337.26,122.91) .. controls (337.87,125.19) and (337.03,126.63) .. (334.75,127.24) .. controls (332.47,127.84) and (331.63,129.28) .. (332.24,131.56) .. controls (332.84,133.84) and (332,135.28) .. (329.72,135.88) -- (328,138.84) -- (323.98,145.76) ;
\draw [shift={(322.98,147.49)}, rotate = 300.17] [color={rgb, 255:red, 65; green, 117; blue, 5 }  ,draw opacity=1 ][line width=0.75]    (10.93,-3.29) .. controls (6.95,-1.4) and (3.31,-0.3) .. (0,0) .. controls (3.31,0.3) and (6.95,1.4) .. (10.93,3.29)   ;
%Straight Lines [id:da85382754590328] 
\draw [color={rgb, 255:red, 189; green, 16; blue, 224 }  ,draw opacity=1 ]   (339.43,123.5) -- (340.67,146.09) ;
\draw [shift={(340.78,148.09)}, rotate = 266.85] [color={rgb, 255:red, 189; green, 16; blue, 224 }  ,draw opacity=1 ][line width=0.75]    (10.93,-3.29) .. controls (6.95,-1.4) and (3.31,-0.3) .. (0,0) .. controls (3.31,0.3) and (6.95,1.4) .. (10.93,3.29)   ;
%Shape: Circle [id:dp7668617553754236] 
\draw  [fill={rgb, 255:red, 0; green, 0; blue, 0 }  ,fill opacity=1 ] (214.33,93.3) .. controls (214.33,91.37) and (215.9,89.8) .. (217.83,89.8) .. controls (219.77,89.8) and (221.33,91.37) .. (221.33,93.3) .. controls (221.33,95.23) and (219.77,96.8) .. (217.83,96.8) .. controls (215.9,96.8) and (214.33,95.23) .. (214.33,93.3) -- cycle ;
%Curve Lines [id:da07961207090103284] 
\draw [color={rgb, 255:red, 208; green, 2; blue, 27 }  ,draw opacity=1 ]   (221.53,96) .. controls (223.84,95.25) and (225.4,96.01) .. (226.23,98.28) .. controls (227.02,100.55) and (228.55,101.34) .. (230.84,100.65) .. controls (233.14,99.98) and (234.5,100.71) .. (234.91,102.83) .. controls (235.58,105.11) and (237.05,105.94) .. (239.34,105.32) .. controls (241.64,104.73) and (243.09,105.58) .. (243.68,107.87) .. controls (244.24,110.15) and (245.66,111.01) .. (247.93,110.46) .. controls (250.22,109.93) and (251.6,110.8) .. (252.08,113.09) .. controls (252.79,115.54) and (254.27,116.52) .. (256.53,116.01) .. controls (258.79,115.53) and (260.1,116.43) .. (260.47,118.7) .. controls (261.05,121.14) and (262.45,122.13) .. (264.68,121.66) .. controls (266.91,121.21) and (268.27,122.2) .. (268.77,124.63) .. controls (268.98,126.86) and (270.29,127.85) .. (272.72,127.6) .. controls (275.15,127.35) and (276.42,128.33) .. (276.55,130.54) .. controls (276.87,132.91) and (278.21,133.98) .. (280.57,133.73) .. controls (282.92,133.48) and (284.21,134.53) .. (284.44,136.86) .. controls (284.62,139.17) and (285.85,140.2) .. (288.13,139.94) -- (290.79,142.19) -- (296.63,147.26) ;
\draw [shift={(297.98,148.44)}, rotate = 221.55] [color={rgb, 255:red, 208; green, 2; blue, 27 }  ,draw opacity=1 ][line width=0.75]    (10.93,-3.29) .. controls (6.95,-1.4) and (3.31,-0.3) .. (0,0) .. controls (3.31,0.3) and (6.95,1.4) .. (10.93,3.29)   ;
%Curve Lines [id:da5512421798537821] 
\draw [color={rgb, 255:red, 65; green, 117; blue, 5 }  ,draw opacity=1 ] [dash pattern={on 0.75pt off 0.75pt}]  (221.53,96) .. controls (223.84,95.25) and (225.46,96.01) .. (226.39,98.28) .. controls (227.02,100.41) and (228.54,101.11) .. (230.93,100.38) .. controls (233,99.5) and (234.57,100.22) .. (235.62,102.53) .. controls (236.34,104.69) and (237.76,105.34) .. (239.88,104.48) .. controls (242.01,103.62) and (243.64,104.36) .. (244.76,106.7) .. controls (245.53,108.89) and (247,109.56) .. (249.16,108.71) .. controls (251.33,107.87) and (252.81,108.55) .. (253.6,110.74) .. controls (254.39,112.93) and (255.88,113.61) .. (258.07,112.78) .. controls (260.26,111.96) and (261.74,112.65) .. (262.53,114.85) .. controls (263.32,117.05) and (264.8,117.74) .. (266.99,116.92) .. controls (269.54,116.29) and (271.2,117.07) .. (271.96,119.27) .. controls (272.7,121.47) and (274.16,122.17) .. (276.33,121.36) .. controls (278.5,120.56) and (279.94,121.26) .. (280.64,123.46) .. controls (281.31,125.65) and (282.89,126.43) .. (285.38,125.81) .. controls (287.51,125.02) and (288.88,125.71) .. (289.49,127.89) .. controls (290.39,130.24) and (291.88,131.02) .. (293.97,130.23) .. controls (296.37,129.62) and (297.8,130.4) .. (298.27,132.55) .. controls (298.99,134.86) and (300.5,135.71) .. (302.81,135.11) .. controls (305.11,134.52) and (306.53,135.36) .. (307.08,137.63) .. controls (307.54,139.87) and (308.98,140.78) .. (311.41,140.36) -- (313.24,141.58) -- (319.33,146.11) ;
\draw [shift={(320.72,147.29)}, rotate = 221.55] [color={rgb, 255:red, 65; green, 117; blue, 5 }  ,draw opacity=1 ][line width=0.75]    (10.93,-3.29) .. controls (6.95,-1.4) and (3.31,-0.3) .. (0,0) .. controls (3.31,0.3) and (6.95,1.4) .. (10.93,3.29)   ;
%Curve Lines [id:da9696780353638337] 
\draw [color={rgb, 255:red, 189; green, 16; blue, 224 }  ,draw opacity=1 ]   (222.33,95.2) .. controls (261.12,103.73) and (320.45,119.1) .. (337.12,147.88) ;
\draw [shift={(337.85,149.2)}, rotate = 242.26] [color={rgb, 255:red, 189; green, 16; blue, 224 }  ,draw opacity=1 ][line width=0.75]    (10.93,-3.29) .. controls (6.95,-1.4) and (3.31,-0.3) .. (0,0) .. controls (3.31,0.3) and (6.95,1.4) .. (10.93,3.29)   ;
%Shape: Circle [id:dp04704717669961744] 
\draw  [fill={rgb, 255:red, 0; green, 0; blue, 0 }  ,fill opacity=1 ] (327.92,81.1) .. controls (327.92,80.16) and (328.68,79.4) .. (329.62,79.4) .. controls (330.56,79.4) and (331.32,80.16) .. (331.32,81.1) .. controls (331.32,82.04) and (330.56,82.8) .. (329.62,82.8) .. controls (328.68,82.8) and (327.92,82.04) .. (327.92,81.1) -- cycle ;
%Curve Lines [id:da36957469655732256] 
\draw [color={rgb, 255:red, 74; green, 144; blue, 226 }  ,draw opacity=1 ]   (343.13,114.5) .. controls (350.57,108.01) and (342.52,93.77) .. (334.4,85.07) ;
\draw [shift={(333.13,83.75)}, rotate = 405] [color={rgb, 255:red, 74; green, 144; blue, 226 }  ,draw opacity=1 ][line width=0.75]    (10.93,-3.29) .. controls (6.95,-1.4) and (3.31,-0.3) .. (0,0) .. controls (3.31,0.3) and (6.95,1.4) .. (10.93,3.29)   ;
%Curve Lines [id:da3446043311206779] 
\draw [color={rgb, 255:red, 74; green, 144; blue, 226 }  ,draw opacity=1 ]   (222.33,95.2) .. controls (284.58,119.42) and (309.98,94.52) .. (323.9,83.4) ;
\draw [shift={(325.38,82.25)}, rotate = 502.59] [color={rgb, 255:red, 74; green, 144; blue, 226 }  ,draw opacity=1 ][line width=0.75]    (10.93,-3.29) .. controls (6.95,-1.4) and (3.31,-0.3) .. (0,0) .. controls (3.31,0.3) and (6.95,1.4) .. (10.93,3.29)   ;
%Shape: Ellipse [id:dp7668506023237343] 
\draw   (250.25,115.58) .. controls (250.25,86.31) and (277.62,62.58) .. (311.38,62.58) .. controls (345.13,62.58) and (372.5,86.31) .. (372.5,115.58) .. controls (372.5,144.85) and (345.13,168.58) .. (311.38,168.58) .. controls (277.62,168.58) and (250.25,144.85) .. (250.25,115.58) -- cycle ;

% Text Node
\draw (330.98,111.68) node    {$v$};
% Text Node
\draw (203.33,83.9) node    {$v_{\mathit{cpy}}^{(1)}$};
% Text Node
\draw (383.48,75.18) node    {$\I$};

\end{tikzpicture}

\caption{The interpretation~$\I_{+\{(v,1)\}}$ obtained from~$\I$ by duplicating a node~$v$.}
\end{figure}

As in the case of previous constructions, one can show that the $S$--duplication of~$\I$ preserves satisfaction of ABoxes and TBoxes.
\begin{lemma} \label{lem:repairmodelhood}
For any finite~$\I$ and normalized ABox~$\abox$ and normalized TBox~$\tbox$, 
if~$\I \models (\abox, \tbox)$, then for any~$S \subseteq (\DI \times \Nat_+)$, 
the~$S$--duplication~$\I_{{+}S}$ of~$\I$ is also a model of~$(\abox, \tbox)$.
\end{lemma}
\begin{proof}
Since~$\I$ is a submodel of~$\I_{{+}S}$ we conclude that~$\I_{{+}S} \models \abox$. To see that~$S$--duplication does not violate the TBox~$\tbox$, it is sufficient to see that for any~$i \in \Nat_+$ and~$v \in \DI$ an element~$v_{\mathit{cpy}}^{(i)}$ is bisimilar to~$v$ (which follows immediately from Definition~\ref{def:repair}). 
Hence~$\I_{{+}S} \models (\abox, \tbox)$.
\end{proof}
Moreover a conjunctive query~$q$ has a match in~$\I$ if and only if it has a match in~$\I_{{+}S}$.
\begin{lemma} \label{lem:CQduplthesamematches}
For any conjunctive query~$q$ and any~$S \subseteq (\DI \times \Nat_+)$ and any interpretation~$\I$, 
the equivalence~$\I \models q \Leftrightarrow \I_{{+}S} \models q$ holds.
\end{lemma}
\begin{proof}
Without loss of generality we assume all concepts appearing in~$q$ are atomic. 
If~$\I$ has a match~$\pi$ of~$q$, then trivially~$\pi$ is also a match in~$\I_{{+}S}$ (due to the fact that~$\I$ is a submodel of~$\I_{{+}S}$). 
For the second direction, assume that there is a query match~$\pi$ of~$q$ in~$\I_{{+}S}$. Let us define~$\homo : \I_{{+}S} \rightarrow \I$ 
as~$\homo\big(v^{(i)}_{\mathit{cpy}}\big) = v$ for freshly copied elements and as~$\homo(v) = v$ otherwise. It is easy to see that~$\homo$ is a homomorphism, 
and hence~$\homo \circ \pi$ is a match of~$q$ in~$\I$. Thus the equivalence~$\I \models q \Leftrightarrow \I_{{+}S} \models q$ holds. 
\end{proof}

From Lemma~\ref{lem:CQduplthesamematches} and Lemma~\ref{lem:eqqueriesunravloose} we can immediately conclude:
\begin{lemma} \label{lem:duplTheSameMatches}
For any conjunctive query~$q$, any positive integer~$k > |q|$ and any finite interpretation~$\I$ the 
following equivalence holds:~$\Iunrav \models q \Leftrightarrow \Iloose{k}_{{+}S} \models q$.
\end{lemma}

Note that for any finite~$\I$ being a model of a normalized~$\kb = (\abox, \tbox, \ercbox)$ it could be the case 
that~$\Iloose{k}$ does not satisfy the ERCBox~$\ercbox$ anymore. However, the inequalities from~$\ercbox$ have the convenient property that
if a vector~$\vec{x}$ containing the cardinalities of all atomic concepts' extensions 
is a solution to~$\ercbox$, then also a vector~$c \cdot \vec{x}$, i.e., the vector obtained by 
multiplying each entry of~$\vec{x}$ by a constant~$c$, is a solution to~$\ercbox$. Thus there is also a solution 
to~$\ercbox$ in the shape~$(1+|\DIloose{k}|) \cdot \vec{x_{\I}}$, where~$\vec{x_{\I}}$ is the solution to~$\ercbox$
describing the atomic concept extensions' cardinalities in~$\I$. Since~$\Iloose{k}$ preserves (non-)emptiness of all concepts from~$\I$, we can
simply duplicate an appropriate number of elements from~$\Iloose{k}$, until the ERCBox~$\ercbox$ will be satisfied again.
The whole procedure is described in the forthcoming lemma.

\begin{lemma} \label{lem:repairingERC}
For any consistent normalized~$\ALCSCC$ knowledge base~$\kb = (\abox, \tbox, \ercbox)$ and for 
any of its finite models~$\I$ there exists a finite~$S \subseteq (\DI \times \Nat_+)$ such that~$\Iloose{k}_{{+}S} \models (\abox, \tbox, \ercbox)$ holds.
\end{lemma}
\begin{proof}
\newcommand{\concepts}{\mathbb{C}}
\newcommand{\types}{\mathbb{T}_{\concepts}}
Let~$\concepts$ be the set of all atomic concepts appearing in normalized~$\kb = (\abox, \tbox, \ercbox)$. In this proof, a \emph{type}
means a conjunction of (possibly negated) concepts from~$\concepts$. With~$\types$ we denote the set of all possible types.

The ERCBox~$\ercbox'$ is obtained from~$\ercbox$ by replacing each inequality~$\varepsilon$ from~$\ercbox$ of the form:
\[
\varepsilon = 	N_1 |C_1|+ \ldots + N_k |C_k| + B \le N_{k+1}|C_{k+1}| + \ldots + N_{k+\ell}|C_{k+\ell}|
\]
with the corresponding inequality~$\varepsilon'$:
\[
\varepsilon = \Sigma_{i=1}^{k} N_k \big( \Sigma_{C \in \concepts, C \models C_i} |C| \big) + B \leq
\Sigma_{i=k+1}^{k+\ell} N_k \big( \Sigma_{C \in \concepts, C \models C_i} |C| \big). 
\]
Note that any model~$\I \models (\abox, \tbox, \ercbox)$ is also a model of~$(\abox, \tbox, \ercbox')$ and vice versa.

Let~$\vec{x_{\I}}$ be the solution to~$\ercbox'$ describing the types' cardinalities in~$\I$ (such solution exists since~$\I \models \ercbox'$).
As we have already mentioned before, the inequalities from~$\ercbox$ have the convenient property that
if a vector~$\vec{x}$ is a solution to~$\ercbox'$, then also a vector~$c \times \vec{x}$, i.e., the vector obtained by 
multiplying each entry of~$\vec{x}$ by a constant~$c$, is a solution to~$\ercbox'$. 
Thus there is also a solution~$\vec{y}$ to~$\ercbox'$ in the shape~$\vec{y} = (1+|\DIloose{k}|) \cdot \vec{x_{\I}}$. 

The desired set~$S \subseteq \Nat \times \DIloose{k}$ is defined as follows. It is composed of all 
pairs~$(c - |t^{\Iloose{k}}| , w_t)$ for each type~$t \in \types$ having a non-zero entry~$c$ in~$\vec{y}$ (where~$w_t$ is an arbitrary fixed domain element from~$\Iloose{k}$ having a type~$t$).  Note that such an element~$w_t$ exists since the~$k$--loosening and forward-unravelings preserve types (see e.g proofs of Lemma~\ref{lem:loosemodelhood} and Lemma~\ref{lem:unravtbox}). 

It remains to argue that~$\Iloose{k}_{{+}S} \models (\abox, \tbox, \ercbox)$ holds. 
To see that~$\Iloose{k}_{{+}S} \models \ercbox$ it is enough to see 
that~$\Iloose{k}_{{+}S} \models \ercbox$ holds due to the fact that the vector describing the types' 
cardinalities in~$\Iloose{k}_{{+}S}$ is equal to~$\vec{y}$ (and~$\vec{y}$ was obtained by multiplying 
each entry of the initial solution~$\vec{x_{\I}}$). 
Moreover we conclude~$\Iloose{k}_{{+}S} \models (\abox, \tbox)$ holds from Lemma~\ref{lem:repairmodelhood}. 
Hence~$\Iloose{k}_{{+}S} \models (\abox, \tbox, \ercbox)$.
\end{proof}

This concludes our construction, the core result of which can be informally stated as follows: \textsl{For any~$\ALCSCC$ knowledge base~$\kb$ and every CQ~$q$ holds: if~$\kb \models q$ then there is a forest-shaped query match of~$q$ into every model of~$\kb$.} This follows from the fact that the any model of~$\kb$ not admitting such a match would allow us to construct a model without any query matches, contradicting the assumption.  
We make this statement more formal by introducing the forthcoming notion of~$n$--acyclic models.

\subsubsection{The notion of~$n$--acyclic models}

Given a finite interpretation~$\J$ we say that it is \emph{$k$--acyclic}, 
if there exists a finite interpretation~$\I$ such that~$\J = \Iloose{k}_{{+}S}$ holds
for some finite set~$S \subseteq (\DI \times \Nat_+)$.

The next lemma states that to falsify conjunctive query we do not need to look for arbitrary finite counter-models 
but it is enough to consider the class of~$(|q|+1)$--acyclic models. Indeed:

\begin{lemma} \label{lem:nacyccountermodels}
For any normalized~$\ALCSCC$ knowledge base~$\kb = (\abox, \tbox, \ercbox)$ 
and any conjunctive query~$q$, if there is a finite interpretation such 
that~$\I \models \kb$ but~$\I \not\models q$, then there is a~$(|q|+1)$--acyclic 
model~$\I'$ such that~$\I' \models \kb$ and~$\I' \not\models q$.
\end{lemma}
\begin{proof}
It is enough to take~$\J = \Iloose{(|q|+1)}_{{+}S}$ for~$S$ given in~\ref{lem:repairingERC}. The modelhood preservation 
follows from Lemma~\ref{lem:loosemodelhood} and Lemma~\ref{lem:repairmodelhood}. 
Query non entailment is due to Lemma~\ref{lem:loosenonentail} and Lemma~\ref{lem:eqqueriesunravloose}. 
\end{proof}

Moreover conjunctive query entailment over~$(|q|+1)$--acyclic models is equivalent to entailment over their forward-unravelings. 
This fact follows directly from Lemma~\ref{lem:duplTheSameMatches}.

\begin{lemma} \label{cor:2}
For any interpretation~$\I$ being a~$(|q|+1)$--acyclic model of an~$\ALCSCC$ knowledge base~$\kb = (\abox, \tbox, \ercbox)$ 
composed of a normalized ABox~$\abox$, TBox~$\tbox$ and ERCBox~$\ercbox$ the equivalence~$\I \models q \Leftrightarrow \Iunrav \models q$ holds. 
\end{lemma}

Due to Lemma~\ref{cor:2} we can restrict our attention to query matches over the unfolding of~$(|q|+1)$--acyclic models only. 
it allow us to use a machinery of spoilers, splittings and fork rewrittings from~\cite{Lutz08}, 
developed for deciding unrestricted CQ entailment, to the case of finite query entailment 
with only some minor modifications.

\subsection{Deciding query entailment in exponential time} \label{sec:queryentailmentalgo}

Now we are ready to employ the announced exponential time method for deciding conjunctive query entailment from~\cite{Lutz08}. 
For a given~$\kb = (\abox, \tbox, \ercbox)$ and a query~$q$, we enumerate a set of~$\ALCHcap$ knowledge 
bases~$\kb_s = (\abox', \tbox')$ called \emph{spoilers} and check whether~$\kb \cup \kb_s$ is consistent.
Spoilers are modeled to prevent forest-shaped query matches. They are constructed by, on the one hand, rolling-up tree-shaped partial query matches into concepts and forbidding existence of such concept in a 
model and, on the other hand, forbidding certain behaviour of the Abox part of a model. Lutz \cite{Lutz08} shows that one can restrict ones attention to exponentially many spoilers and that the size of each such spoiler is only polynomial in~$|\kb|$ and~$|q|$. 
The algorithm for CQ entailment is then obtained by simply replacing Lutz's satisfiability algorithm for~$\ALCHcap$ knowledge bases\footnote{Note that $\ALCHcap$ is a sub-logic of $\ALCSCC$.} by our finite satisfiability algorithm for~$\ALCSCC$ knowledge bases from the previous sections. We derive correctness of the procedure as follows:~$\kb \cup \kb_s$ is satisfiable for some spoiler~$\kb_s$ exactly if there is a model of~$\kb$ without forest-shaped matches of $q$ and hence -- thanks to our above argument -- there is a model without any match of $q$.\\

Let~$q$ be a conjunctive query and let~$\Var(q)$ be the set of variables appearing in~$q$.
Through this Section we always assume that~$q$ contains only atomic concepts and no answer variables.
Note that~$q$ can be seen as a directed graph~$G_q = (V_q, E_q)$, where vertices from~$V_q$
are simply variables from~$\Var(q)$ and for any two nodes~$x,y$ there exists an edge~$(x,y) \in E_q$
between them if and only ifs~$r(x,y) \in q$ for some~$r \in N_r$. We say that~$q$ is \emph{tree-shaped} if~$G_q$ is a directed tree.

We start by introducing a notion of \emph{forks} and \emph{splittings} from~\cite{Lutz08}. 
\paragraph*{Forks.}
For a conjunctive query~$q$ we say that a conjunctive query~$q'$ 
\emph{is obtained from~$q$ by fork elimination}, if~$q'$ is obtained
from~$q$ by selecting two atoms~$r(y,x)$ and~$s(x,z)$ and identifying
variables~$y$ and~$z$. A query~$\forkrev{q}$ is a \emph{fork rewriting} of~$q$
if~$\forkrev{q}$ is obtained from~$q$ by applying fork elimination (possibly multiple times).
A \emph{maximal fork rewriting} fork rewriting of~$q$ is a query~$\maxfr{q}$ obtained 
by exhaustively application of fork elimination. It is known from~\cite{Lutz08} that 
maximal fork rewriting is unique (up to variable renaming), 
thus we speak about \emph{the} maximal fork rewriting.
\begin{figure}[!h]
\centering

\tikzset{every picture/.style={line width=0.75pt}} %set default line width to 0.75pt        

\begin{tikzpicture}[x=0.75pt,y=0.75pt,yscale=-1,xscale=1, scale=1.5]

%Straight Lines [id:da08955436624594393] 
\draw [color={rgb, 255:red, 65; green, 117; blue, 5 }  ,draw opacity=1 ]   (128.67,140) .. controls (126.31,139.99) and (125.14,138.8) .. (125.15,136.44) .. controls (125.16,134.09) and (123.99,132.9) .. (121.64,132.89) .. controls (119.28,132.88) and (118.11,131.69) .. (118.12,129.33) .. controls (118.14,126.97) and (116.97,125.78) .. (114.61,125.77) .. controls (112.26,125.76) and (111.09,124.57) .. (111.1,122.22) -- (108.36,119.45) -- (102.74,113.76) ;
\draw [shift={(101.33,112.33)}, rotate = 405.35] [color={rgb, 255:red, 65; green, 117; blue, 5 }  ,draw opacity=1 ][line width=0.75]    (10.93,-3.29) .. controls (6.95,-1.4) and (3.31,-0.3) .. (0,0) .. controls (3.31,0.3) and (6.95,1.4) .. (10.93,3.29)   ;

%Straight Lines [id:da5396095700792809] 
\draw [color={rgb, 255:red, 208; green, 2; blue, 27 }  ,draw opacity=1 ]   (130,80.25) -- (100.92,109.09) ;
\draw [shift={(99.5,110.5)}, rotate = 315.24] [color={rgb, 255:red, 208; green, 2; blue, 27 }  ,draw opacity=1 ][line width=0.75]    (10.93,-3.29) .. controls (6.95,-1.4) and (3.31,-0.3) .. (0,0) .. controls (3.31,0.3) and (6.95,1.4) .. (10.93,3.29)   ;

%Straight Lines [id:da7170143856144333] 
\draw [color={rgb, 255:red, 208; green, 2; blue, 27 }  ,draw opacity=1 ]   (132.33,138.17) -- (157.79,111.2) ;
\draw [shift={(159.17,109.75)}, rotate = 493.36] [color={rgb, 255:red, 208; green, 2; blue, 27 }  ,draw opacity=1 ][line width=0.75]    (10.93,-3.29) .. controls (6.95,-1.4) and (3.31,-0.3) .. (0,0) .. controls (3.31,0.3) and (6.95,1.4) .. (10.93,3.29)   ;

%Shape: Circle [id:dp42010912939254386] 
\draw  [fill={rgb, 255:red, 0; green, 0; blue, 0 }  ,fill opacity=1 ] (128.17,82.08) .. controls (128.17,81.07) and (128.99,80.25) .. (130,80.25) .. controls (131.01,80.25) and (131.83,81.07) .. (131.83,82.08) .. controls (131.83,83.1) and (131.01,83.92) .. (130,83.92) .. controls (128.99,83.92) and (128.17,83.1) .. (128.17,82.08) -- cycle ;
%Shape: Circle [id:dp3335043229090753] 
\draw  [fill={rgb, 255:red, 0; green, 0; blue, 0 }  ,fill opacity=1 ] (157.33,109.75) .. controls (157.33,108.74) and (158.15,107.92) .. (159.17,107.92) .. controls (160.18,107.92) and (161,108.74) .. (161,109.75) .. controls (161,110.76) and (160.18,111.58) .. (159.17,111.58) .. controls (158.15,111.58) and (157.33,110.76) .. (157.33,109.75) -- cycle ;
%Shape: Circle [id:dp26530100517725685] 
\draw  [fill={rgb, 255:red, 0; green, 0; blue, 0 }  ,fill opacity=1 ] (128.67,140) .. controls (128.67,138.99) and (129.49,138.17) .. (130.5,138.17) .. controls (131.51,138.17) and (132.33,138.99) .. (132.33,140) .. controls (132.33,141.01) and (131.51,141.83) .. (130.5,141.83) .. controls (129.49,141.83) and (128.67,141.01) .. (128.67,140) -- cycle ;
%Shape: Circle [id:dp8188292502434251] 
\draw  [fill={rgb, 255:red, 0; green, 0; blue, 0 }  ,fill opacity=1 ] (99.5,110.5) .. controls (99.5,109.49) and (100.32,108.67) .. (101.33,108.67) .. controls (102.35,108.67) and (103.17,109.49) .. (103.17,110.5) .. controls (103.17,111.51) and (102.35,112.33) .. (101.33,112.33) .. controls (100.32,112.33) and (99.5,111.51) .. (99.5,110.5) -- cycle ;
%Shape: Circle [id:dp6865138070436834] 
\draw  [fill={rgb, 255:red, 0; green, 0; blue, 0 }  ,fill opacity=1 ] (258.68,109.63) .. controls (258.68,108.65) and (259.47,107.85) .. (260.46,107.85) .. controls (261.44,107.85) and (262.23,108.65) .. (262.23,109.63) .. controls (262.23,110.61) and (261.44,111.4) .. (260.46,111.4) .. controls (259.47,111.4) and (258.68,110.61) .. (258.68,109.63) -- cycle ;
%Shape: Circle [id:dp4163852559162984] 
\draw  [fill={rgb, 255:red, 0; green, 0; blue, 0 }  ,fill opacity=1 ] (307.59,110.48) .. controls (307.59,109.47) and (308.41,108.65) .. (309.42,108.65) .. controls (310.43,108.65) and (311.26,109.47) .. (311.26,110.48) .. controls (311.26,111.5) and (310.43,112.32) .. (309.42,112.32) .. controls (308.41,112.32) and (307.59,111.5) .. (307.59,110.48) -- cycle ;
%Straight Lines [id:da8292821928965393] 
\draw [color={rgb, 255:red, 208; green, 2; blue, 27 }  ,draw opacity=1 ]   (263.62,109.8) -- (304.39,110.02) ;
\draw [shift={(306.39,110.03)}, rotate = 180.3] [color={rgb, 255:red, 208; green, 2; blue, 27 }  ,draw opacity=1 ][line width=0.75]    (10.93,-3.29) .. controls (6.95,-1.4) and (3.31,-0.3) .. (0,0) .. controls (3.31,0.3) and (6.95,1.4) .. (10.93,3.29)   ;

%Straight Lines [id:da3537687575241444] 
\draw [color={rgb, 255:red, 139; green, 87; blue, 42 }  ,draw opacity=1 ] [dash pattern={on 4.5pt off 4.5pt}]  (257.02,109.8) -- (215.5,110.23) ;
\draw [shift={(213.5,110.25)}, rotate = 359.40999999999997] [color={rgb, 255:red, 139; green, 87; blue, 42 }  ,draw opacity=1 ][line width=0.75]    (10.93,-3.29) .. controls (6.95,-1.4) and (3.31,-0.3) .. (0,0) .. controls (3.31,0.3) and (6.95,1.4) .. (10.93,3.29)   ;

%Shape: Circle [id:dp7264748063119706] 
\draw  [fill={rgb, 255:red, 0; green, 0; blue, 0 }  ,fill opacity=1 ] (208.58,110.42) .. controls (208.58,109.41) and (209.4,108.59) .. (210.41,108.59) .. controls (211.43,108.59) and (212.25,109.41) .. (212.25,110.42) .. controls (212.25,111.43) and (211.43,112.25) .. (210.41,112.25) .. controls (209.4,112.25) and (208.58,111.43) .. (208.58,110.42) -- cycle ;
%Straight Lines [id:da13174087861229578] 
\draw [color={rgb, 255:red, 208; green, 2; blue, 27 }  ,draw opacity=1 ]   (132.39,83.58) -- (156.06,106.44) ;
\draw [shift={(157.5,107.83)}, rotate = 224] [color={rgb, 255:red, 208; green, 2; blue, 27 }  ,draw opacity=1 ][line width=0.75]    (10.93,-3.29) .. controls (6.95,-1.4) and (3.31,-0.3) .. (0,0) .. controls (3.31,0.3) and (6.95,1.4) .. (10.93,3.29)   ;

% Text Node
\draw (130,75) node [scale=0.7]  {$x$};
% Text Node
\draw (97,111) node [scale=0.7]  {$y$};
% Text Node
\draw (163.5,111.33) node [scale=0.7]  {$z$};
% Text Node
\draw (130,147.07) node [scale=0.7]  {$t$};
% Text Node
\draw (115.67,89) node [scale=0.7]  {$r$};
% Text Node
\draw (145.67,88) node [scale=0.7]  {$r$};
% Text Node
\draw (152.13,125.8) node [scale=0.7]  {$r$};
% Text Node
\draw (108,126.67) node [scale=0.7]  {$s$};
% Text Node
\draw (212.69,115.13) node [scale=0.7]  {$y$};
% Text Node
\draw (312.67,115.22) node [scale=0.7]  {$z$};
% Text Node
\draw (259,100) node [scale=0.7]  {$\ xt$};
% Text Node
\draw (241.91,114.58) node [scale=0.5]  {$r\ \cap \ t$};
% Text Node
\draw (286.13,113.24) node [scale=0.5]  {$r\ $};

\end{tikzpicture}

\caption{A query~$q = r(x,y) \wedge r(x,z) \wedge r(t,z) \wedge s(t,y)$ (left) and its fork-rewriting (right) obtained by identifying variables~$x$ and~$t$.}
\end{figure}
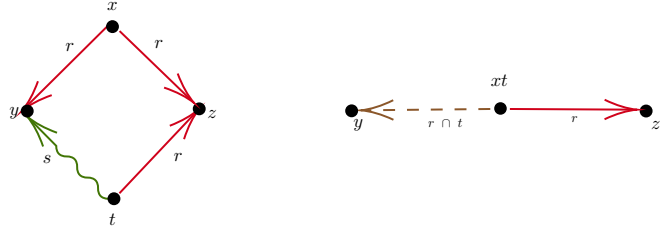

\paragraph*{Splittings.}
The next definition speaks about the abstract way how a conjunctive query
can match a model, without making reference to a concrete model nor a concrete match.

Let~$\kb = (\abox, \tbox, \ercbox)$ be a normalized~$\ALCSCC$ knowledge base composed 
of an Abox~$\abox$, Tbox~$\tbox$ and an ERCBox~$\ercbox$. 
A \emph{splitting} of a conjunctive query~$q$ w.r.t~$\kb$ is a tuple 
\[ \Pi = (R, T, S_1, S_2, \ldots, S_n, \mu, \nu), \]
where
the sets~$R, T, S_i$ induce a partition of the set~$\Var(q)$,
the function~$\mu : \{1,2,\ldots,n\} \rightarrow R$ assigns to each set~$S_i$ a variable~$\mu(i) \in R$, and
the function~$\nu : R \rightarrow \IndA$ assigns to each variable from~$R$ a named individual from~$\abox$.
A splitting~$\Pi$ has to satisfy the following 
conditions:\footnote{With~$\restr{q}{X}$ we denote the restriction of a query to the set of variables~$X$}
\begin{itemize}
	\item the query~$\restr{q}{T}$ is a variable disjoint union of tree-shaped queries,
	\item queries~$\restr{q}{S_i}$ for all~$i \in \{1,2,\ldots,n\}$ are tree-shaped,
	\item for any atom~$r(x,y) \in q$ the variables~$x,y$ either belong to the same set~$R, T, S_1, S_2, \ldots, S_n$
	or~$x \in R, y \in S_i$ with~$x$ being the root of a tree~$\restr{q}{S_i}$, and
	\item for any~$i \in \{ 1, 2, \ldots, n \}$ there is an atom~$r(\mu(i), x_0) \in q$ with~$x_0$ the root of~$\restr{q}{S_i}$.
\end{itemize}

% The presented definition may seem to be cumbersome, hence we would like to give some intuitions behind it. 
It might be easier to think that a splitting~$\Pi$ actually consists of ``roots''~$R$ (corresponding to the Abox 
part of the model) named by the function~$\nu$), together with their ``subtrees''~$S_i$ 
and of some arbitrary trees~$T$ somewhere far in a model.\\

% \bbe{TODO: Pic of an example splitting. Might be veeeeeeery helpful.}

\paragraph*{Rolling up concepts.}
We employ a known technique~\cite{CaDL98,Lutz08,GHLS07} of \emph{rolling-up 
a tree-shaped query into a concept}. For a given conjunctive query~$q$ we 
define an~$\ALCHcap$ concept~$C_{q,x}$ (for each variable~$x \in \Var(q)$) as follows.
If~$x$ is a leaf in~$G_q$ then 
$$C_{q,x} = \bigsqcap_{C(x) \in q} C.$$
Otherwise we set
$$
C_{q,x} = \bigsqcap_{C(x) \in q} C \sqcap \bigsqcap_{(x,y) \in E_q} \exists( \bigcap_{s(x,y) \in q} s).C_{q,y}.
$$

The forthcoming lemma links together all presented notions.

\begin{definition} \label{def:agrees}
Let~$q$ be a conjunctive query and let~$\kb = (\abox, \tbox, \ercbox)$ be a (consistent) normalized~$\ALCSCC$ knowledge base
with a model~$\I$. We say that a pair~$(\forkrev{q}, \Pi)$, composed of a fork rewriting~$\forkrev{q}$ of~$q$ 
and a splitting~$\Pi = (R,T,S_1, S_2, \ldots, S_n, \mu, \nu)$ w.r.t~$\kb$, \emph{is compatible with}~$\I$, if:  
\begin{itemize}
	\item for each disconnected component~$\widehat{q}$ of~$T$, there is an element~$d \in \DI$ with~$d \in (C_{\widehat{q}})^{\I}$,
	\item if~$C(x) \in \forkrev{q}$ with~$x \in R$, then~$\nu(x)^{\I} \in C^{\I}$,
	\item if~$r(x,y) \in \forkrev{q}$ with~$x,y \in R$, then~$(\nu(x)^{\I}, \nu(y)^{\I}) \in r^{\I}$, and
	\item for all~$1 \leq i \leq n$ we have (for~$x_0$ being the root of~$\restr{\forkrev{q}}{S_i}$):
	\[
		\nu(\mu(i))^{\I} \in 
		\left( \exists \left( \bigcap_{s(\mu(i), x_0) \in \forkrev{q}} s \right).C_{\restr{\forkrev{q}}{S_i}, x_0} \right)^{\I}
	\]
\end{itemize}
\end{definition}

\begin{lemma}
Take~$q$ and~$\kb$ as stated in Definition~\ref{def:agrees} and let~$\I$ be any~$(|q|+1)$--acyclic model of~$\kb$.
Then~$\I \models q$ if and only if there exists a pair~$(\forkrev{q}, \Pi)$ of a fork 
rewriting and splitting such that~$(\forkrev{q}, \Pi)$ is compatible with~$\I$.
\end{lemma}
\begin{proof}
Let~$\Iunrav$ be the forward-unraveling of~$\I$.
A similar lemma was proven in~\cite{Lutz08} and its proof without any changes at 
all can be seen as a proof that~$\Iunrav \models q$ iff~$\Iunrav$ is compatible with some~$(\forkrev{q}, \Pi)$.

Hence if~$\Iunrav$ is compatible with some~$(\forkrev{q}, \Pi)$ we can infer that~$\Iunrav \models q$ holds and 
by Corollary~\ref{cor:2} we conclude that~$\I \models q$. For the opposite way, assume that~$\Iunrav \models q$
holds. Thus~$\Iunrav$ is compatible with some~$(\forkrev{q}, \Pi)$. The construction of forward-unravelings is concept 
preserving (see e.g. the proof of Lemma~\ref{lem:unravtbox}), thus the first 
and the last item of Definition~\ref{def:agrees} are satisfied by~$\I$. To conclude the satisfaction of the second and the third items 
of Definition~\ref{def:agrees} it is enough to see that forward-unravelings preserve Aboxes (namely Lemma~\ref{lem:unravabox}).
Hence~$\I$ is compatible with~$(\forkrev{q}, \Pi)$.
\end{proof}

\paragraph*{Spoilers and super-spoilers.}
Let~$\kb = (\abox, \tbox, \ercbox)$ be normalized~$\ALCSCC$ knowledge base,
let~$q$ be a conjunctive query and let~$\Pi = (R,T,S_1, S_2, \ldots, S_n, \mu, \nu)$ be a splitting of~$q$ w.r.t~$\kb$. 
Moreover, let~$q_1, \ldots, q_n$ be the tree-shaped disconnected components of~$\restr{q}{T}$ with roots~$x_1, \ldots, x_n$.  

We say that~$\ALCHcap$ knowledge base~$\kb_s = (\abox_s, \tbox_s)$ 
is a \emph{spoiler} for~$q$,~$\kb$ and~$\Pi$ if one of the following conditions hold:
\begin{itemize}
	\item $\big(\top \sqsubseteq \neg C_{q_i, x_i} \big) \in \tbox_s$, for some~$1 \leq i \leq k$,
	\item there is an atom~$C(x) \in q$ with~$q \in R$ but~$\neg C(\nu(x)) \in \abox_s$
	\item there is an atom~$r(x,y) \in q$ with~$x,y \in R$ but~$\neg r(\nu(x), \nu(y)) \in \abox_s$
	\item $\neg D(\nu(\mu(i))) \in \abox_s$ for some~$1 \leq i \leq n$, where (for~$x_0$ being the root of~$\restr{q}{S_i}$):
	\[
		D = \left( \exists \left( \bigcap_{(\mu(i), x_0) \in q} s \right).C_{\restr{q}{S_i}} \right)^{\I}
	\]
\end{itemize}
A \emph{super-spoiler} for~$q$ and~$\kb$ is a \emph{minimal}~$\ALCHcap$ 
knowledge base~$\kb_s = (\abox_s, \tbox_s)$ such that for any 
splitting~$\Pi$ of~$q$ w.r.t~$\kb$, the knowledge base~$\kb_s$ is a spoiler for~$q$,~$\kb$ and~$\Pi$.

The following lemma describes the purpose of spoilers:

\begin{lemma} \label{lem:superspoiler1}
Let~$\kb = (\abox, \tbox, \ercbox)$ be a normalized~$\ALCSCC$ knowledge base
and let~$q$ be a conjunctive query. The query~$\kb \not\models q$ if and
only if there exists a super-spoiler~$\kb_s = (\abox_s, \tbox_s)$ such that
the knowledge base~$(\abox \cup \abox_s, \tbox \cup \tbox_s, \ercbox)$ is consistent.
\end{lemma}

\begin{proof}
Note that a similar Lemma was proven in~\cite{Lutz08} for infinite tree-shaped models. Its proof can be read without any changes
as a proof of the following statement: for all unravelings~$\Iunrav$ the condition~$\Iunrav \not\models q$ holds 
iff~$(\abox \cup \abox_s, \tbox \cup \tbox_s, \ercbox)$ is consistent for some super-spoiler~$\kb_s = (\abox_s, \tbox_s)$. 

If~$\kb \not\models q$ then (from Lemma~\ref{lem:nacyccountermodels}) there exists a~$(|q|+1)$--acyclic 
counter-model~$\I$ for~$q$, i.e., a model~$\I$ satisfying~$\I \not\models q$. Then also~$\Iunrav \not\models q$ (follows from Corollary~\ref{cor:2}). 
From~\cite{Lutz08} we infer that there exists a super-spoiler~$\kb_s = (\abox_s, \tbox_s)$ for~$\Iunrav$. Since~$\Iunrav$ and~$\I$ satisfy 
the same~$\ALCSCC$ formulae, we conclude that~$(\abox \cup \abox_s, \tbox \cup \tbox_s, \ercbox)$ is consistent.

For the opposite way assume that there exists a super-spoiler~$\kb_s = (\abox_s, \tbox_s)$ such that~$\kb' = (\abox \cup \abox_s, \tbox \cup \tbox_s, \ercbox)$ is consistent.
Then there is a~$(|q|+1)$--acyclic model~$\I$ of~$\kb'$. Aiming for contradiction assume that~$\kb \models q$. Hence there is a query match in~$\I$ and from 
Corollary~\ref{cor:2} we also know that~$\Iunrav \models q$. But it contradicts the Lutz's Lemma~\cite{Lutz08} for infinite tree-shaped models.
Hence,~$\Iunrav \not\models q$. Thus~$\I \not\models q$ which clearly implies that~$\kb \not\models q$.
\end{proof}

The last ingredient for designing an exponential time algorithm for deciding query
entailment is to estimate the number of super-spoilers as well as their size.
By showing that one can restrict attention only to trees being subtrees of a 
maximal fork rewriting, Lutz~\cite{Lutz08} have shown that (independently 
of the underlying DL formalism) the following lemma holds:
\begin{lemma}[\cite{Lutz08}] \label{lem:superspoilerbounds}
Let~$\kb = (\abox, \tbox, \ercbox)$ be a normalized~$\ALCSCC$ knowledge base and let~$q$ be a conjunctive query.
Then the total number of super-spoilers for~$\kb$ and~$q$ is only exponential in~$(|q| + |\kb|)$ and the size of
each super-spoiler is only polynomial in~$(|q| + |\kb|)$. Moreover the set of super-spoilers can be enumerated in exponential time.
\end{lemma}
\begin{proof}
Immediate conclusion from Lemma~$4$, Lemma~$5$ and Lemma~$6$ from~\cite{Lutz08}.
\end{proof}

The algorithm for deciding conjunctive query entailment for~$\ALCSCC$ knowledge bases~$\kb = (\abox, \tbox)$ w.r.t Aboxes, Tboxes and ERCBoxes is quite simple. 
We enumerate all super-spoilers~$\kb_s = (\abox_s, \tbox_s)$ (from Lemma~\ref{lem:superspoilerbounds} we know that 
there are only exponentially many of them and the enumeration process can be done in exponential time) 
and run a satisfiability test for~$\kb' = (\abox \cup \abox_s, \tbox \cup \tbox_s, \ercbox)$
by employing an algorithm described in Theorem~\ref{exptime:thm}. Since the size of~$\kb_s$ is only polynomial in~$(|q| + |\kb|)$
then the size of~$\kb'$ is also only polynomial in~$(|q| + |\kb|)$. Hence the satisfiability check can be done in~$\ExpTime$ (by~Theorem~\ref{exptime:thm} again).
We return the answer that~$q$ is not entailed by~$\kb$ if~$\kb'$ is satisfiable for some 
super-spoiler and that the query is entailed otherwise. Correctness of the procedure is guaranteed by Lemma~\ref{lem:superspoiler1}.
Hence we obtain:
\begin{theorem}
Conjunctive query entailment from~$\ALCSCC$ ERCBoxes wrt.~$\ALCSCC$ ABoxes is~$\ExpTime$-complete.
\end{theorem}

Moreover, since~$\ALCHQ$ is a sublogic of~$\ALCSCC$ 
(in a sense that for every~$\ALCHQ$ concept we find an equisatisfiable~$\ALCSCC$ concept), 
as a corollary we obtain the first known exponential time algorithm for deciding finite 
query entailment over~$\ALCHQ$ knowledge bases.
\begin{corollary}
Conjunctive query entailment from~$\ALCHQ$ TBoxes wrt.~$\ALCHQ$ ABoxes is~$\ExpTime$-complete.
\end{corollary}

The~$\ExpTime$ lower bounds comes already from~$\ALC$ concept satisfiability w.r.t TBoxes.

\section{Conclusion}

We have introduced the DL~$\ALCplus$, which allows for mixing local and global cardinality constraints. 
Though being considerably more expressive than previously investigated DLs with cardinality constraints,
reasoning in~$\ALCplus$ has turned out to be not harder that reasoning in~$\ALC$ with very simple cardinality restrictions.
However, extending~$\ALCplus$ with inverse roles causes undecidability for the standard inference satisfiability, as does considering 
the non-standard inference of query entailment in~$\ALCplus$. We were able to show that decidability of query entailment can be regained
by considering restricted cardinality constraints (ERCBoxes) in the sub-logic $\ALCSCC$ of~$\ALCplus$. The~\ExpTime upper bound proved for this task
depends on the ExpTime upper bound for ABox consistency in $\ALCSCC$ w.r.t.\ ERCBoxes shown for the first time in the present paper.
 
Some of the results presented here have already been sketched in a paper at the DL workshop~\cite{BaBR19}. However, there
the positive result for query entailment was restricted to a setting without ABox since we did not yet have the result for ABox consistency, 
and only a~\TwoExpTime upper bound for the complexity was shown. In addition, the undecidability result for~$\ALCIplus$ is also not contained
in~\cite{BaBR19}.

Regarding future work, it would be interesting to investigate the impact that adding inverse roles has on reasoning 
in~$\ALCSCC$ w.r.t.\ different kinds of terminological boxes (TBox, ERCBox, ECBox), though this will probably be a very hard task. 
From an application point of view,
as a first step towards a more practical query answering algorithm, we intend to investigate the ABox consistency problem in~$\ALCSCC$ w.r.t.\ ERCBoxes.
Since type elimination algorithms are not only worst-case, but also best-case exponential, we will try to devise a tableau-based algorithm 
for this problem, which may use numerical algorithms and satisfiability checkers for QFBAPA as sub-procedures.

\section*{Acknowledgements}
Franz Baader was partially supported by the German Research Foundation (DFG) within the Research Unit 1513 Hybris and
grant 389792660 as part of TRR 248.
Bartosz Bednarczyk was supported by the European Research Council (ERC) through the Consolidator Grant~771779 (DeciGUT) and the Polish Ministry of Science and Higher Education program ``Diamentowy Grant'' no. DI2017 006447.
Sebastian Rudolph was supported by the European Research Council (ERC) through the Consolidator Grant~771779 (DeciGUT).

\bibliographystyle{plain}
\bibliography{bibliography}

\end{document}